\documentclass[sn-mathphys,Numbered]{sn-jnl}

\usepackage{graphicx}%
\usepackage{multirow}%
\usepackage{amsmath,amssymb,amsfonts}%
\usepackage{amsthm}%
\usepackage{mathrsfs}%
\usepackage[title]{appendix}%
\usepackage{soul,xcolor}%
\usepackage{textcomp}%
\usepackage{manyfoot}%
\usepackage{booktabs}%
\usepackage{algorithm}%
\usepackage{algorithmicx}%
\usepackage{algpseudocode}%
\usepackage{listings}%

\usepackage[normalem]{ulem}

\newcommand{\RNum}[1]{\uppercase\expandafter{\romannumeral #1\relax}}

\usepackage{tabularx}

\newtheorem{prop}{Proposition}
\newtheorem{corollary}{Corollary}

\theoremstyle{thmstyleone}%
\newtheorem{theorem}{Theorem}
\theoremstyle{thmstyletwo}%
\newtheorem{example}{Example}%

\theoremstyle{thmstylethree}%

\begin{document}

\title[Article Title]{Dislocations and Fibrations: The Topological Structure of Knotted Smectic Defects}

\author[1]{\fnm{Paul G.} \sur{Severino}}

\author[1,2]{\fnm{Randall D.} \sur{Kamien}}

\author*[3,4]{\fnm{Benjamin} \sur{Bode}}\email{benjamin.bode@upm.es}

\affil[1]{\orgdiv{Department of Physics and Astronomy}, \orgname{University of Pennsylvania}, \orgaddress{\street{209 South 33rd St.}, \city{Philadelphia}, \state{Pennsylvania}, \postcode{19104}, \country{USA}}}

\affil[2]{\orgdiv{Department of Mathematics}, \orgname{University of Pennsylvania}, \orgaddress{\street{209 South 33rd St.}, \city{Philadelphia}, \state{Pennsylvania}, \postcode{19104}, \country{USA}}}

\affil[3]{\orgdiv{Institute of Mathematical Sciences (ICMAT)}, \orgname{Spanish National Research Council (CSIC)}, \orgaddress{\street{Calle Nicolás Cabrera 13-15, Campus Cantoblanco UAM}, \postcode{28049} \city{Madrid}, \country{Spain}}}

\affil[4]{\orgdiv{
Departamento de Matem\'atica Aplicada a la Ingenier\'ia Industrial}, \orgname{ETSIDI, Universidad Polit\'ecnica de Madrid}, \orgaddress{\street{Rda. de Valencia 3, 28012}, \postcode{28012} \city{Madrid}, \country{Spain}}}

\setstcolor{red}

\abstract{In this work, we investigate the topological properties of knotted defects in smectic liquid crystals.  Our story begins with screw dislocations, whose radial surface structure can be smoothly accommodated on $S^3$ for fibred knots by using the corresponding knot fibration.  To understand how a smectic texture may take on a screw dislocation in the shape of a knot without a fibration, we study first knotted edge defects.  Unlike screw defects, knotted edge dislocations force singular points in the system for any non-trivial knot.  We provide a lower bound on the number of such point defects required for a given edge dislocation knot and draw an analogy between the point defect structure of knotted edge dislocations and that of focal conic domains.  By showing that edge dislocations, too, are sensitive to knot fibredness, we reinterpret the so-called Morse-Novikov points required for non-fibred screw dislocation knots as analogous smectic defects.  Our methods are then applied to $+1/2$ and negative-charge disclinations in the smectic phase, furthering the analogy between knotted smectic defects and focal conic domains and uncovering an intricate relationship between point and line defects in smectic liquid crystals.  The connection between smectic defects and knot theory not only unravels the uniquely topological knotting of smectic defects but also provides a mathematical and experimental playground for modern questions in knot and Morse-Novikov theory.}

\maketitle

\tableofcontents

\section{Introduction}\label{sec1}

Line defects in three-dimensional materials, owing to their topological origin, must extend to the boundary of the sample, end on another defect, or close on themselves.  As such, closed-loop defects are often low energy excitations in an otherwise defect-free system.  In crystalline phases, for example, shear can proliferate dislocation loops via the Frank-Read mechanism\cite{Frank_Read}.  Dislocation loops also play an important role in smectic liquid crystals, moderating the transition to the nematic phase in models for the smectic-A to nematic transition\cite{Helfrich_Dislocation_Loop_Melting,Nelson_Toner_Smectic_Melting,Smectic_Melting_Experiment_Moreau_2006}.  Furthermore, elliptical disclination loops are, quite literally, at the core of focal conic domains in smectic and cholesteric liquid crystals.  The topological requirement that a lone defect return to itself does not specify, however, the path in which it should do so.  In other words, topological defects can, in principle, form knots and links.  Recently it has been recognized that knotted defect lines -- for example in nematic liquid crystals\cite{Copar_Colloids_Knotted_Defects,Machon_Knots_Review}, superfluid defects\cite{BEC_Knots,Kauffman_Superfluids}, Skyrmion models\cite{Skyrme_Knots}, and optical\cite{Optical_Knots} or fluid\cite{Moffatt_Fluid_Knots,Irvine_Fluid_Knots} vortices -- display novel mechanical and topological properties with respect to their unknotted counterparts.  Knotted defects are particularly accessible in nematic liquid crystals, where knotted disclination lines can be experimentally generated and probed using colloidal dimers\cite{Copar_Colloids_Knotted_Defects,copar_stability_nematic_colloids} or nonorientable surfaces with homeotropic anchoring\cite{Machin_Nonorientable_Boundary_Conditions}.  Despite the prevalence of knots in other ordered media, knotted dislocations have not been observed in smectic liquid crystals or crystalline phases at large.  With this in mind, we investigate the ability for defects in smectic liquid crystals -- systems with one-dimensional periodicity -- to knot.

Not only are knotted dislocations of interest as a fundamental excitation available to smectics and crystals, but knotted smectic defects are also knots in the pure mathematical sense.  A common tool for describing the extra topological character of knotted defects is knot theory\cite{Rolfsen_Knots_And_Links}; for example, knotted nematic disclinations carry an extra $\mathbb{Z}_n$ topological invariant, where $n$ is the Alexander polynomial of the knot evaluated at -1\cite{Machon_Alexander_Knots_1,Machon_Alexander_Knots_2,Machon_Knots_Review}.  In mathematical knot theory two knots are considered equivalent (\emph{i.e.}, of the same knot type) if they are smoothly isotopic, that is, one can be smoothly deformed into the other without any cutting, gluing, or passing the knot through itself. On the other hand, nematic disclinations and fluid vortices can, in principle, smoothly change their knot types over time. 
 Topological defects in ordered media are classified using homotopy theory: this approach assigns -- according to the symmetries of the medium, which determine the ground state manifold (GSM) for the system -- line defects in three dimensions to elements of $\pi_1(\text{GSM})$.  Two line defects in classes $\alpha,\beta\in\pi_1(\text{GSM})$ are topologically forbidden from passing through each other if and only if the commutator $\alpha\beta\alpha^{-1}\beta^{-1}$ is nontrivial\cite{Mermin_RevMod_Defects}.  A knotted defect, necessarily comprised of a single defect type, can therefore continuously undo its crossings and untie itself.  While knotted defects and vortices may experience an energetic cost to rewire their crossings, topology does not stabilize the knots as required in knot theory.  However, defects in smectic liquid crystals and crystals, due to the breaking of translational symmetry, succumb to extra topological restrictions beyond homotopy theory\cite{poenaru,Mermin_RevMod_Defects,chen_goldstone}.  For example, edge and screw defects, despite belonging to the same element of the fundamental group, can be topologically linked\cite{Mosna_Kamien_Linked_Dislocations}.  Further, the motion of edge dislocations via glide requires local melting of the crystalline order for the dislocation to pass through atomic or smectic layers\cite{Hocking_Kamien_Peierls_Nabarro}.  These arguments provide a topological barrier for smooth defect crossing of even a single defect type, allowing knotted defects deep in the smectic phase to be particularly amenable to the tools of knot theory.  

Using tools from knot and Morse-Novikov theory, we study the topological properties of knotted defects in the smectic phase.  We begin with knotted screw dislocations, whose purely-radial, twisting defect structure is achieved smoothly on $S^3$ using the mathematics of knot fibrations.  However, not all knots can be fibred in $S^3$.  This has important implications for screw dislocations, as one could imagine forcing the local defect structure of a screw dislocation along a knot that is not fibred.  By definition, the smectic layers around such a dislocation cannot smoothly fill the rest of space.  What types of singular structures would emerge in a smectic texture containing a screw dislocation in the shape of a non-fibred knot, and how would they be classified as smectic defects?  To attack this question, we turn first to edge dislocations.  Surprisingly, we find that edge dislocations cannot be tied into \emph{any} non-trivial knot\footnote{By non-trivial, we mean knots which are not isotopic to a simple planar loop, \emph{i.e.}, the unknot.} without creating extra defects in the system.  By studying smectic configurations with knotted edge dislocations as foliations containing Seifert surfaces, we prove, under generic assumptions about smectic layers, that a knotted edge dislocation necessitates a minimum of $4g$ point defects, where $g$ is the minimum-possible genus of a Seifert surface for the given knot.  The topological origin of these extra point defects is straightforward: Morse critical points are required to transform smectic layers from flat sheets at the boundaries to Seifert surfaces of necessarily non-zero genus.  These Morse-type point defects are common in smectic liquid crystals, as they lie at the core of focal conic domains.  We thus draw an analogy between the pairs of Morse-index 1 and 2 point defects required to knot an edge dislocation and the point defect structure of focal conic domains.  While the lower bound of $4g$ point defects holds for any knot, edge dislocations, too, are sensitive to fibrations -- only when the knot is fibred can the minimum number of point defects be achieved.  We explicitly provide a construction for smectic configurations with knotted edge dislocations and their associated Morse-type point defects for all knots, fibred as well as non-fibred.  The requirement of point defects and sensitivity to fibredness for knotted edge dislocations can be used to reinterpret the singular structures that emerge in screw dislocations in the shape of non-fibred knots.  Namely, we identify the so-called Morse-Novikov singularities required for singular fibrations of non-fibred knots as the exact same focal conic-like point defects required for knotted edge dislocations. To further this analogy, we prove similar results about knotted smectic disclinations.  For example, knotted $+1/2$ disclinations -- a topological generalization of focal conic domains to knotted defect cores -- require $2g+1$ point defects, where $2g$ point defects are attributed to the defect knottedness and $1$ point defect is the original focal conic domain defect.  The topological properties of knotted smectic defects uncovered in this work are unique amongst knots in other ordered media and represent the remarkable complexity of the smectic phase.  Further, our results bring about a deep connection between smectic defects and knot theory that provides both motivation for open mathematical questions and a potential platform to directly study Morse-Novikov points and knot fibrations experimentally.

The paper is organized as follows.  In Section \ref{sec:background} we review the basics of smectic liquid crystals, smectic defects, and focal conic domains.  We also provide background on the relevant knot theory, namely fibred knots and open book decompositions.  In Section \ref{sec:screw} we describe the construction of screw dislocations using knot fibrations on $S^3$, discuss how knotted screw dislocations may be realized in $\mathbb{R}^3$, and analyze the screw-like geometry of knot fibrations.  With the goal of describing smectic configurations containing screw dislocations in the shape of non-fibred knots, we switch gears in Section \ref{sec:edge_dislocations} and study knotted edge dislocations.  To study smectic configurations with knotted edge dislocations, we first introduce Morse points and how smectic phase fields may be studied as Morse functions on $S^3$.  By constructing the local character of edge dislocations from a pair Seifert surfaces, we then prove that a minimum of $4g$ point defects are required to knot an edge dislocation (Theorem \ref{thmedge}).  In Section \ref{sec:fibrednessANDdislocations} we unravel the relationship between knotted smectic defects and fibredness by showing that edge dislocations can only achieve the lower bound of $4g$ point defects when the knot is fibred (Theorem \ref{thm:fib_edge}).  Here we give explicit constructions for smectic configurations containing a knotted edge dislocation and its associated Morse critical points.  The common language of Heegaard splittings and cobordisms allows us to interpret the Morse-Novikov point defects required for non-fibred screw dislocations as the same focal conic-like point defects.  In Section \ref{sec:+12} we provide similar constructions for knotted smectic disclinations, characterizing the number of point defects required to knot $+1/2$ and negatively-charged disclinations (Theorems \ref{thmdisc} and \ref{thmnegdisc}).  Determining the number of point defects required for each knotted defect requires formulating the question in purely mathematical terms.  As such, certain results are presented in the form of proofs.  Explanations and interpretations are provided for each.

 \section{Background}\label{sec:background}

\subsection{Smectic Dislocations and Disclinations}\label{subsec:smectic_defects}

Smectic liquid crystals are phases that not only inherit the orientational order of their rod-like constituents but also form evenly spaced layers of codimension one.  In three dimensions these are surfaces wherein the nematic director $\boldsymbol{n}$ is parallel to the layer normal $\boldsymbol{N}$.  The completely flat, evenly spaced layers of the ground state are described by a density wave,

\begin{align}
    \rho(\boldsymbol{x})=\rho_0+\delta\rho \cos(\Phi),~~~~\Phi=\boldsymbol{k}\cdot\boldsymbol{x}+\phi_0,
\end{align}

where $\phi_0$ is the phase at the origin and $\boldsymbol{k}$ describes the orientation of the layer normal with respect to a fixed axis.  The layers sit on density maxima $\Phi=2\pi\mathbb{Z}$, but the level sets $\Phi\in\mathbb{R}$ also define `virtual layers' -- \emph{i.e.} the rest of the density wave -- that foliate $\mathbb{R}^3$ with flat sheets.

\begin{figure}
    \centering
    \includegraphics[scale=.475]{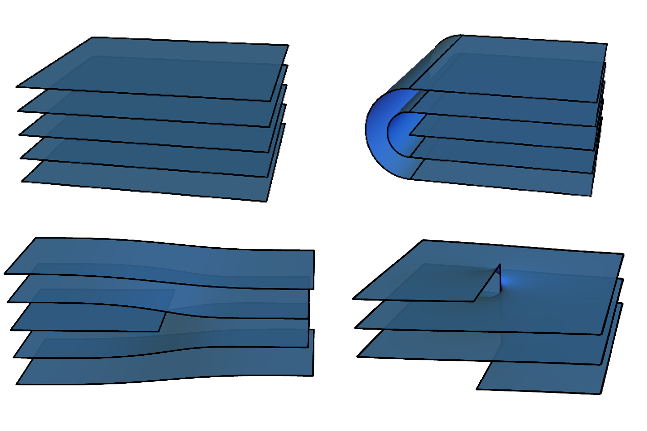}
    \caption{\label{Smectic_Examples}{Examples of smectic configurations.  Top left: smectic ground state of evenly spaced, flat layers.  Top right: $+1/2$ disclination line.  Bottom left: edge dislocation.  Bottom right: screw dislocation.  Only the density maxima (blue) of the smectic foliation are shown.}}
    \label{fig:Smectic_Examples}
\end{figure}

Three-dimensional smectics enjoy two types of topological line defects: disclinations and dislocations.  Disclinations are singular lines in the nematic director and layer normal.  The director field $\boldsymbol{n}$ is obtained from the phase field $\Phi$ away from defects via $\boldsymbol{n}=\boldsymbol{N}=p(\nabla\Phi)$, where $p:\mathbb{R}^3\backslash\{d\}\to\mathbb{RP}^2$ is the projection map $p(x,y,z)=[x:y:z]$ and $d$ is the defect set.  Disclinations in the smectic phase can be characterized by the number of $2\pi$ rotations of the director around the singular line.  For example, the disclination in Figure \ref{fig:Smectic_Examples} (top right) has `charge' $+1/2$.  Note that the definition of the director field in the smectic phase tells us immediately that $\boldsymbol{n}$ is only well-defined where the field $\Phi$ is well-defined and $\nabla\Phi$ is non-zero.  Singular lines in $\Phi$, dislocations, are thus singular in $\boldsymbol{n}$ as well.  In three dimensions dislocations come in two types, edge and screw, traditionally defined by the whether paths encircling the defect cause perpendicular (edge) or parallel (screw) displacement relative to the defect line\cite{Volterra1907,chaikin_and_lubensky}.  The smectic ground state, $+1/2$ disclination, edge dislocation, and screw dislocation are shown in Figure \ref{fig:Smectic_Examples}.

While, at the level of the phase and director fields, smectic disclinations and disclinations are singular lines around which $\Phi$ or $\nabla\Phi$ winds, not every phase field results in level sets that have the layered structure required of a smectic liquid crystal.  For example, disclinations in two-dimensional smectics cannot have charge $q>1$, as such winding in the director field is incompatible with smectic layers\cite{poenaru}.  Such considerations -- which are not captured by the traditional homotopy description of topological defects -- are relevant to dislocations as well.  Consider the simplest structure yielding a $2\pi$ winding in the phase field: a family of $2\pi$ surfaces meeting radially along the line, depicted on the left side of Figure \ref{fig:ScrewRadialSurfaces}.  Such a phase field, despite having the desired $2\pi$ winding of $\Phi$, does not constitute a smectic configuration, as it has no layered structure: while the layers meet along the defect by design, at infinity the surfaces splay apart.  One can attribute this configuration's shortcomings to the winding in $\nabla\Phi$, as this radial $+1$ disclination structure is both antithetical to smectic layering and persists to the boundary of the system.  By forcing this structure to instead have the appropriate boundary conditions -- namely that $\nabla\Phi$ approaches a constant at infinity -- in one of two different ways, one recovers edge and screw dislocations.  By reviewing this construction, we provide a definition for screw vs. edge dislocations and, in so doing, briefly recapitulate the results of \cite{Severino_Kamien}.

\begin{figure}
    \centering
    \includegraphics[scale=.85]{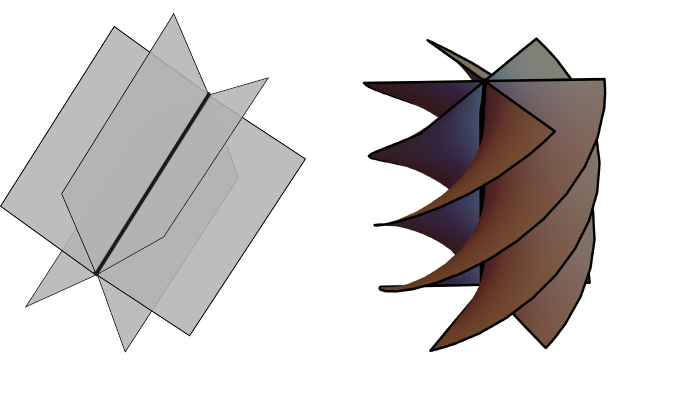}
    \caption{Example of a radial family of surfaces meeting along a line (left) and a screw dislocation (right).  In both cases a $2\pi$ family of surfaces are bounded by a singular line, leading to $2\pi$ winding in $\Phi$.  However, the screw dislocation's helicoidal geometry, characterized by the twisting of the surfaces around the dislocation, provides a natural layered structure far from the defect core.}
    \label{fig:ScrewRadialSurfaces}
\end{figure}

One way to modify the aforementioned radial family of surfaces into a smectic configuration is to simply add a second, hyperbolic $-1$ disclination parallel to the radial defect.  This cancels directly the $+1$ disclination associated with winding in $\Phi$, allowing boundary conditions ($\nabla\Phi=\hat{z}$) to be satisfied.  The presence of the hyperbolic defect restores the layered structure of the smectic, and, in the resultant smectic configuration, paths encircling the pair of disclinations cause the $2\pi$ winding in $\Phi$ and perpendicular displacement expected for an edge dislocation.  The charge-neutral disclination dipole at the core of an edge dislocation is captured by the standard phase field $\Phi_{\text{edge}}=z-\arctan\left(z/x\right)$.  On the other hand, a $2\pi$ winding in a \emph{nematic} director field is topologically trivial, owing to $\pi_1(\mathbb{R}P^2)=\mathbb{Z}_2$\cite{escapeinto,Mermin_RevMod_Defects}.  As such, one could, rather than cancel the disclination charge directly, attempt to relax the director field to $\nabla\Phi=\hat{z}$ at the boundary in a manner consistent with smectic order.  Smectic layers of negative Gaussian curvature allow the locally-radial $+1$ winding of $\nabla\Phi$ to be relaxed to constant $\nabla\Phi$ at infinity\cite{Severino_Kamien}.  Endowing the radial family of surfaces with requisite Gaussian curvature by twisting each surface along the defect line, shown in Figure \ref{fig:ScrewRadialSurfaces}, leads precisely to the helicoidal level sets of $\Phi_{\text{screw}}=z-\arctan(y/x)$ and the vertical displacement of the screw dislocation.  Note that the local defect structure of edge and screw dislocations are distinct -- the core of an edge dislocation contains both a radial and hyperbolic disclination in parallel whereas the screw dislocation maintains a purely-radial disclination core.  In the case of straight-line dislocations, this topological difference matches the different geometric properties of edge and screw dislocations\cite{Volterra1907,Severino_Kamien}.  For knotted dislocations it will be essential to rely on the topological definition of screw and edge defects: screw dislocations are defined as defects with purely-radial surface structure along the defect line whereas edge dislocations are defined by their disclination-charge-neutral defect core.  The difference in core topology between edge and screw dislocations leads to dramatically different global structure when the defects are knotted.

\begin{figure}
    \centering
    \includegraphics[scale=.53]{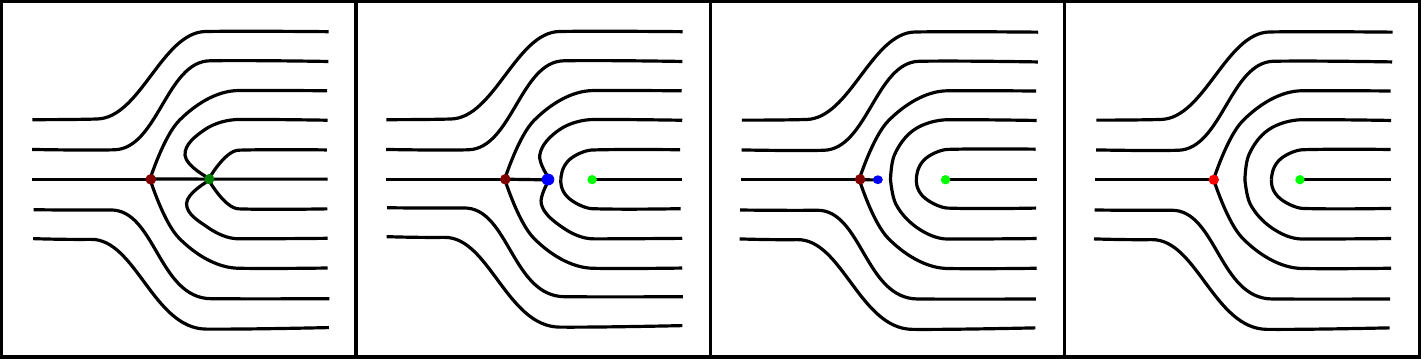}
    \caption{Transforming, with only a local operation, an edge dislocation with a $\pm1$ disclination core (left) to an edge dislocation with a pair of $\pm1/2$ disclinations (right).  This process involves extending the extending the melted region (blue) around the $+1$ disclination (dark green) along the layer ($\Phi_0$) connecting the disclinations.  As the melted region is moved toward the $-1$ disclination (dark red), the surfaces surrounding the the $+1$ disclination can be glued together, turning the defect into $-1/2$ disclination (light green).  Once all surfaces entering the blue defect have been glued, the surface $\Phi_0$ can be completely removed, transforming the $-1$ disclination into a $-1/2$ disclination (light red). The operation to the surface geometry occurs within a neighborhood of the disclination pair; far from the dislocation core, the surfaces remain unchanged.}
    \label{fig:Edge_Open_Book}
\end{figure}

\

An important energetic and topological consideration special to smectics is that the up-down symmetry of their constituents leads to the existence of half-charged disclinations.  In the construction of edge and screw dislocations above, not only have the smectic layers been orientable, but the foliation of surfaces is co-orientable -- also known as transverly oriented\cite{FoliationsCamacho,FoliationsGilbert}.  However, around half-charged disclinations, the nematic director winds by an odd multiple of $\pi$.  For example, in the straight-line $+1/2$ disclination of Figure \ref{fig:Smectic_Examples}, all smectic layers are orientable but closed paths around the disclination cause the surface normal and nematic director to flip from $\boldsymbol{n}$ to $-\boldsymbol{n}$ when returning to the same location on a surface, causing the foliation at large to be non-co-orientable.  In real systems, the evenly-spaced nature of $+1/2$ disclinations energetically favors dislocation cores of half-charged disclinations: the edge dislocation with a $\pm1/2$ disclination pair and the screw dislocation with a pair of $+1/2$ disclinations.  While much of the following mathematical analysis makes use of oriented surfaces and co-orientable foliations, a dislocation core with non-co-orientable defects can always be inserted with only a local surgery to smectic layers around a co-oriented dislocation core.  This is well studied for screw dislocations where, in a Burgers-displacement sized neighborhood of the defect, a helically winding pair of $+1/2$ disclinations replaces the radial defect at the helicoidal core\cite{screwone,screwtwo,Disclination_Pairs}.  Such a process is outlined for edge dislocations in Figure \ref{fig:Edge_Open_Book}.  Here the melted core of the $+1$ disclination is extended along the surface, defined by the level set $\Phi_0=0$, connecting the $+1$ and $-1$ disclinations.  As the melted region is extended, the surfaces $\Phi=\pm \epsilon$ can be glued together. While this introduces a non-co-orientability to the foliation of surfaces, the up-down symmetry of the director field allows for these surfaces to be glued without mismatch of the surface normals.  Continuing this process until the surfaces $\Phi=\delta$ and $\Phi=2\pi-\delta$ are glued restores the smectic order locally, leaving behind a $\pm1/2$ disclination pair.  The difference between edge and screw dislocations can be stated in terms of the topologically distinct pair of half-charged disclinations at their core: an edge dislocation is comprised of a $\pm1/2$ disclination pair in parallel whereas a screw dislocation is comprised only of a pair of $+1/2$ disclinations, maintaining $+1$ disclination charge.  These constructions, which will be generalized to knotted defects, enable us to study the topological properties of dislocations comprised of co-orientable (integer-charged disclinations) and non-co-orientable (half-integer-charged disclinations) defect cores interchangeably.

Now that the basics of screw and edge dislocations have been introduced, one may ask what shapes these line defects can enjoy in three dimensions.  Since the topological character of a dislocation line, characterized by the surfaces defined by the level set of the smectic phase field, is invariant under diffeomorphisms of space, a line defect may take a distorted path through the sample.  But, without the presence of other defects, the only other option is that the dislocation returns to itself.  The simplest case is an unknotted edge dislocation loop, which can be constructed from a pair of charge-neutral dislocations in two dimensions.  Rotating such a configuration around the vertical axis results in an axial configuration defined by

\begin{align}
    \Phi=z-\arctan\left(\frac{z}{d/2-\sqrt{x^2+y^2}}\right),\label{eqn_edge}
\end{align}

where the initial point defects were separated by a distance $d$.  As in the straight-line edge dislocation, the same topological considerations apply to this edge dislocation loop: it is comprised of a $\pm1$ pair of disclinations (which in turn may split into $\pm1/2$ disclinations with only a local change to the surfaces) that sit parallel to each other.  However, when considering the possibility of knotted dislocations, it becomes apparent that the extra rotational symmetry of an unknotted edge dislocation greatly simplifies its description.\footnote{For example, to turn a pair of points in two dimensions into a knot or link requires rotating the configuration out of the plane by at least $2\pi$, since, by the F\'ary-Milnor theorem, knots must have total curvature $\kappa\ge 4\pi$\cite{RevModPhys_Geometry}.  But applying this extra rotation to two point dislocations causes overlaps in the surfaces as they are swept out in three dimensions, spoiling the description.}  Extending the construction of dislocations to include knots requires a different approach, and is better formulated utilizing the mathematics of knot theory.

\subsection{Focal Conic Domains}\label{sub:FCD}

Focal conic domains\cite{FocalConicsKlemanLavrentovich} are, perhaps, the most distinctive and important texture in smectic liquid crystals.  In a focal conic domain, the smectic layers take on the geometry of the Dupin cyclides\cite{Dupin,GrandjeanFriedel}.  Dupin cyclides are a special class of canal surfaces, defined as the surfaces swept out by a sphere of varying radius as its center moves along some path.  The centers of the spheres that sweep out the Dupin cyclides trace out a perpendicular ellipse and hyperbola that pass through each other's foci.  This particular choice of ellipse and hyperbola, allowing the Dupin cyclides to be envelopes of spheres on two sides, endows the surfaces with geometric properties desirable of layered systems.  Namely, smectic layers in the shape of Dupin cyclides have perfect layer spacing, at the cost of having layers with non-zero mean curvature and singular lines in the director field.  Focal conic domains, which often form with only the negative Gaussian curvature sections of the Dupin cyclides, serve a variety of roles in smectic liquid crystals, from mediating boundary conditions to forming tilt grain boundaries.  The characteristic confocal ellipse and hyperbola, seen in the microscope due to the rapid change in the dialectric tensor, make focal conic domains readily identifiable in experiment.

\begin{figure}
    \centering
    \includegraphics[height=5cm]{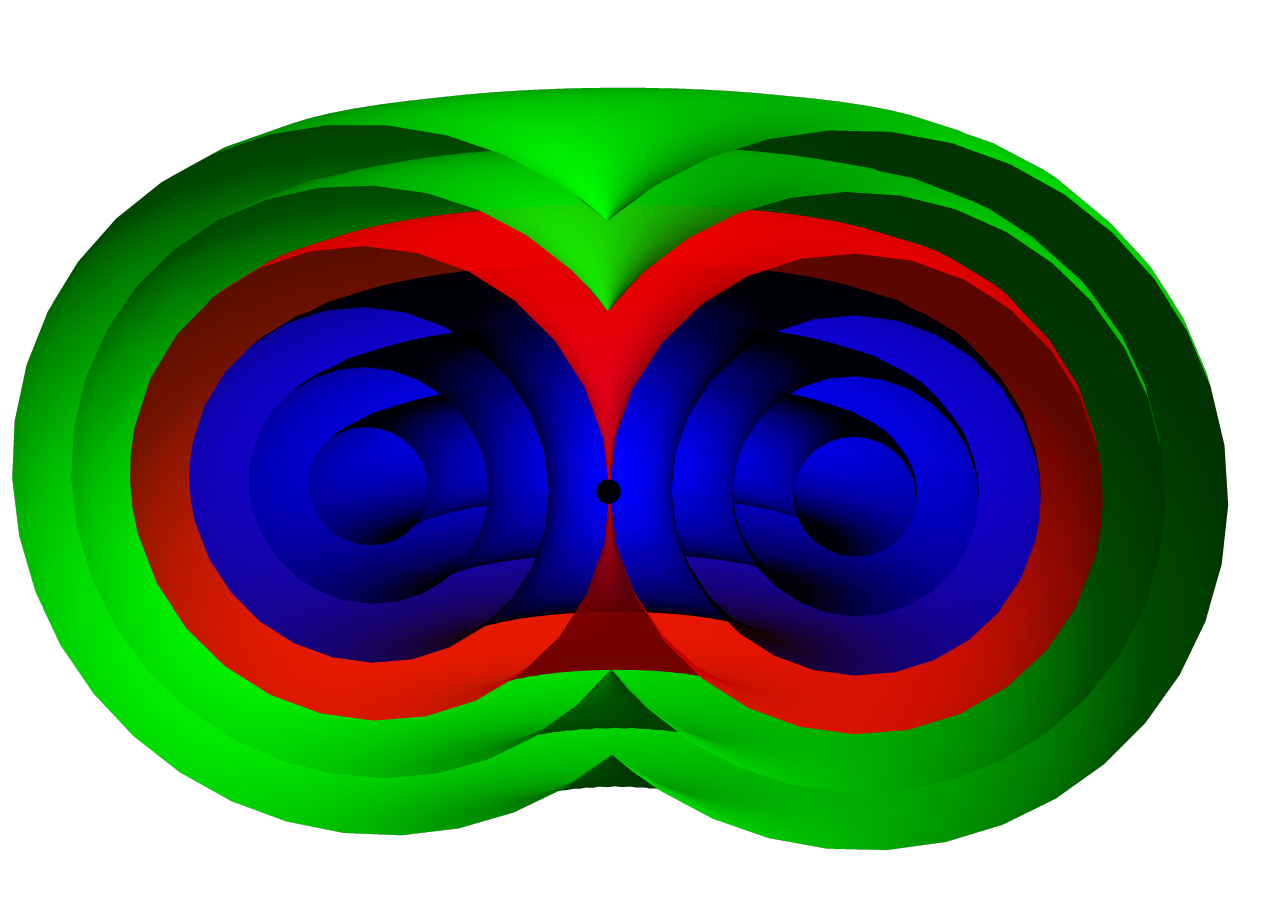}
    \includegraphics[height=5cm]{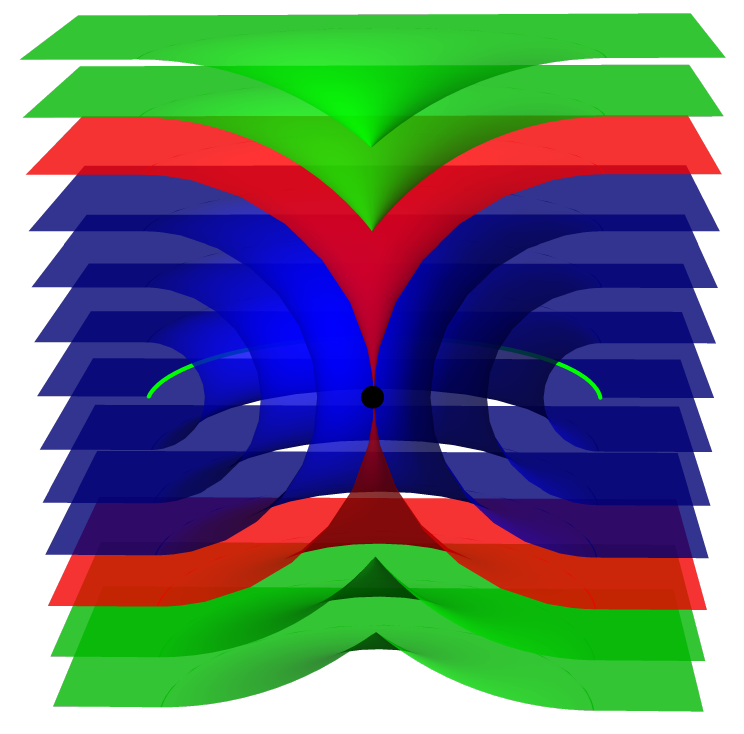}
    \caption{Left: Cutaway view of the symmetric Dupin cyclides.  The blue surfaces are concentric tori, growing until the red surface (horned torus).  This surface has a singular point (black), after which the surfaces become spheres (green) each containing two non-topological curvature singularities.  Right: Negative Gaussian curvature section of the symmetric Dupin cyclides connected to flat layers, yielding a smectic configuration satisfying generic boundary conditions -- a toric focal conic domain.  The $+1/2$ disclination loop is shown as the bright green line, and the singular point of the horned torus is the central point defect.}
    \label{fig:cyclides}
\end{figure}

Focal conic domains are non-topological structures:  for one, the ubiquity of focal conic domains is served by the fact that they can be inserted into the topologically-trivial smectic ground state or distorted smectic layers\cite{SethnaKleman}.  After all, the equal-spacing condition leading to the geometry of the Dupin cyclides is an energetic, not topological condition.  Smectic configurations will allow, to a measurable effect, some dilation to compensate for curvature energy in accordance with the non-linear smectic energy functional\cite{Lacaze_Kamien_OilyStreaks}.  Nonetheless, focal conic domains have interesting topological structure as it relates to point and line defects.  Consider first the geometry of the symmetric Dupin cyclides, or toric focal conic domains, shown in Figure \ref{fig:cyclides}.  In this configuration, the smectic layers start as concentric tori surrounding a $+1$ disclination.  Eventually, the evenly-spaced tori reach a point where one layer -- the horned torus -- pinches off at the center.  After this point defect is borne, the smectic layers are topological spheres with two curvature singularities.  The line of curvature singularities is a non-topological defect, as the surfaces, at the cost of even spacing, can smoothly remove the discontinuity.  The $+1$ disclination loop and central point defect are, on the other hand, topological defects.  In toric focal conic domains, the line and point defects work in tandem: the $+1$ disclination forces the smectic layers to be tori, and the point defect decreases the genus of the layers so that, in accordance with the geometry of the cyclides, the layers can become topological spheres.

\begin{figure}
    \centering
    \includegraphics[height=10cm]{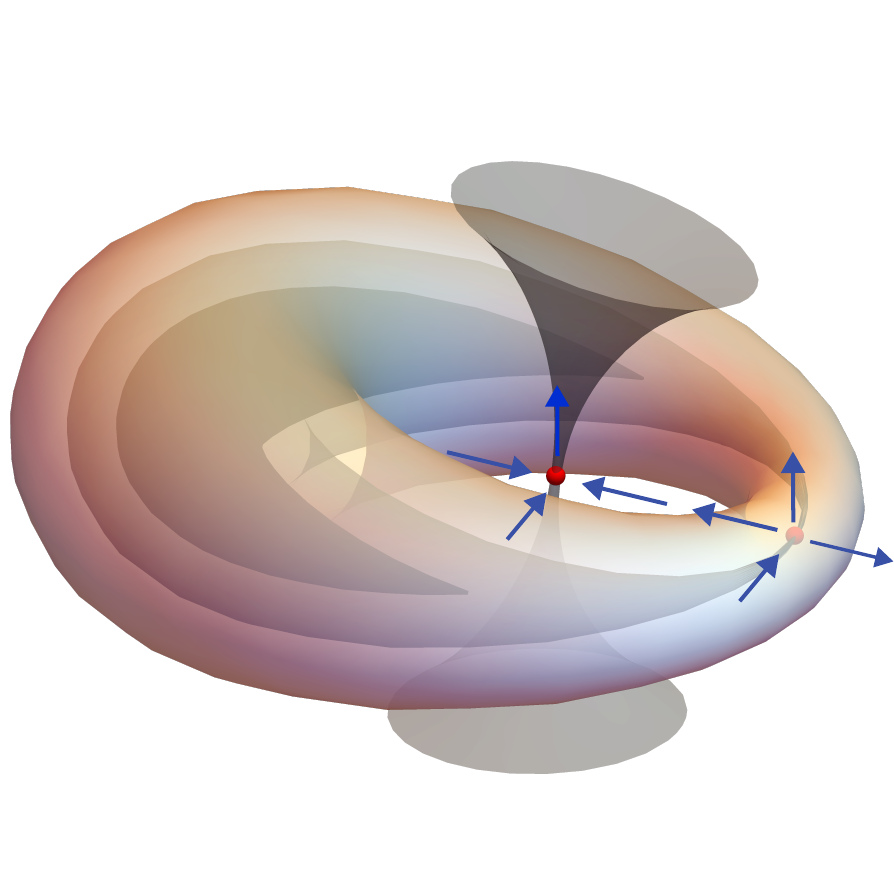}
    \caption{Point defect structure of general Dupin cyclides.  The inner surfaces are `croissant'-like topological spheres, passing through the right-most critical point to become toroidal.  The surfaces transform back to spheres following the central point defect.  Arrows indicate the orientation of the layer normal and nematic director, showing the opposite nematic point charge of the two point defects.}
    \label{fig:tilted_cyclides}
\end{figure}

By looking at the more general, tilted Dupin cyclides, the topological structure of focal conic domains reveals itself.  In this configuration, the disclination line disappears and instead the ellipse is, like the hyperbola, a non-topological singular line in curvature.  Instead, the general Dupin cyclides contain only topological point defects.  In addition to the point defect at the focus of the ellipse, a new point defect has opened up on the remains of the disclination.  Across this singular point, ``croissant-shaped" smectic layers meet at a point, after which the layers become topological tori.  As defects in the nematic director field\cite{Kamien_RevModPhys_Defects}, these two point defects have opposite charge and could be annihilated.  Note that there is also a third point defect, locally of the form of concentric spheres, from which the ``croissant layers" grow.  This point charge remains far away from the focal conic domain core\footnote{Taking only the negative Gaussian curvature section removes this global charge, leading to flat smectic layers far from the focal conic domain.}.  Further, the two conical point defects of the Dupin cyclides have opposite charge as defects in the smectic phase.  The topology of individual smectic point defects can be classified according to Morse theory\cite{aspects_topology_smectics}.  While the relationship between Morse functions and phase fields with respect to (knotted) line defects will be formally explored in Section \ref{sub:MorseAndProjection}, it is sufficient here to note that Morse critical points in three dimensions play an important role: they have the ability to increase or decrease the genus of the level sets of the phase field.  This is the case for the ``croissant" and ``horned torus" point defects in the Dupin cyclides that, respectively, increase and decrease the genus of smectic layers.  One sees that, because the layer structure of Dupin cyclides forces smectic layers that are topological spheres, then tori, then spheres again, focal conic domains require two Morse-type point defects of neighboring Morse index (index 1 and 2).  The opposite charge of these defects leads to a `topologically trivial' texture far from the core.  However, the presence of these point defects, forced upon the smectic by the Dupin cyclides, reflects a topological structure at the core of focal conic domains.  Our analysis shows that this very same topological structure is required to knot smectic defects.  Specifically, the topology of knotted defects forces the smectic configuration to contain surfaces of varying genus, thereby requiring Morse-type point defects of exactly the same form -- \emph{i.e.}, index 1 and 2 critical points -- as a focal conic domain.

\subsection{Seifert Surfaces and Knot Fibrations}

In order to build knotted smectic defects and understand their topological properties, one requires ingredients from knot theory.  The basic building blocks for knotted smectic defects are Seifert surfaces.  A Seifert surface is a compact, oriented surface whose boundary is a knot or link\footnote{Links are more general as they refer to a collection of knots.  Our results apply to both knots and links, but we will remark upon differences between the two as needed.}.  Seifert surfaces exist for all knots in $\mathbb{R}^3$, and the topological character of such surfaces is often non-trivial.  Namely, a Seifert surface can have genus zero only if it bounds the unknot: if a genus-zero Seifert surface existed for a given knot, then the surface could, by definition, be smoothly deformed into a planar disk, which bounds the unknot. 
 It follows that Seifert surfaces of non-trivial knots must have non-zero genus, and the minimum possible genus, $g$, for a Seifert surface of a given knot is a knot invariant.  For example, the minimum possible genus of a Seifert surface for the trefoil knot is $g=1$.  

As will be the case for screw dislocations, one may desire not just one Seifert surface but a family of Seifert surface of the same knot.  In particular, screw dislocations call for a helicoid-like structure along a knot -- $2\pi$ family of Seifert surfaces $F_\theta$ that meet radially along the knot.  Such structure is provided by a knot fibration, formally defined as a map $f:S^3\backslash K\to S^1$ from the complement of the knot to the circle such that each of its level sets $f^{-1}(\theta)$, $\theta\in S^1$, yields the interior of a Seifert surface $F_\theta$.  The union of the fibres $F_\theta$ fills the space around the knot, just as the $2\pi$ family of helicoids around a straight-line screw dislocation fill space around the defect line.  Furthermore, all of the fibres are topologically equivalent, \emph{i.e.}, of the same genus.  While fibrations are defined on compact spaces, one can project the Seifert surfaces from $S^3$ to $\mathbb{R}^3$ to study the fibration (with one exceptional surface passing through the point at infinity).  The use of $S^3$ is natural for knot fibrations, as we can identify $S^3$ with the Euclidean 3-sphere of unit radius in $\mathbb{R}^4\cong\mathbb{C}^2$ so that we may use a pair of complex variables $(u,v)$ with $|u|^2+|v|^2=1$ as coordinates on $S^3$.  For example, an explicit way to build a knot fibration is with a Milnor map $f:\mathbb{C}^2\to\mathbb{C}$ whose zeros on $S^3\subset \mathbb{C}^2$ form a knot $K$.  We may then associate, via $f(u,v)/|f(u,v)|=e^{i\theta}$, a point on the circle to every pair of complex numbers $(u,v)$ corresponding to a point on $S^3\backslash K$.  The level sets of the circle-valued map $f(u,v)/|f(u,v)|$ are a $2\pi$ family of surfaces with common boundary $K$.  The simplest choice of $f(u,v)$,

\begin{align}\label{eq:torusknot}
   f(u,v)=u^p+v^q,
\end{align}

yields zeroes that define a $(p,q)$-torus knot or link: $p=1,q=0$ yields an unknot, $p=2,q=2$ a Hopf link, and $p=3,q=2$ is trefoil knot. The example above originates from the study of isolated singularities of complex plane curves. The knots and links that can be obtained as the links of such isolated singularities are a comparatively small set of fibred links, called the algebraic links. It follows from basic approximation arguments that fibrations given by the argument $f/|f|$ of a real polynomial map $f:\mathbb{R}^4\to\mathbb{R}^2$ exist for every fibred link. However, constructing such polynomial fibrations for a given fibred link is very challenging. By constructing isolated singularities of real polynomial maps, explicit expressions for Milnor fibrations have been found for certain families of links that go beyond the set of algebraic links\cite{BodeLemniscateFibrations,Perron}.  

The surface geometry of the Seifert surfaces of a knot fibration near the knot is that of the pages meeting along the binding of an open book.  As such, knot fibrations are also known as open book decompositions with the corresponding fibred knot or link referred to as the binding of the open book.  Note, however, that not all knots have an open book decomposition in $S^3$: not all knots $K$ admit a fibration map $f:S^3\backslash K\rightarrow S^1$.  A test for whether a knot is fibred is to compute the knot group, $G\equiv\pi_1(S^3-K)$, and its commutator subgroup $[G,G]$.  A knot is fibred if and only if the commutator subgroup is finitely generated\footnote{This also translates to the Alexander polynomial of the knot: if the Alexander polynomial is not monic, then the knot is not fibred.\cite{Rolfsen_Knots_And_Links}}\cite{Neuwirth_Thesis,Rolfsen_Knots_And_Links,Stallings_Comm_Subgroup}.  Fibredness of links is also detected by Link Floer homology\cite{Ghiggini, Ni}, a link invariant that categorifies the Alexander polynomial.  Examples of knots without fibrations, in the notation of Rolfsen\cite{Rolfsen_Knots_And_Links}, are the $5_2$ knot, the $6_1$ knot (also known as the Stevedore knot), and a pair of split unknots.  The topological origin and consequences of the failure of a knot to be fibred on $S^3$ will be explored through the lens of smectic defects, as the global, topological structure of smectic configurations with knotted defects is sensitive to whether the knot in which a defect is tied is a fibred knot.

\section{Knotted Screw Dislocations}\label{sec:screw}

\subsection{Fibred Knots}

We begin our analysis of knotted smectic defects with screw dislocations.  As discussed in Section \ref{subsec:smectic_defects}, screw dislocations are singular lines in $\Phi$ and $\nabla\Phi$ defined by a fully-radial defect core.  Thus, in order to construct a screw dislocation line in the shape of a knot, one requires a family of surfaces that meet radially along the knot.  This topological structure is provided by knot fibrations: given a map $f:S^3\backslash K\rightarrow S^1$, one can define smectic layers as the fibre surfaces $f^{-1}(\theta)$.  Since the fibres are a $2\pi$ family of Seifert surfaces that meet radially along their common, knotted boundary, the local defect structure is that of a screw dislocation.  Further, knot fibrations ensure that the smectic layers away from the knotted screw dislocation are smooth and space-filling so that there are no other singular points or lines in $S^3$.  Note that the phase field of a straight-line dislocation, $\Phi_{\text{screw}}=z-\arctan(y/x)$, is only defined modulo $2\pi$: the level sets $\Phi=0$ and $\Phi=2\pi n$ for $n\in\mathbb{Z}$ are identical due to the translational symmetry $z\sim z+2\pi n$ of the screw dislocation.  Hence, the phase field $\Phi_{\text{screw}}$ can be viewed as a circle-valued map $\Phi_{\text{screw}}:\mathbb{R}^3\backslash\{x=0,y=0\}\rightarrow S^1$, making knot fibrations $\Phi:S^3\backslash K\rightarrow S^1$ a natural generalization for the phase fields of knotted screw defects.

While knot fibrations provide the appropriate topological structure for the phase field on $S^3$, one must project to $\mathbb{R}^3$ in order construct smectic configurations with knotted screw dislocations.  The standard projection from $S^3$ to $\mathbb{R}^3$, both for physical models involving knots\cite{Optical_Knots} and general knot visualization, is stereographic projection.  However, recall that the fibre structure ensures that $S^3$ -- and under stereographic projection $\mathbb{R}^3$ as well -- is filled by the $2\pi$ family of fibre surfaces.  This alone is not problematic, as the straight-line screw dislocation also fills $\mathbb{R}^3$ with only a $2\pi$ family of helicoidal surfaces.  But, stereographic projection sends a point on $S^3$ to infinity in $\mathbb{R}^3$, resulting in a point charge in the system.  As a consequence, $\nabla\Phi$ is not constant at the boundary, and there is no natural layered structure far from the knot.  Boundary conditions played an important role in understanding straight-line dislocations as well: in the family of radial surfaces meeting along a line shown in Figure \ref{fig:ScrewRadialSurfaces}, the failure to satisfy boundary conditions was attributed to the disclination charge in the system.  The triviality of $+1$ disclination charge in a nematic director field enabled the configuration to be twisted into helicoids, naturally satisfying boundary conditions and forming layers far from the defect core.  However, in the case of knot fibrations under stereographic projection, the non-trivial nature of point charge in nematic and smectic liquid crystals means there is no natural way to make such a configuration conform to boundary conditions without introducing other defects into the system.

In this paper, we present two methods of projecting a knot fibration in order to produce a smectic configuration with a knotted screw dislocation in $\mathbb{R}^3$.  The first method, described below, utilizes a hyperbolic disclination loop to screen the would-be charge from stereographic projection, yielding a smectic configuration either on $\mathbb{R}^3$ or within a droplet (a topological three-ball).  The second projection method produces, at the cost of two point defects, a screw dislocation in $\mathbb{R}^3$ wherein all the surfaces of the knot fibration reach the boundary with constant $\nabla\Phi$, similar to the straight-line screw dislocation.  While the latter method aligns with that of knotted edge dislocation and disclinations described in Sections \ref{sec:edge_dislocations} and \ref{sec:+12}, it requires the introduction of Morse critical points as they relate to smectic liquid crystals and a particular projection that involves cutting between such points on $S^3$.  As such, we will introduce this second projection method in the context of edge dislocations and apply it to screw dislocations only after they are revisited for non-fibred knots in Section \ref{sec:fibrednessANDdislocations}.  Although appearing somewhat different, these two projection methods arrive at the same defect topology in $\mathbb{R}^3$.

Here we describe our first projection method, which can be used to visualize the smectic layers defined by a knot fibration in $S^3$.  Rather than projecting $S^3$ to the entire Euclidean space $\mathbb{R}^3$, we will map to a three-dimensional ball. We first explain this projection in two dimensions, from $S^2$ to the disk, and then describe the three-dimensional analog. Stereographic projection identifies $S^2$ with the plane $\mathbb{R}^2$, mapping a chosen point on the two-sphere, for example the north pole, to the point at infinity. Since $\mathbb{R}^2$ is diffeomorphic to the interior of the disk via a simple rescaling, we can identify the two-sphere with a disk, where the entire boundary circle is identified with the north pole. This identification results from poking a hole into the two-sphere and unwrapping the punctured sphere as indicated in Figure~\ref{fig:Projection_Example}.
Projecting a phase field with domain $S^2$ that is regular on the north pole to the disk in this way results in a phase field on the disk that is constant on the boundary circle. The level set of the field on $S^2$ that contains the north pole corresponds to the following level set on the disk: the union of the boundary circle and a curve that intersects the boundary circle in exactly two points.  From the perspective of a smectic texture, the boundary identification as a point on $S^2$ is not problematic, as the curve and the bounding circle can be identified as one single smectic layer with two curvature defects where a line intersects the circle.  We can use these two point disclinations to fill in the rest of the smectic layers around the disk by connecting one, flat layer from infinity to each disclination.  These two layers and the disk separate $\mathbb{R}^2$ into two subspaces which are each homeomorphic to $\mathbb{R}^2$.  This allows the rest of space to be foliated by lines that, chosen to be evenly spaced and aligned at infinity, satisfy boundary conditions.  This projection method is shown in Figure \ref{fig:Projection_Example}, where the phase field on $S^2$ are longitudinal lines.  The resulting configuration contains two pairs of $\pm1$ disclinations: a charge-neutral dislocation pair.

\begin{figure}
    \centering
    \includegraphics[scale=.3]{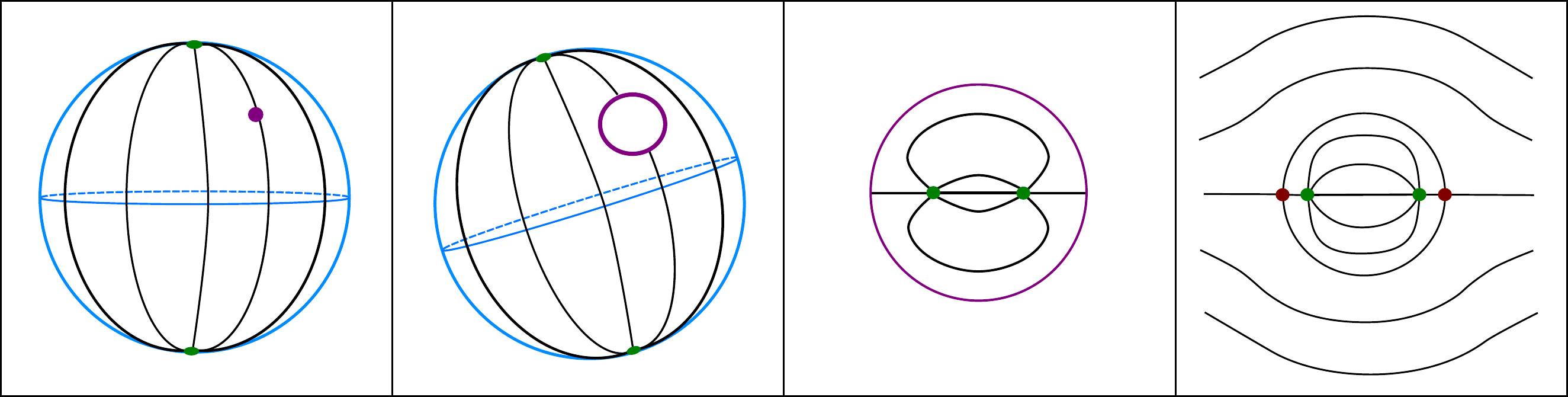}
    \caption{\label{Projection_Example}{Example of our projection method transforming transforming a phase field on $S^2$ to a smectic configuration on $\mathbb{R}^2$ satisfying the boundary condition that $\nabla\Phi=\hat{z}$.  Start with lines of longitude on the two-sphere (first panel) with two singular points (green) and a projection point (purple).  Poking and unwrapping from the projection point (intermediate stage shown in second panel) results in a phase field on the disk (third panel) with the bounding circle (purple) identified as the projection point on $S^2$.  The bounding circle is then turned into one smectic layer (black in final panel) with two curvature singularities (red), and layers are attached to each new defect.  Layers can then be filled in above and below in a manner that satisfies boundary conditions.  This process results in two pairs of $\pm1$ disclinations -- the topology of a charge-neutral pair of edge dislocations (final panel).}}
    \label{fig:Projection_Example}
\end{figure}

This process can be extended to project the phase field from a knot fibration in $S^3$ to a smectic configuration in $\mathbb{R}^3$.  Since $S^3$ is diffeomorphic to  the 3-ball with its boundary two-sphere identified to a point, there is a corresponding projection map that identifies fields on $S^3$ with fields on the 3-ball that are constant on the boundary sphere.  Explicitly, one can use complex coordinates $(u,v)$ on $S^3$ ($|u|^2+|v|^2=1$) parameterized by three hypersherical angles,

\begin{align}  
    u&=\cos(a) + i \sin(a)\cos(b), \\v&=\sin(a)\sin(b)\cos(c)+i \sin(a) \sin(b) \sin(c),
\end{align}  

where $a,b\in[0,2\pi]\in[0,2\pi]$ and $c\in[0,\pi]$.  The three-sphere can then be `unwrapped' along one of the hyperspherical angles $b$; namely, the projection to $\mathbb{R}^3$ is realized by mapping the hyperspherical angles to the usual spherical coordinates, $(b,a,c)\rightarrow (r,\theta,\varphi)$.  Without loss of generality we may assume that $\Phi_0:=\Phi^{-1}(0)$ is the fibre that contains the north pole, the projection point. The image of $\Phi_0$ under the projection map will intersect the bounding two-sphere along a closed curve, and along this singular circle we can attach to $\Phi_0$ a flat layer form infinity in the $xy$-plane.  This flat layer and the bounding two-sphere split $\mathbb{R}^3$ into two components which are individually homeomorphic to $\mathbb{R}^3$, each of which can be foliated with flat smectic layers.  The result of this projection is simply a higher-dimensional analog of the two-sphere projection shown in the final panel Figure \ref{fig:Projection_Example}.  The knot fibration -- \emph{i.e.}, the surface structure of a screw dislocation -- sits within a three-ball and is surrounded by flat smectic layers and a hyperbolic singular line.  Note that, by requiring $\nabla\Phi$ to be a constant far from the knot, one may insert a knotted screw dislocation into a generic smectic texture; since the smectic must be, locally, in its ground state away from defects, one can excise a section of flat layers and replace them with a knotted screw dislocation under the above projection.  This is made possible by an unknotted\footnote{Note that any simple loop that lies on an embedded two-sphere must be the unknot.} hyperbolic disclination loop of charge $-1$, which screens the would-be point charge of stereographic projection from $S^3$.  Figure \ref{fig:Knot_Open_Book} shows the smectic layers around a trefoil-knotted screw dislocation under this projection.

\begin{figure}
    \centering
    \includegraphics[scale=.35]{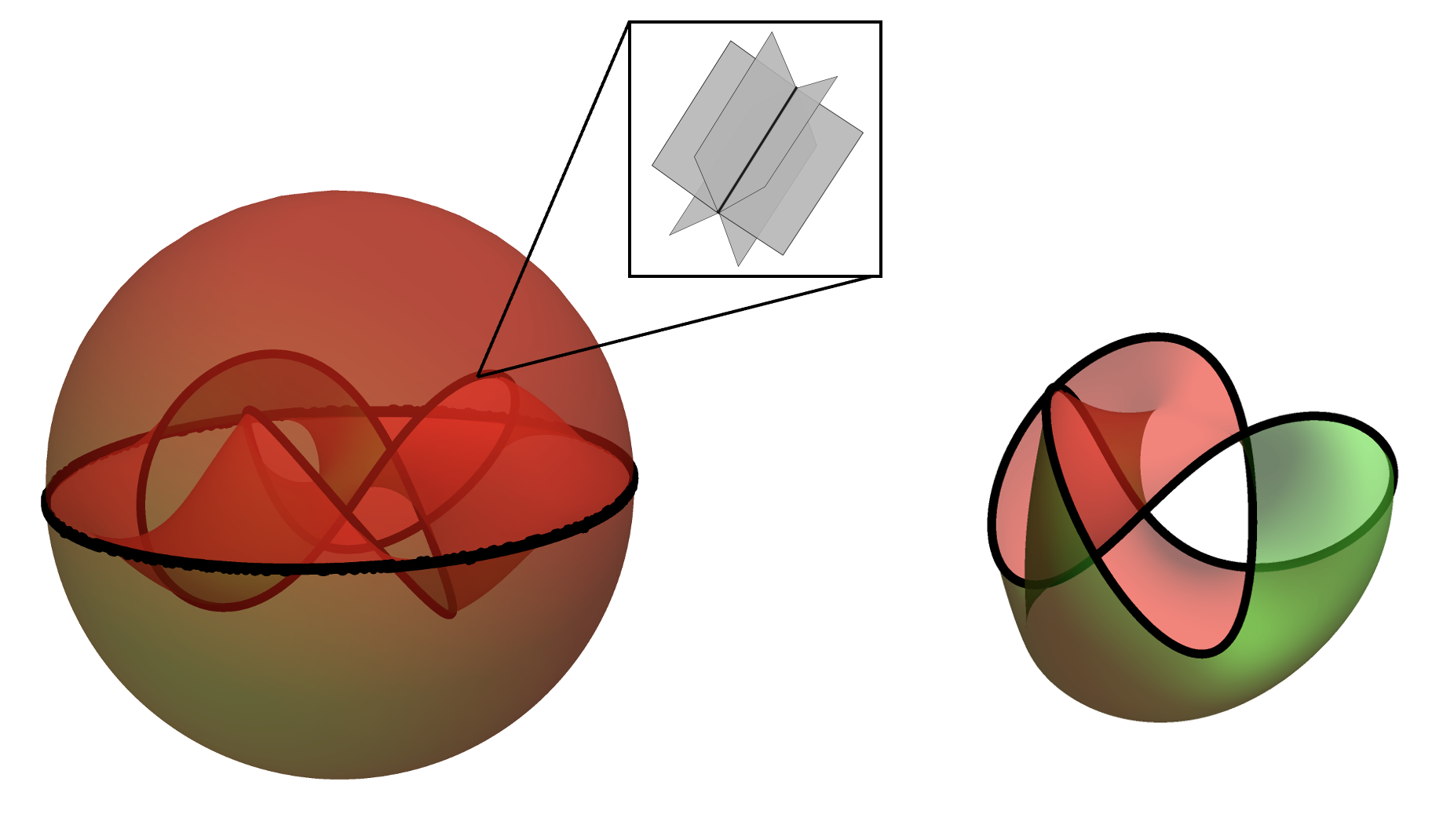}
    \caption{\label{Knot_Open_Book}{Structure of a trefoil-knotted screw dislocation.  Locally a $2\pi$ family of surfaces meet radially along the knot (inset).  One surface ($\Phi_0$) of the fibration (left) passes through the projection point on $S^3$, intersecting the bounding two-sphere along a circular disclination line.  Taking this surface and the sphere to be one smectic layer, we attach a sheet to the disclination line and foliate space above and below with flat layers identically to the lower dimensional analog in Figure \ref{fig:Projection_Example}.  All other surfaces bounding the knot -- for example, the surface defined by $\Phi=\pi$ (right) -- lie within the three-ball.  Green and red hues represent the two sides of the Seifert surfaces.}}
    \label{fig:Knot_Open_Book}
\end{figure}

While the disclination arising from this projection is shown as an equatorial loop in Figure \ref{Knot_Open_Book}, it can be shrunk to a point defect.  By a simple diffeomorphism in $\mathbb{R}^3$, we can move the surface $\Phi_0$ (the Seifert surface that intersects the projection point) inside the three-ball and the flat layer attached outside the three-ball jointly upward.  This moves the $-1$ disclination upward on the bounding two-sphere, eventually shrinking the disclination loop to a point at the north pole.  The local surface structure of this point defect is shown in Figure \ref{Fig:Screw_Projection_LoopPoint}.  After shrinking the disclination loop to a point, we can subsequently split the point defect into two (Morse-type) point defects connected along a layer.  This can be used as a springboard for our second projection method, outlined for edge dislocations on $S^3$ in Section \ref{sec:edge_dislocations} and applied to screw dislocations in Section \ref{sec:fibrednessANDdislocations}.  By explicitly using point defects on $S^3$, the second projection also realizes a screw dislocation in $\mathbb{R}^3$ with two point defects, demonstrating why the surface structure of two Morse-type point defects is always required.

\begin{figure}
    \centering
    \includegraphics[scale=.55]{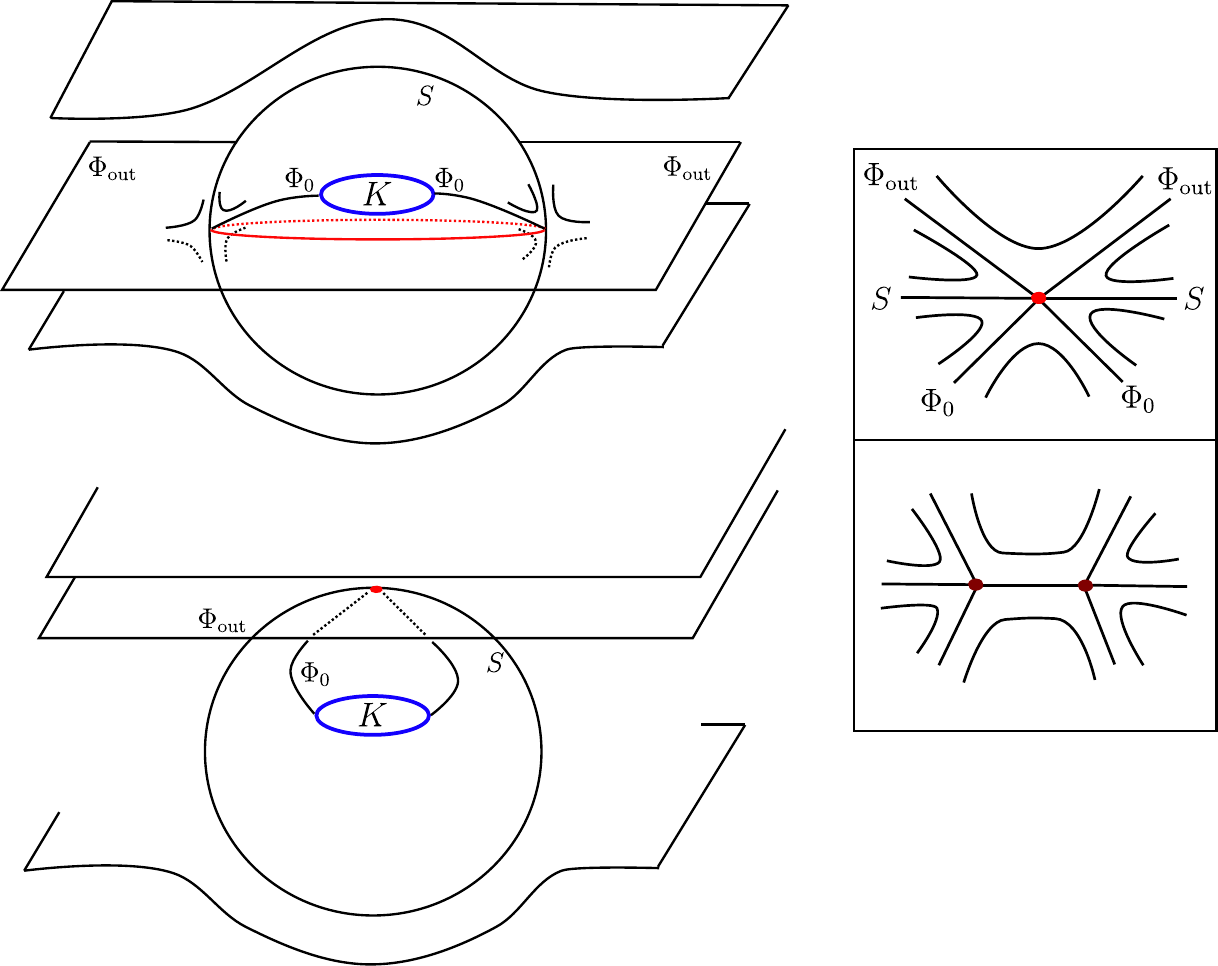}
    \caption{\label{Fig:Screw_Projection_LoopPoint}{The projection method described in this section produces a knotted screw dislocation within a three-ball surrounded by a hyperbolic disclination, shown on the top left.  By shrinking the disclination loop to a point at the north pole of the bounding two-sphere $S$, this defect can be made into a point (bottom left).  A cross sectional view of this point defect is shown on the top inset with the original surfaces labeled.  The topology of this point defect, which will be elucidated in Section \ref{sub:notfibredSCREW}, can be visualized by splitting the defect into two Morse-type critical points connected along a layer, shown in the bottom inset.}}
    \label{Fig:Screw_Projection_LoopPoint}
\end{figure}

\subsection{Screw Geometry}

While knot fibrations ensure that knotted defects are, topologically, screw dislocations, it is nonetheless worthwhile to connect back to the geometric properties of screw defects.  Consider the traditional description\cite{Volterra1907,chaikin_and_lubensky} of smectic and crystalline defects, which classifies screw dislocations according to a displacement parallel to the defect line when traversing a path around the dislocation.  In this Volterra scheme, measuring the direction of the displacement requires reference to a global ground state\cite{Severino_Kamien}.  While the topological and geometric classifications of screw dislocations match for a straight-line defect, the failure of knotted screw dislocations to be a simple perturbation about the ground state complicates the geometric description of dislocations.  In order to generalize the notion of `displacement' in a way compatible with the local coordinates around the dislocation rather than the global coordinates of boundary conditions, one can first consider surrounding a straight-line screw dislocation with a cylindrical tube.  Within this tubular neighborhood of the defect, the smectic layers are twisting bands bounding the defect line $l$ and a curve $\alpha$ along which the helicoidal layers intersect the cylindrical tube.  The curve $\alpha$ captures the vertical displacement of the screw dislocation: following $\alpha$ one full turn around the defect results in a displacement parallel to the defect line by an integer multiple of the layer spacing.  As such, an alternative means to capture the displacement of a straight-line screw dislocation is via the non-zero twist of $\alpha$ around the defect line $l$, $\text{Tw}(\alpha,l)\neq0$.  Since this definition of displacement is local -- making no reference to boundary conditions or a ground state -- it generalizes naturally to knotted screw dislocations.

This analysis of the straight-line screw dislocation can be repeated to study the geometry of knotted screw dislocations.  Start by taking a tubular neighborhood -- a knotted sheath -- around the knotted defect.  Each smectic layer, individually Seifert surfaces defined by the knot fibration, intersects a tubular neighborhood of the knot along a closed curve $\alpha$ that is, itself, knotted.  This construction is known as the Seifert framing of a knot\cite{Meilhan}.  In analogy to the non-zero twist around a straight line screw dislocation, the quantity $\text{Tw}(\alpha,K)$ captures the screw-like geometry of the knotted defect.  Within a tubular neighborhood of the knot, the Seifert surfaces twist about the knot at each crossing, and the total, signed number of such twists is characterized by $\text{Tw}(\alpha,K)$.  Given a particular embedding of the knot and a knot fibration, this twist allows us to interpret the surface structure within a tubular neighborhood of the defect as helicoid-like bands that twist about the defect line.  As with any pair of curves, the twist can then be extracted from C\u{a}lug\u{a}reanu's theorem, $\text{Lk}(\alpha,K)=\text{Tw}(\alpha,K)+\text{Wr}(K)$.  However, the Seifert framing of a knot\footnote{For links, the sum of the linking numbers is zero, meaning that individual components may have non-zero linking number with their parallel copies.} has the special property that the intersection of any Seifert surface with a tubular neighborhood of a knot has linking number zero, $\text{Lk}(\alpha,\Tilde{K})=0$.  For knotted screw dislocations, the twist is then uniquely determined by the writhe of the knot, which can, in turn, be calculated directly from the given knot diagram by counting the signed crossings.  This can be used to determine, for a given knot and its embedding, an analog of the Burgers vector integrated along the total length of the knot.

While the Seifert framing provides intuition for the screw-like geometry of smectic defects constructed from knot fibrations, it does not define a topological invariant.  Writhe \emph{is not} a knot invariant: Reidemeister \RNum{1} moves, which add or remove an extraneous crossing from a knot diagram, increase or decrease the writhe.  In the Seifert framing of a knot, Reidemeister \RNum{1} moves thus subtract or add twist, meaning that the smectic layers within a tubular neighborhood of the defect can arbitrarily change the total number of twists they make along the knot.  Reidemeister moves do not change the defect structure along the knot -- the defect still retains its purely-radial surface structure everywhere, continuing to be a screw dislocation.  However, the geometric description of dislocations breaks down, as the twist and corresponding analog of Burgers displacement are allowed to vary.  Note, though, that even for a straight-line screw dislocation, the pitch of the helicoid and corresponding Burgers displacement is set by energetics and not topology.  While a helicoidal geometry is required by topological considerations\cite{Severino_Kamien}, the exact pitch is set by the preferred layer spacing of the smectic liquid crystal.  In this spirit, the Seifert framing answers the following question.  Suppose one starts with a straight-line screw dislocation of length $L$, where $L$ is an integer multiple of the layer spacing $a$.  If a neighborhood of the straight-line defect is to be tied into a knot while retaining the same number of helicoidal turns, the knot must have a particular writhe, $\text{Wr}(K)=-\text{L}/a$.  A knot fibration ensures that the layers, as Seifert surfaces for the knot, match up smoothly away from such a defect core.  The final conformation of the knotted defect (its writhe) and smectic layers (their twist around the knot) would then be determined by energetics.

Although the net twist around a knotted screw dislocation can be zero across the entire knot, the Seifert surfaces locally twist at crossings in the knot diagram.  There is a unique case in which there is not only zero net twist but zero local twist of Seifert surfaces around the defect: the planar unknot.  It seems here that the topological definition of screw dislocations in terms of purely-radial defect structures fails completely to match intuition from the geometry of straight-line screw dislocations.  However, in the special case of a zero-writhe unknot fibration, the projection from $S^3$ to $\mathbb{R}^3$ results in a unique configuration.  Recall that the projection method described above cancels the point charge induced from maps to $\mathbb{R}^3$ by a hyperbolic disclination loop of charge $-1$.  When the defect on $S^3$ is a planar unknot, the unknotted defect in $\mathbb{R}^3$ can be brought into a neighborhood of the the hyperbolic disclination from projection so that the two defects sit in parallel.  A radial and hyperbolic disclination of charge $-1$ and $+1$ sitting in parallel is an edge dislocation: the original defect from the knot fibration and the defect induced by projection to $\mathbb{R}^3$ pair up to form not a screw but an edge defect!  Note that nontrivial knots cannot be brought into parallel with the necessarily unknotted hyperbolic defect.  Only in the particular case of a planar unknot can the extra hyperbolic disclination play a role in the dislocation structure itself.

\subsection{Knots without Fibrations?}

Given a fibration $f:S^3\backslash K\rightarrow S^1$, one can construct a smectic screw dislocation in the shape of $K$ using our methods of projection.   However, not all knots are fibred in $S^3$.  How might these knots be realized in a smectic configuration as knotted screw dislocations?  One could imagine, for example, taking a tubular neighborhood of a straight-line screw dislocation and manually tying it into a knot that is not fibred.  By definition, the surfaces around the knot cannot be grown outwardly and smoothly connect to form a family of Seifert surfaces.  The smectic configuration surrounding a knotted screw dislocation for a non-fibred knot must, then, have singularities.  But, will the resultant defects be walls, lines, or points?  To what topological class of smectic defects do these singularities belong?

One particular family of non-fibred links are the so-called split links.  A link $L$ is split if there is a two-sphere $S$ disjoint from $L$ such that both connected components of $S^3\backslash S$ intersect $L$. Calling the corresponding links in the two components $K_1$ and $K_2$, we say that $L$ is the split union of $K_1$ and $K_2$. If $K_1$ and $K_2$ are themselves fibred links, we may use method to construct smectic configurations for $K_1$ and likewise for $K_2$ individually. Since these configurations satisfy the usual smectic boundary conditions, we can insert the two configurations separately in the smectic. The resulting smectic configuration has screw dislocations lines given by the split union of $K_1$ and $K_2$, \emph{i.e.}, $L$, and two split, unknotted hyperbolic disclination lines coming from our projection method applied to $K_1$ and $K_2$ separately.  For example, a pair of split, unlinked circles does not have a fibration on $S^3$.  In this case, our projection method applied to each split circle individually reproduces a smectic with multiple edge dislocation loops, as expected.  Given knots or links of only one connected component and no fibration on $S^3$, we have, in principle, multiple options for how to form a knotted screw dislocation of said type.  One method would be to add an extra unknotted component to the knot on $S^3$.  Stallings proved that any knot can be turned into a fibred link on $S^3$ by adding an unknot of linking number zero with respect to the original knot\cite{Stallings_Fibred_Knots_Links}.  As such, any knotted screw dislocation may be formed with the presence of another unknotted defect.

Alternatively, we could forge ahead and attempt to create a screw dislocation that is a knot $K$, without fibration on $S^3$, alone.  Technically, there is nothing about our description of screw dislocations that requires an open book decomposition, that is, a knot fibration. Every knot $K$ admits smooth maps $\Phi:S^3\backslash K\to S^1$  such that regular level sets form Seifert surfaces of $K$. The distinction between such general maps and knot fibrations is by definition, the existence of critical points. In other words, $\Phi$ is a fibration map if and only if there are no points $(u,v)\in S^3$ where the gradient $\nabla \Phi(u,v)$ vanishes completely. Here, the gradient simply refers to three linearly independent directional derivatives in $S^3$. The condition that regular level sets are Seifert surfaces of $K$ thus means that a level set of $\Phi$ that does not contain any critical points is a Seifert surface. This is equivalent to $\Phi$ lying in a certain homotopy class on the complement of the knot\cite{hirasawa2003constructions}.
We may then apply our method to the map $\Phi$ in order to make a smectic configuration in $\mathbb{R}^3$ even though $\Phi$ is not a fibration. The resulting smectic configuration has a screw dislocation line in the shape of $K$ and no other line defects. However, since such a circle-valued Morse function $\Phi$, has, in general, critical points, the smectic configuration also contains a certain number of point defects.

A relatively recent result\cite{weber} from the study of circle-valued Morse functions states that, for any knot, one can always achieve a map $\Phi:S^3\backslash K\rightarrow S^1$ by allowing a certain number of Morse critical points, the minimum number of such points required to achieve a singular fibration is known as the Morse-Novikov number, $MN(K)$, of the knot.  This has important consequences for knotted screw dislocations, as it implies that the defect set need only be zero-dimensional (points) rather than lines or planes.  Yet, the application of this result is not straightforward.  For one, while the Morse-Novikov number can be related to other knot invariants\cite{weber,pajitnov,hirasawa2003constructions}, it is, in general, difficult to compute\footnote{In \cite{BodeMikami} the third author outlined a construction of polynomial maps for any given link $L$ such that the (not necessarily minimal) number of critical points of the resulting phase field can be directly calculated from a braid representative of $L$.}.  Further, the topological character of the critical points -- why they are required for non-fibred knots and to what they correspond as smectic defects -- is not immediately obvious.  However, by switching gears and studying knotted edge dislocations, we gain insight into the role of point defects in the topology of knotted line defects.  For knotted edge dislocations, we will show that Morse critical points are already required for \emph{any} non-trivial knot, regardless of fibredness.  The presence of these point defects -- which we will relate back to the point defects of focal conic domains -- will help us reinterpret the origin of Morse-Novikov points and reveal the relationship between fibredness and smectic defects.

\section{Knotted Edge Dislocations}\label{sec:edge_dislocations}

\subsection{Morse Points and Smectic Phase Fields on $S^3$}\label{sub:MorseAndProjection}

In the following sections we study the topology of knotted edge dislocations using tools from Morse theory.  To this end, we must first introduce Morse functions as they relate to smectic liquid crystals.  Morse functions are defined as smooth, real-valued functions on some domain manifold that only have non-degenerate critical points.  For a standard reference that covers all the fundamental definitions and results we point the readers to \cite{milnor}. There is also an analogous notion for circle-valued Morse functions \cite{pajitnovbook}. The classification of smectic point defects as Morse critical points has important physical implications for the defect structure\cite{chen_goldstone} and has been reviewed in \cite{aspects_topology_smectics}.  Recall, for example, that the topological structure of focal conic domains is better understood using the Morse theory of their point defects, as discussed in Section \ref{sub:FCD}.  Importantly, though, a generic smectic configuration is not described by a (real or circle-valued) Morse function.  This is because smectic dislocations and disclinations are curves of critical points in $\Phi$ and $\nabla\Phi$, not isolated point singularities as required by Morse theory.  Furthermore, classical Morse theory applies to real-valued functions on compact manifolds whose level sets are compact surfaces without boundary.  Typical smectic phase fields, even without line defects, do not satisfy these requirements: the smectic layers are typically defined on all of $\mathbb{R}^3$ or a subset of $\mathbb{R}^3$ with boundary, and the smectic layers themselves are generically non-compact.  There are analogous theories for circle-valued Morse functions and singular foliations in more general 3-manifolds manifolds with or without boundary \cite{pajitnovbook,Morseboundary,vera1, vera2, Morsefoliation}.  However, we need not resort to the complications of such theories; instead, we require only the most basic results from real-valued Morse theory.  We enable this simplification by allowing smectic layers to live in $S^3$ with a projection that yields a smectic configuration in $\mathbb{R}^3$ with appropriate (constant $\nabla\Phi$) boundary conditions.  

As an example of how a smectic configuration can be treated as a real-valued Morse function on $S^3$, consider the smectic ground state.  In the ground state ($\Phi=z$) the level sets of $\Phi$ are flat and equidistant smectic layers.  Starting with the ground state configuration defined within a unit cube $X=[0,1]^3$ in $\mathbb{R}^3$, we glue a disk $D_t$ to $[0,1]^2\times \{z\}$ for every $z\in [0,1]$, resulting in $S^2\times [0,1]$.  This can be thought of as taking the flat layers of the ground state and `closing up' each flat surface at the boundary so that the smectic layers are, instead, spheres.  Now, glue two three-balls $B_1$ and $B_2$ along $S^2\times\{0\}$ and $S^2\times \{1\}$, respectively.  The resulting three-manifold is $S^3$, as $B_1':=B_1\cup_\varphi(S^2\times[0,1])$, where $\varphi$ is the gluing map, is another three-ball.  Thus $B_1'$ and $B_2$ yield the genus-zero Heegaard splitting of $S^3$: just as two disks glued along their boundary form a two-sphere, two there-balls glued together along their boundary make a three-sphere.  Figure \ref{Cube_S3} shows this embedding with $D_{1/2}$ the unique disk that passes through the point at infinity.

\begin{figure}
    \centering
    \includegraphics[scale=.3]{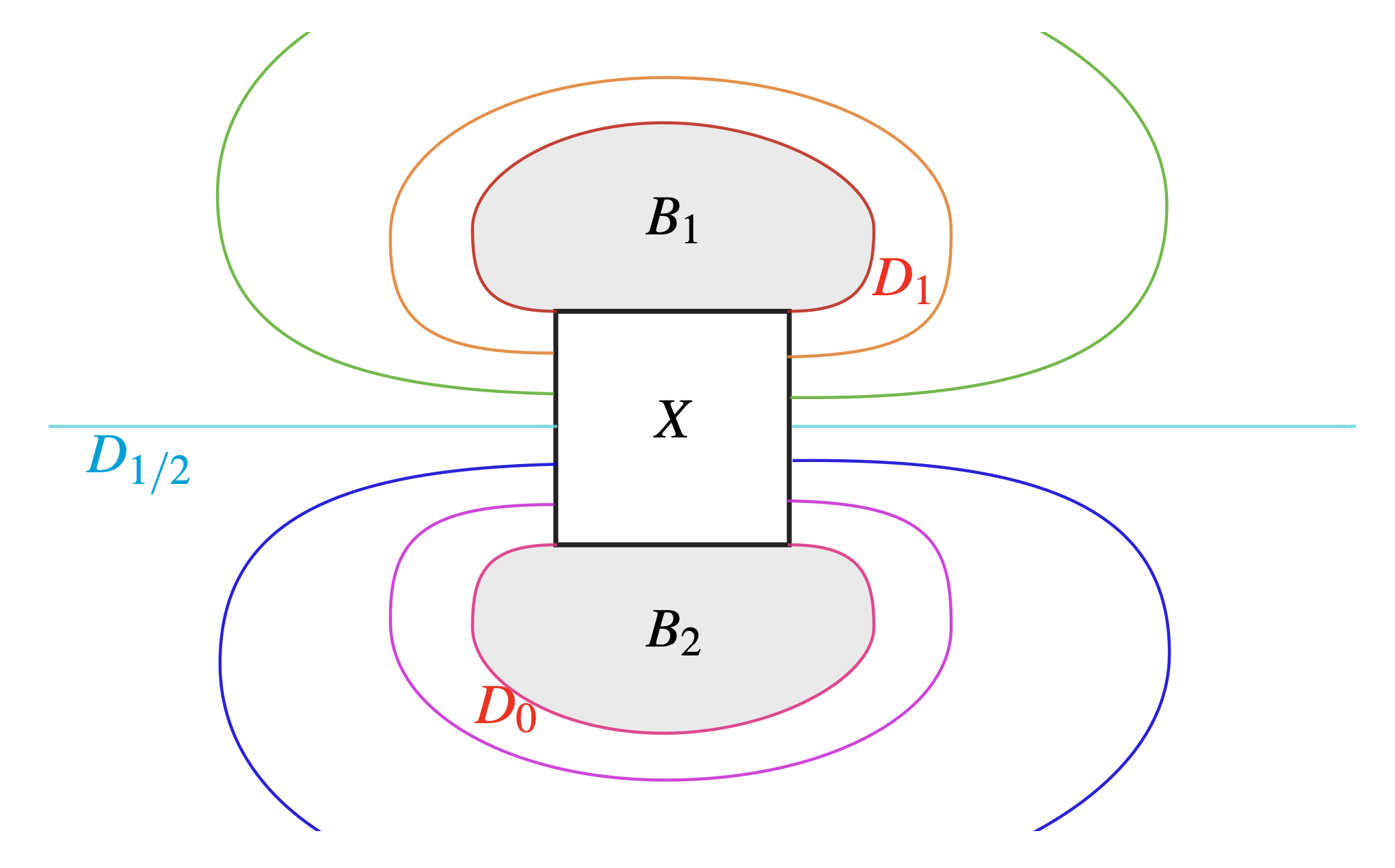}
    \caption{\label{Cube_S3}{Embedding a cube with smectic boundary conditions in the three-sphere, shown as a two-dimensional cross section.}}
    
\end{figure}

Given a smectic configuration within a unit cube $X$ containing only isolated (Morse) critical points and with constant $\nabla\Phi$ on the boundary, we can use this embedding of $X$ in $S^3$ to count the number of critical points based on the topology of the level sets.  This is because we can simply extend any map $f$ on $X$ with the desired boundary condition to a map $F:S^3\rightarrow S^1$ by using the embedding of $X$ in $S^3$ discussed above. The boundary conditions imply that $f$ is constant on $[0,1]^2\times\{t\}$ for all $t\in[0,1]$, allowing us to define $F$ to be constant on each added disk $D_t$.  On the three-balls $B_1$ and $B_2$ we define the level sets of $F$ to be concentric two-spheres and an isolated singular point at the center of each ball, located at $p_1\in B_1$ and $p_2\in B_2$.  Because the gluing operations are smooth, $F$ is smooth and is, again, a circle-valued Morse map whose critical points are the critical points of $f$ as well as $p_1$ and $p_2$.  Since $S^3$ is simply-connected, it follows that there is a map $\Tilde{F}:S^3\rightarrow\mathbb{R}$ such that $F=\sigma\circ\Tilde{F}$, where $\sigma:\mathbb{R}\rightarrow S^1$, $\sigma(x)=e^{\mathrm{i}x}$.  As $\sigma$ does not have any critical points, $\Tilde{F}$ is a real-valued Morse map, the critical points of which are exactly the critical points of $F$.  And, after a homotopy of $\Tilde{F}$ outside of $X$, we can assume that $p_1$ and $p_2$ are a global maximum and global minimum, respectively.  Since $p_1$ and $p_2$ are the only critical points of $\tilde{F}$ outside of $X$, it follows that the number of critical points of $f$ in $X$ is the number of critical points of $\tilde{F}$ in $S^3$ minus 2. We can thus determine the number of point defects in a smectic configuration, \emph{i.e.}, the number of critical points of $f$, by studying the number of critical points of the real-valued Morse function $\tilde{F}$. Determining the number of critical points of a Morse function from information on the topology of level sets is a key element of classical Morse theory, and it will allow us to simply count the number of point defects required to make a smectic configuration with a knotted edge dislocation.

\subsection{Local Edge Defect Structure}

Armed with the ability to study foliations of surfaces with smectic boundary conditions and only isolated singular points as real-valued Morse functions on $S^3$, we now turn our attention to edge dislocations.  Recall that edge dislocations are defined by a pair of parallel $\pm1/2$ disclinations at the defect core.  Due to the presence of these singular lines, smectic configurations with edge dislocations, strictly speaking, do not meet the requirements to be studied as Morse functions.  However, we will demonstrate that the smectic layers around knotted edge dislocations share the topology of a pair Seifert surfaces of the same knot.  This allows us to reduce the topological properties of knotted edge dislocations to foliations of $S^3$ involving a pair of Seifert surfaces -- \emph{i.e.}, a problem accessible to real-valued Morse theory.  To build intuition for how this procedure works and how a knotted edge dislocation may be constructed in the first place, consider the surface geometry of the smectic layers around a straight-line edge dislocation.  The presence of a hyperbolic disclination means that one level set of the density wave enters the $+1/2$ disclination and exits as two. Between the latter two surfaces lies the $-1/2$ disclination, and the rest of space is filled with parallel sheets that wrap around the $+1/2$ disclination.  This is illustrated in Figure~\ref{fig:Edge_Open_Book}.  In the same manner, constructing a smectic configuration with an edge dislocation in the shape of a knot requires two steps.  First, one must provide a set of smectic layers with the local structure of an edge dislocation -- in other words, a pair of $\pm1/2$ disclinations -- along a given knot.  Then, the space `above' and `below' the knotted dislocation must be filled with smectic layers.  

Using Seifert surfaces for a given knot, we can construct the local structure of an edge dislocation in the following manner.  Start with two Seifert surfaces, $\Sigma^1$ and $\Sigma^2$, sharing the same knotted boundary $K$.  Recalling that the dislocation core must be comprised of a parallel pair of $\pm1/2$ disclinations, take the surface $\Sigma^1$ to be the `input' layer to the hyperbolic $-1/2$ disclination.  We may then thicken $\Sigma^2$ to an interval worth of Seifert surfaces so that the surfaces $\Sigma^2_0:=\Sigma^2\times\{0\}$ and $\Sigma^2_{2\pi}:=\Sigma^2\times\{2\pi\}$ act as the two `output' layers for the $-1/2$ disclination.  The space between $\Sigma^2_{0}$ and $\Sigma^2_{2\pi}$ is filled by a 1-parameter family of equivalent Seifert surfaces $\Sigma^2_t$, $t\in[0,2\pi]$.  A brief examination of the density wave demonstrates that this surface geometry already contains the phase-field structure of a dislocation.  Take the value for the phase field on $\Sigma_1$ to be $\Phi=0$.  Choosing an orientation of $\Sigma_1$, passing in a neighborhood below the knot over to $\Sigma^2_0$ suggests that the value of the phase field on $\Sigma^2_0$ should be $\Phi=0$ as well.  To avoid singular points or lines, the path from $\Sigma^1$ to $\Sigma^2_0$ can be taken along the surface $\Tilde{\Sigma}^1\cup\Tilde{\Sigma}^2_0$ formed by gluing $\Sigma^1$ and $\Sigma^2_0$ along $K$ and pushing off below $\Sigma^1\cup\Sigma^2_0$.  From $\Sigma^2_0$, traversing the $2\pi$ family of Seifert surfaces $\Sigma^2_t$ leads to a value of $\Phi=2\pi$ for the surface $\Sigma^2_{2\pi}$.  However, passing in a neighborhood above the knot from $\Sigma_1$ to $\Sigma^2_{2\pi}$ (in an analogous way) suggests they too should share the same phase-field value.  As shown in Figure \ref{fig:Edge_Seifert}, this inclusion of $2\pi$ extra phase picked up by the closed circuit leads to the vertical displacement of an edge dislocation, and it occurs along the common knot bounding the Seifert surfaces.

Furthermore, this surface structure can be transformed into the usual pair of parallel $\pm1/2$ disclinations via only a local modification.  Given an orientation of the knot $K$, start by taking the surface $\Sigma^2_\pi$ and pushing its boundary a distance $d$ perpendicular to $K$ in the direction of $\Sigma_\pi$.  The boundary of the surface $\Sigma^2_\pi$ retains its knot type, but we have created a second singular line $\Tilde{K}$, of the same knot type as $K$, parallel to $K$.  Along the path between $K$ and $\Tilde{K}$, glue the surfaces $\Sigma^2_\delta$ and $\Sigma^2_{2\pi-\delta}$ together and smooth out their boundary.  As in the case of the straight line edge dislocation, the unoriented nature of the smectic layers resolves any tension in the alignment of the surface normals.  Repeating this for all $\delta$ until $\Sigma^2_{\pi\pm\epsilon}$ have been connected creates a $-1/2$ disclination at $K$ and a $+1/2$ disclination at $\Tilde{K}$.  The process of transforming a pair of Seifert surfaces into an edge dislocation is shown in Figure \ref{fig:Edge_Seifert}; this method is identical to that of Figure \ref{fig:Edge_Open_Book} extended to knotted edge dislocations.  Note, also, that this procedure is entirely local.  While the surfaces $\Sigma^1$ and $\Sigma^2_t$ are modified in a neighborhood of the knot, the smectic layers away from the dislocation and its disclination core are unchanged.

\begin{figure}
    \centering
    \includegraphics[scale=2.5]{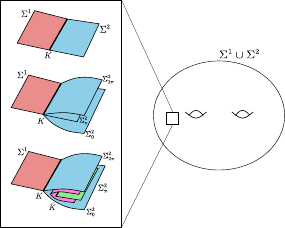}
    \caption{\label{Fig:EdgeKnot}{Constructing a knotted edge dislocation using two Seifert surfaces $\Sigma^1$ and $\Sigma^2$.  The inset demonstrates the local transformation of $\Sigma^1,\Sigma^2$ into a $\pm1/2$ disclination pair ($\Tilde{K}$ and $K$ respectively).  Note that $\Sigma^1\cup\Sigma^2$ is a surface without boundary of genus at least $2g$, where $g$ is the minimum genus of a Seifert surface for the knot.  Filling in the remaining smectic layers requires foliating the inside and outside of $\Sigma^1\cup\Sigma^2$, shown on the right, in $S^3$.}}
    \label{fig:Edge_Seifert}
\end{figure}

An important consequence of this construction is that the local surface structure of \emph{any} knotted edge dislocation can be reduced to that of just two Seifert surfaces.  Given a knotted edge dislocation in a real system, the topology of which must necessarily be composed of a parallel pair of $\pm1/2$ disclinations, we may, as explicitly described in the next section, run the above argument in reverse to reduce the topological properties of the smectic layers into that of a pair of Seifert surfaces.  This consideration becomes essential when studying the \emph{global} structure of a knotted edge dislocation -- in other words, how the rest of space around the dislocation is foliated by smectic layers.  Consider first an unknotted edge dislocation loop.  Running our argument in reverse, the local surface structure around the defect is reduced to two Seifert surfaces glued along the unknot.  Since the unknot has a genus zero Seifert surface, these two Seifert surfaces, glued along their common boundary, form a flat sheet that extends to the boundary of the unit cube $X$.  This allows flat layers to foliate above and below the sheet, as in the straight-line edge dislocation.  However, non-trivial knots have Seifert surfaces of necessarily non-zero genus.  The local structure of an edge dislocation for non-trivial knots, then, forces the smectic texture to contain surfaces of non-zero genus.  But, in order to satisfy boundary conditions, the surface structures far from the dislocation must be flat sheets.  The transition from flat sheets -- which, under our embedding of smectic configurations on $S^3$ are spheres -- to surfaces of non-zero genus cannot be accommodated continuously: it necessitates Morse critical points.  It follows that knotted edge dislocations force the presence of a certain number of point defects on the global smectic texture.

As a brief aside, note that, in order for the two original surfaces $\Sigma^1$ and $\Sigma^2$ to be Seifert surfaces, they must be taken to be compact.  However, in the case of edge dislocations, one expects smectic layers that extend to the boundary.  In analogy to the edge dislocation loop defined above in Eqn. \ref{eqn_edge} where the `input' surface enters the unknot from the boundary and `output' surfaces are compact, we may take $\Sigma^1$ to reach the boundary of $X$ in $S^3$.  Both choices are equivalent in our analysis, as we again study the surfaces on $S^3$ precisely so that all surfaces are compact.  In $\mathbb{R}^3$ the Seifert surface extending to the boundary can be considered a punctured Seifert surface, where identifying the circular intersection of $\Sigma^1$ with the boundary as a point results in a Seifert surface.

\subsection{Global Structure: $4g$ Point Defects}

Given that the local surface structure of a knotted edge dislocation is composed of a pair of Seifert surfaces glued along their common knot, we prove the following result about the number of point defects required to form a knotted edge dislocation.

\ 

\begin{theorem}\label{thmedge}
    Smectic textures containing a knotted edge dislocation, under the following assumptions, require a minimum of $4g$ point defects, where $g$ is the minimum genus of a Seifert surface for the given knot.
\end{theorem}

\ 

A priori this is a surprising result.  Theorem \ref{thmedge} states that any knottedness of the edge dislocaiton forces point defects in the system.  This should be contrasted with knotted disclinations in the nematic phase, which, beyond the knot itself, require no extra singularities (and, at most, can accommodate a point charge $q=1$\cite{Kamien_RevModPhys_Defects}).  Further, this result differs from screw dislocations, whose global structure may remain defect-free on $S^3$ for fibred knots.  Our assumptions in the proof are as follows:

\begin{itemize}
    \item \underline{No other line defects}: We assume that the smectic phase field $\Phi$ has an edge dislocation in the shape of a knot $K$, \emph{i.e.}, both of the $\pm1/2$ disclinations form isotopic knots $\tilde{K}$ and $K$ that lie in parallel. Our considerations apply to links as well, so that in general $K$ could also be multi-component link. However, there should be no additional screws or disclination lines in the system. This assumption is necessary, since we are specifically interested in a relationship between the topology of $K$ and topological properties of the liquid crystal configuration.
    \item \underline{Bounded domain}: We only consider smectic configurations on a finite (\emph{i.e.}, bounded) domain. The boundary condition thus implies that the level sets of $\Phi$ on the boundary $\partial X$ are horizontal (\emph{i.e.} $\nabla\Phi=\hat{z}$ on the entire boundary). This excludes some well-known configurations such as the standard edge or the standard screw dislocation, which were discussed in Section~\ref{sec:background}, as in both cases the defect lines are non-compact. Since we are concerned with (compact) knotted line defects and, by the first assumption, there are no other (non-compact) defects near the knot, this assumption should be somewhat generic.  Furthermore, this assumption allows the knotted edge dislocation to be inserted into the bulk of a larger smectic texture.  Since, far from any defects, smectic configurations must locally be in its ground state, we may build a larger smectic texture with a knotted dislocation by removing a cube of ground state and replacing it with $X$ as constructed in the following sections.  
    \item \underline{Orientability of $\Sigma$}: Let $\Sigma$ be the unique layer whose boundary is the $+1/2$ disclination $\tilde{K}$. We assume that $\Sigma$ contains no point defects and that it is orientable. Therefore, it is a Seifert surface of $\tilde{K}$ (if $\Sigma$ does not extend to the boundary $\partial X$) or a punctured Seifert surface of $\tilde{K}$ (if it extends to the boundary $\partial X$). In principle, the non-orientable character of the director field means that non-orientable layers, such as Moebius bands, are possible and are only ruled out by our assumption. As edge dislocations have no free half-integer disclination charge, non-orientability of surfaces and non-co-orientability of the foliation of surfaces need not be accommodated.  However, we will give an example of potential non-orientable surfaces when discussing isolated knotted $+1/2$ disclinations in Section~\ref{sec:+12}.
    \item \underline{The deformed $\Phi$ is Morse}: We assume that the phase field describing the smectic layers (away from the defects) is a Morse function on $X\backslash(K\cup\tilde{K})$.  With the chosen boundary condition on $\partial X$, fixed profiles in tubular neighbourhoods of $K$ and $\tilde{K}$, and assumption that there are no other line defects in the system, this is a generic condition in the sense that any non-Morse function -- \emph{i.e.}, a configuration with degenerate smectic point defects -- can be perturbed by an arbitrary small amount to make any degenerate critical points non-degenerate. 
\end{itemize}

A general Morse function on $X\backslash(K\cup\tilde{K})$ will have isolated critical points.  In constructing edge dislocations above, we started from a pair of Seifert surfaces and arrived at the surface structure of a parallel pair of $\pm1/2$ or, by the local transformation of Figure \ref{fig:Edge_Open_Book}, $\pm1$ disclinations.  By design the surfaces `involved' in the disclinations -- namely $\Sigma^1$ and those between $\Sigma^2_0$ and $\Sigma^2_{\pi}$ -- contained no critical points.  We assume only that $\Phi$ is a Morse function, meaning it may be the case that a foliation of $X\backslash(K\cup\tilde{K})$ contains a parallel pair of knotted disclinations whose associated surfaces themselves have critical points.  However, if such isolated points exist, we can move them through the layers until there are no critical points between $\Sigma^2_0$ and $\Sigma^2_{2\pi}$; in other words, the connected component of $X\backslash(\Sigma^2_0\cup\Sigma^2_{2\pi})$ that contains $\Tilde{K}$ contains no critical points.  The motion of the critical points does not alter the behavior of $\Phi$ on $\partial X$, the number of critical points, nor the knot types of $K$ and $\Tilde{K}$.  It does, however, change the topology of the surfaces through which the critical points move.

Since $\Phi$ is Morse, both $\Sigma^2_{0}$ and $\Sigma^2_{2\pi}$ are orientable and therefore Seifert surfaces of $K$ (or punctured Seifert surfaces if they meet $\partial X$).  The fact that there are no critical points between them shows that they are equivalent surfaces, \emph{i.e.}, that they have the same genus, say $g_0=g_{2\pi}$.  Likewise, every layer in the same connected component of $X\backslash(\Sigma^2_0\cup\Sigma^2_{2\pi})$ must share the same genus.  This allows us to run the construction of edge dislocations in reverse.  Starting from a knotted pair of parallel $\pm1/2$ disclinations, we may perform an operation within a neighborhood of the knotted pair that replaces the $\pm1/2$ disclinations with a $\pm1$ disclination pair.  We may then rid the connected component of $X\backslash(\Sigma^2_0\cup\Sigma^2_{2\pi})$ containing $\Tilde{K}$ of any critical points, making every surface therein a Seifert surface of $\Tilde{K}$ sharing the same genus.  The knot $\Tilde{K}$ can then be contracted onto $K$ by shrinking the surface $\Sigma_{\pi}$ connecting the two knots, leaving $\Sigma^2\times[0,2\pi]$ as the entire connected component, with boundary points of $\Sigma^2\times\{t\}$ identified on $K$.  We may thus homotope $\Phi$ to shrink the interval $[0,2\pi]$ to a single point, which deforms $\Sigma^2\times[0,2\pi]$ to simply $\Sigma^2$.  The resulting foliation no longer contains any line defects; the knot $\Tilde{K}$ is gone and the knot $K$ is no longer a defect.  Viewed in this way, the `dislocation' is now simply an embedded knot on the surface $\Sigma^1\cup\Sigma^2$, which can be made smooth via small deformation.  Note that this procedure does not change the number of critical points of $\Phi$.  Thus every density wave that satisfies our assumptions above and has a knotted edge dislocation in the shape of $K$ has the same number of critical points as a circle-valued function $f$ on $X$ that has as a connected component of a level set in the shape of two (possible punctured) Seifert surfaces $\Sigma^1\cup\Sigma^2$ of $K$ glued along $K$ and that satisfies the boundary condition on $\partial X$. Note that the collapse of the interval does not affect the topology of the layers outside of that interval of surfaces, so that $f$ can still be taken to be a circle-valued Morse function, which is now well-defined on all of $X$.

As a result, to count the number of point defects forced upon the smectic texture by the presence of a knotted edge dislocation, one need only calculate the number of critical points of a foliation of smectic layers on $X$ (in $S^3$) containing a level set that is the union of two Seifert surfaces.  To do this, we will need the following.

\ 

\begin{prop}\label{propMorse}
    Let $\tilde{F}:S^3\rightarrow \mathbb{R}$ be a smooth, Morse function such that there is a $y\in \mathbb{R}$ such that there is a connected component of $\tilde{F}^{-1}(y)$ that is a closed surface of genus $\Tilde{g}$.  Then $f$ has at least $\Tilde{g}$ critical points of index 1 and at least $\Tilde{g}$ critical points of index $2$.
\end{prop}

\begin{proof} 
Recall how the topology of a level set changes as we pass through a critical level set of a real-valued Morse function.  The level set $\Tilde{F}^{-1}(a)$ is diffeomorphic to $\Tilde{F}^{-1}(b)$ with $a<b$ if the interval $[a,b]$ contains no critical values of $\Tilde{F}$.  The Morse Lemma implies that if $[a,b]$ contains exactly one critical value then the Euler characteristic of $\Tilde{F}^{-1}(b)$ differs from $\Tilde{F}^{-1}(a)$ by $\pm2$.  The precise change of topology depends on the index of the corresponding critical point according to the following table.

\begin{table}[h!]
\centering
    \centering
    \begin{tabular}{| c | c | c | c | c |}
         \hline Index & 0 & 1 & 2 & 3   \\ \hline
         Change in Euler characteristic & +2 & -2 & +2 & -2 \\ \hline Change in number of components & +1 & 0 or -1 & 0 or +1 & -1  \\ \hline Change in sum genera of components & 0 & +1 or 0 & -1 or 0 & 0  \\ \hline
    \end{tabular}
\caption{The change of topology of a level set of a real-valued Morse function as we pass through a critical point of given index.}
\label{table:1}
\end{table}

The Morse Lemma also tells us, up to diffeomorphism, the form of the level sets in a neighborhood of any critical point.  The level sets in the neighborhood of index $0$ and $3$ critical points -- for example the global maximum and minimum $p_1$ and $p_2$ of our foliation of $S^3$ -- are concentric spheres.  The neighborhoods of index $1$ and $2$ critical points are the saddles shown in Figure \ref{fig:Index12}.  Armed with this information, we can prove Proposition \ref{propMorse}.  Start near the global minimum of $\Tilde{F}$, $p_2$, where, for sufficiently small $\epsilon>0$, $\Tilde{F}^{-1}(\Tilde{F}(p_2)+\epsilon)$ is a two-sphere.  By assumption there is a $y\in \mathbb{R}$ such that a connected component of $\tilde{F}^{-1}(y)$ is a closed surface of genus $\Tilde{g}$.  Thus there is a $\Tilde{y}\in\mathbb{R}$ such that the sum of the genera of the connected components of $\Tilde{F}^{-1}(\Tilde{y})$ is at least $\Tilde{g}$.  But, by the Morse Lemma outlined in Table \ref{table:1}, there exist at least $\Tilde{g}$ critical points of index 1 in the interval $[F(p_2)+\epsilon,\Tilde{y}-\epsilon]$; each critical point increases the genus by at most one, and hence a minimum of $\Tilde{g}$ such points are required to transition from the two-sphere $\Tilde{F}^{-1}(F(p_2)+\epsilon)$ to $\Tilde{F}^{-1}(\Tilde{y})$, a surface of genus $\Tilde{g}$.  By the same argument there are at least $\Tilde{g}$ critical points of index $2$ in $[\Tilde{y}+\epsilon,F(p_1)-\epsilon]$ to accommodate the transition from a surface of genus $\Tilde{g}$ to a two-sphere.  Hence $\tilde{F}$ has at least $2\Tilde{g}$ critical points of index 1 or 2 -- a minimum of $\Tilde{g}$ of index 1 and $\Tilde{g}$ of index 2 -- proving the Proposition.
\end{proof} 

\begin{figure}
    \centering
    \includegraphics[scale=.35]{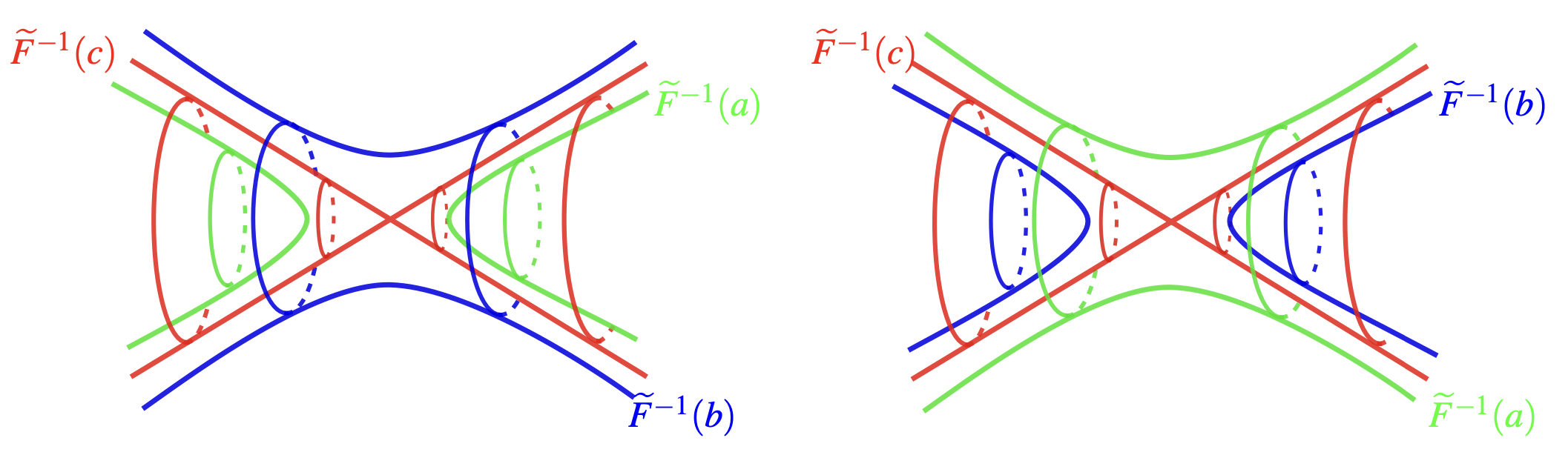}
    \caption{\label{fig:Index12}{Neighborhood of index 1 (left) and index 2 critical points.  Here $c$ is the unique critical value of $\Tilde{F}$ in $[a,b]$ for $a<b<c$.  As defects in smectic layers, these singularities are identical to those of the Dupin cyclides of Figure \ref{fig:cyclides}.}}
    \label{fig:critical}
\end{figure}

We now have enough information to prove Theorem \ref{thmedge}.

\begin{proof}  This follows immediately from Proposition \ref{propMorse} and the definition of smectic edge dislocations. We have already discussed in Section \ref{sub:MorseAndProjection} that the number of critical points of a smectic configuration containing only isolated critical points is equal to the number of critical points of a corresponding real-valued Morse function $\tilde{F}$ on $S^3$ minus 2. By construction of the local surface structure around an edge dislocation in the shape of a knot $K$, a smectic configuration with a knotted edge dislocation contains a level set that is the union of two (possibly punctured) Seifert surfaces.  On the Morse function lifted to $S^3$, there is a $y\in\mathbb{R}$ and a connected component of $\tilde{F}^{-1}(y)$ that consists of two Seifert surfaces of $K$, which are glued along $K$. Earlier we denoted the (punctured) Seifert surfaces by $\Sigma^1$ and $\Sigma^2$. Our embedding of $X$ in $S^3$ fills any possible punctures with the inserted disks $D_t$, so that the resulting layers really are two Seifert surfaces of $K$ in $S^3$, glued along $K$. Suppose that they have genera $g_1$ and $g_2$. Then the surface that is obtained by gluing them along $K$ has genus $\tilde{g}:=g_1+g_2$ and Proposition~\ref{propMorse} shows that $\tilde{F}$ has at least $g_1+g_2$ critical points of index 1 and at least $g_1+g_2$ critical points of index 2. It has at least one local minimum $p_2$ and one local maximum $p_1$ and thus the number of critical points of the corresponding liquid crystal in $X$ is at least $2(g_1+g_2)$. By definition $g$ is the smallest genus of any Seifert surface of $K$ and thus at most $g_1$ and at most $g_2$. Therefore, the number of critical points is at least $4g$.
\end{proof}

Theorem \ref{thmedge} shows that edge dislocations cannot form knots without forcing other defects in the system.  In particular, a knotted edge dislocation must have at least $4g$ point defects, where half are Morse-index 1 and half are Morse-index 2.  A schematic of a trefoil-knotted edge dislocation and its associated $4$ ($g=1$ for the trefoil knot) point defects is shown in Figure \ref{fig:EdgeSchematic}.  The topological origin of these points is clear: because a knotted edge dislocation forces smectic layers of non-zero genus, the smectic configuration must contain point defects that increase and decrease the genus.  This allows us to draw a direct analogy between the topological structure of knotted edge dislocations and focal conic domains.  Recall that, as discussed in Section \ref{sub:FCD}, Dupin cyclides require two point defects -- also of the form of index 1 and 2 Morse critical points -- to transform smectic layers from topological spheres to tori and back to spheres.  In the same way, a knotted edge dislocation requires an equal number of index 1 and 2 critical points.  For an edge dislocation, the layer transformation is not just form spheres to tori but rather from spheres (flat layers in $\mathbb{R}^3$) to the union of two Seifert surfaces.  Knotted edge dislocations are, viewed through this lens, a topological generalization of focal conic domains.  Knotted dislocations certainly do not share the \emph{geometric} properties of focal conic domains; by their nature edge dislocations cause compression of the smectic layers.  But, the requirement of pairs of point defects alongside a knotted edge dislocation can be interpreted as a focal conic-like topological structure.  The analogy between knotted dislocations and focal conic domains is furthered in Section \ref{sec:+12}, where we study the topology of knotted $+1/2$ disclinations.

As point defects in the smectic and nematic order, the $4g$ point defects may arrange themselves in a way compatible with smectic defects\cite{aspects_topology_smectics} but not Morse theory.  For example, the $2g$ index 1 and $2g$ index 2 point defect may combine into degenerate critical points.  The requirement of $4g$ point defects can then be interpreted as a constraint on the the defect topology -- the possibly-degenerate point defects in the system must share the surface topology of $2g$ index 1 and $2g$ index 2 Morse critical points.  Also, as smectic point defects, the critical points required for a knotted edge dislocation may degenerate into line defects, forming structures more akin to focal conic domains.  One may also notice that, in Morse theory, critical points with neighboring index can annihilate.  It might be tempting, then, to say that the smectic configuration could lower its energy by annihilating its $2g$ index 1 and $2g$ index 2 point singularities.  However, as a consequence of Theorem \ref{thmedge}, this annihilation sequence would change the genus of the smectic layers, and hence must undo the knottedness of the edge dislocation.  If the dislocation is to remain knotted, all $4g$ point defects must be present, even if the point defects could annihilate each other in the absence of the line defect.  Here we note a difference between focal conic domains and knotted edge dislocations.  The presence of the index 1 and 2 point defects in focal conic domains comes about through energetics; it is the requirement of perfect spacing that calls for the geometry of the Dupin cyclides in the first place.  On the other hand, the point defects required for knotted edge dislocations are stabilized not by energetics but rather by the knotted topology of the defect line.  Furthermore, the index 1 and index 2 point defects associated to the knotted edge dislocation must lie on different level sets: in $\mathbb{R}^2$ they sit `above' and `below' the knot, as shown schematically in Figure \ref{fig:EdgeSchematic}, in order to increase and decrease the genus in sequence.  Just as dislocation glide requires crossing a topological barrier\cite{Hocking_Kamien_Peierls_Nabarro}, these point defects cannot be moved across level sets without passing through singular configurations in the layered structure\cite{aspects_topology_smectics}.  The knottedness of an edge dislocation is thus topologically protected, as the unknotting process requires not just passing dislocations through each other but also moving point defects across layers in order to annihilate them, both of which cannot be smoothly accommodated in the smectic phase.

\begin{figure}
    \centering
    \includegraphics[scale=.55]{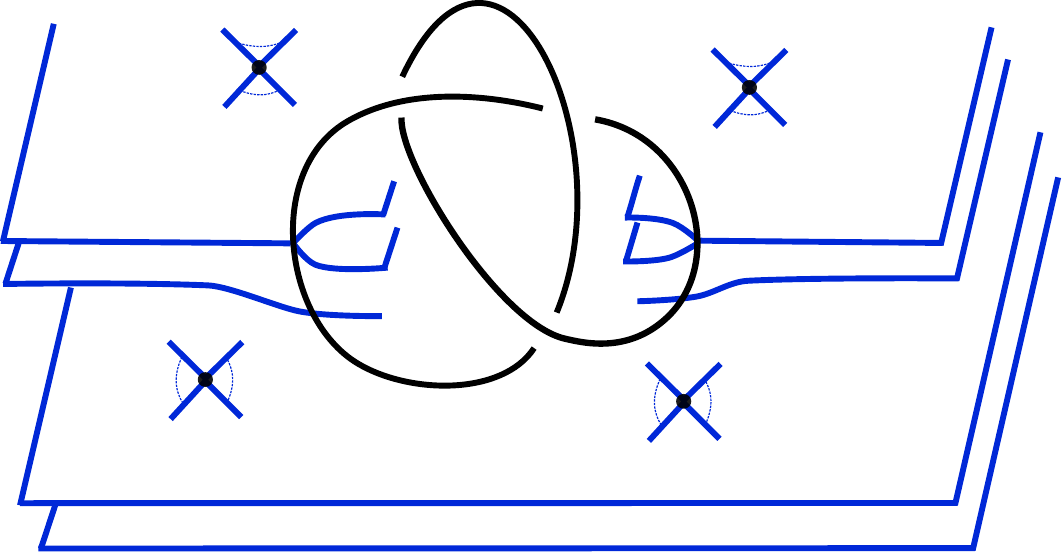}
    \caption{\label{fig:EdgeSchematic}{Schematic of a trefoil-knotted edge dislocation with four point defects.  The bottom two point defects increase the genus of flat layers to match that of the Seifert surfaces of the dislocation, and the top two point defects decrease the genus to restore the flat layers required from boundary conditions. }}
    \label{fig:EdgeSchematic}
\end{figure}

As a final remark on Theorem \ref{thmedge}, note that the minimal number, $4g$, of point defects required for knotted edge dislocations applies to knots but not multi-component links.  This is because the construction for edge dislocations relies on gluing two Seifert surfaces along their common boundary.  When there is only one boundary component, the genus is additive under this operation.  Generally, though, the Euler characteristic is additive under this gluing.  Consider two Seifert surfaces for a multi-component link.  Their individual Euler characteristics are $\chi_1=2-2g_1-b$ and $\chi_2=2-2g_2-b$, where $g_1,g_2$ are the genera of the Seifert surfaces of $\Sigma_1,\Sigma_2$ respectively and $b$ is the number of boundary components, \emph{i.e.}, the number of components of the link.  The Euler characteristic of the union of the two Seifert surfaces glued along their boundary is then $\chi=\chi_1+\chi_2=4-2(g_1+g_2)-2b$.  At the same time, the surface $\Sigma_1\cup\Sigma_2$ has no boundary and genus $\tilde{g}$, so $\chi=2-2\tilde{g}$.  Combining yields a genus $\tilde{g}=g_1+g_2+b-1$ for $\Sigma_1\cup\Sigma_2$.  When the number of boundary components is one -- \emph{i.e.}, when we have a knot -- the genus of the union of two Seifert surfaces is simply the sum of the individual genera.  The genus for multi-component links, on the other hand, depends on the number of boundary components.  Since $g_1$ and $g_2$ must be at least $g$, the minimum genus Seifert surface for the link, Theorem \ref{thmedge} generalizes simply to links: linked edge dislocations require a minimum of $2(2g+b-1)$ point defects.

\subsection{Non-Compact Dislocations}

By assuming trivial boundary conditions, the above analysis does not allow defects to reach the boundary of the system in $\mathbb{R}^3$.  However, it is possible to modify our approach so that the results of Theorem \ref{thmedge} apply to non-compact edge dislocations with knotted topology.  Consider first how a straight-line edge dislocation may be constructed from the union of two Seifert surfaces in $S^3$.  Start with the ground state foliation of the unit cube $X$ embedded in $S^3$ (Figure \ref{fig:Knot_Open_Book}) comprised of concentric spheres and the maxima/minima $p_1$ and $p_2$.  By taking an unknotted curve on one sphere, we can split the sphere into the union of two disks, $\Sigma^1$ and $\Sigma^2$, glued along the unknot.  Note that, in $X$, the unknot is nothing but a straight line extending from one side of $\partial X$ to the other along a layer.  We can take one of the surfaces, $\Sigma^2$, in $S^3$ and thicken it into a $2\pi$ family of surfaces.  By following the same procedure as before, we can transform the unknot into a pair of unknotted $\pm1/2$ disclinations, \emph{i.e.}, an unknotted edge dislocation on $S^3$.  What we have done has not changed the number of critical points in $S^3$; we still have a configuration with only two critical points, the maximum and minimum $p_1$ and $p_2$, only now with an unknotted edge dislocation.  The only difference is that this unknotted edge dislocation on $S^3$ reaches $\partial X$: in $X$, we have a configuration with a straight-line edge dislocation and no critical points, as expected.

Generalizing this process for any line dislocation $\alpha$ extending to the boundary of the sample, we can use the exact same Morse theory analysis to set a lower bound for the number of point defects required to have a given edge dislocation line.  Consider a foliation of $S^3$ containing the union of two Seifert surfaces, $\Sigma^1\cup\Sigma^2$, glued along a common knot $K$.  Unlike the previous sections wherein the knot $K$ was embedded entirely in $X$, we now allow $K$ to reach $\partial X$, piercing the boundary exactly twice at the same height.  We once again require trivial boundary conditions on $X$, embedded on $S^3$ in the same manner as Figure \ref{fig:Knot_Open_Book}.  Since the knot only intersects $\partial X$ twice, this foliation of $S^3\backslash X$ is always possible: the component of $K$ in $S^3\backslash X$ simply splits one of the punctured disks, \emph{i.e.}, the knot `connects back to itself' on the disk above $X$ in $S^3$.  Explicitly, the foliation of $S^3\backslash X$ is exactly the same as the straight-line dislocation (an unknot on $S^3$) described in the previous paragraph.  Once projected to $\mathbb{R}^3$ the to-be dislocation line forms a long knot -- a curve in $\mathbb{R}^3$ that `closes up' by passing through the point at infinity\cite{BudneyLongKnots,VassilievCohomologyKnotSpaces}.  Note that, because the topology of edge dislocations requires only a knot embedded on a closed surface in $S^3$, we can use the \emph{exact same} foliation yielding a compact edge dislocation knot to build one that reaches the boundary.  In other words, the only difference between a smectic configuration in which the knot is compact and in which the knot reaches the boundary is where the knot $K$ is embedded on the surface $\Sigma^1\cup\Sigma^2$ in $S^3$.  Hence, Theorem \ref{thmedge} also applies to edge dislocations that reach the boundary: we still need at least $4g$ point defects in $X$ to accommodate the transition from spheres to the union of two Seifert surfaces and back to spheres.  Given a foliation with $4g$ singular points and the surface $\Sigma^1\cup\Sigma^2$, we can thicken $\Sigma^2$ in $X$ and create an edge dislocation.  This edge dislocation follows a path $\alpha$ in $X$ that starts and ends at the boundary and whose shape, if the two ends were connected, forms a knot $K$.  The shape of this dislocation line is fixed outside of $X$ in $\mathbb{R}^3$ to be a straight line to infinity, as per the foliation on $S^3\backslash X$.  The dislocation is pinned not just at the boundary of $X$ but at infinity in $\mathbb{R}^3$; in other words, the ends of the dislocation cannot find each other and form additional crossings outside of $X$, which would require a different foliation on $S^3$ and different number of point defects.  Under this construction, the same considerations for compact dislocations apply to non-compact dislocation lines: just as a compact edge dislocation in the shape of a knot requires $4g$ point defects, so too does an edge dislocation that forms a knot as it traverses through the sample.

\section{Fibredness and Dislocations}\label{sec:fibrednessANDdislocations}

\subsection{Edge Dislocations}\label{sub:notfibredEDGE}

Upon first viewing, the requirement of at least $4g$ point defects alongside knotted edge dislocations seems to have no relation to whether or not the knot is fibred.  Since the construction of knotted edge dislocations relies only on the topology of two Seifert surfaces, which exist for any knot, Theorem \ref{thm:fib_edge} applies to fibred and non-fibred knots alike.  However, fibredness still plays an important role in the topological structure of edge dislocations.  The minimum number of point defects induced by a knotted edge dislocation, $4g$, is only attainable for knots that are fibred.

\ 

\begin{theorem}\label{thm:fib_edge}
    Let $K$ be a fibred knot. Then there is a singular foliation of $X$ satisfying the assumptions from Section~\ref{sec:edge_dislocations} that has an edge dislocation in the shape of $K$ and \underline{exactly} $4g$ critical points, where $g$ is the minimal genus among Seifert surfaces of $K$.
\end{theorem}

\begin{proof}
    The minimal genus $g$ of any fibred knot is realised by its fibre surface $\Sigma$. Gluing two fibres, $\Sigma^1 \equiv f^{-1}(0)$ and $\Sigma^1 \equiv f^{-1}(\pi)$, of the same fibration along their boundary $K$ results in an embedded closed surface of genus $2g$ in $S^3$. The union of the fibres $\Sigma^1\cup\Sigma^2$ splits $S^3$ into two components, $A\equiv S^3\backslash(\Sigma^1\cup \Sigma^2)^+$ and $B\equiv S^3\backslash(\Sigma^1\cup \Sigma^2)^-$, where, using a common orientation of $\Sigma^1\cup\Sigma^2$, $S^3\backslash(\Sigma^1\cup \Sigma^2)^+$ is the side `above' $\Sigma^1\cup\Sigma^2$ and $S^3\backslash(\Sigma^1\cup \Sigma^2)^-$ is the side `below'.  The topology of $A$ and $B$ are fixed by the knot fibration: since $S^3\cong f^{-1}(\theta)\times[0,2\pi]$ with the surfaces $f^{-1}(0)$ and $f^{-1}(2\pi)$ identified, the two components split by a pair of fibres are $A\cong f^{-1}(\theta)\times[0,\pi]$ and $B\cong f^{-1}(\theta)\times[\pi,2\pi]$.  Thus $A$ and $B$ are solid handlebodies of genus $2g$ with shared boundary $\Sigma^1\cup\Sigma^2$.  This is known as a Heegaard splitting: by using two minimal genus Seifert surfaces from a knot fibration, we have decomposed $S^3$ into two handlebodies, $A$ and $B$.  We may thus describe a singular foliation of $S^3$ that contains the surface $\Sigma^1\cup \Sigma^2$ by describing the singular foliations of $A$ and $B$ individually as solid handlebodies of genus $2g$.

    We start with a single point, which will play the role of a global minimum of the corresponding Morse function. As a critical point of index 0 it is surrounded by concentric spheres. Figure~\ref{fig:critical} shows the neighbourhood of a critical point of Morse index 1 or 2. It also shows how a one-handle can be foliated: by adding an index 1 critical point to the concentric spheres around the global minimum, one achieves a solid handlebody of genus 1 and a foliation of its `inside' containing an index 0 and index 1 critical point.  An example of how a sphere grows to a torus by passing through one critical point is shown in Figure~\ref{fig:sing_fol}.  Since a surface of genus $2g$ is obtained from a sphere by attaching $2g$ one-handles, we obtain a singular foliation of the solid handlebody of genus $2g$ that has one critical point surrounded by concentric spheres (the index 0 point) and $2g$ critical points that look as in Figure~\ref{fig:critical} (the index 1 points). So far we have treated the point in the center of the family of concentric spheres as a minimum, whose surrounding spheres grow into a surface of genus $2g$ by passing through $2g$ critical points. Likewise, a surface of genus $2g$ can shrink to concentric spheres, the center of which is a global maximum, by passing through $2g$ critical points.  By starting at an index $0$ critical point and using $2g$ index $1$ points to increase the genus to $2g$, we can grow spherical surfaces into $\Sigma^1\cup\Sigma^2$ and foliate the inside of $B$.  By subsequently decreasing the genus with $2g$ index 2 critical points, we can grow the surfaces above $\Sigma^1\cup\Sigma^2$ into spheres, foliating the inside of $A$.  This provides a singular foliation of $S^3$ that contains a closed surface $\Sigma^1\cup\Sigma^2$ that is a union of Seifert surfaces alongside one global minimum, one global maximum, and exactly $4g$ critical points of index 1 or 2.

    \begin{figure}
        \centering
        \includegraphics[height=5.5cm]{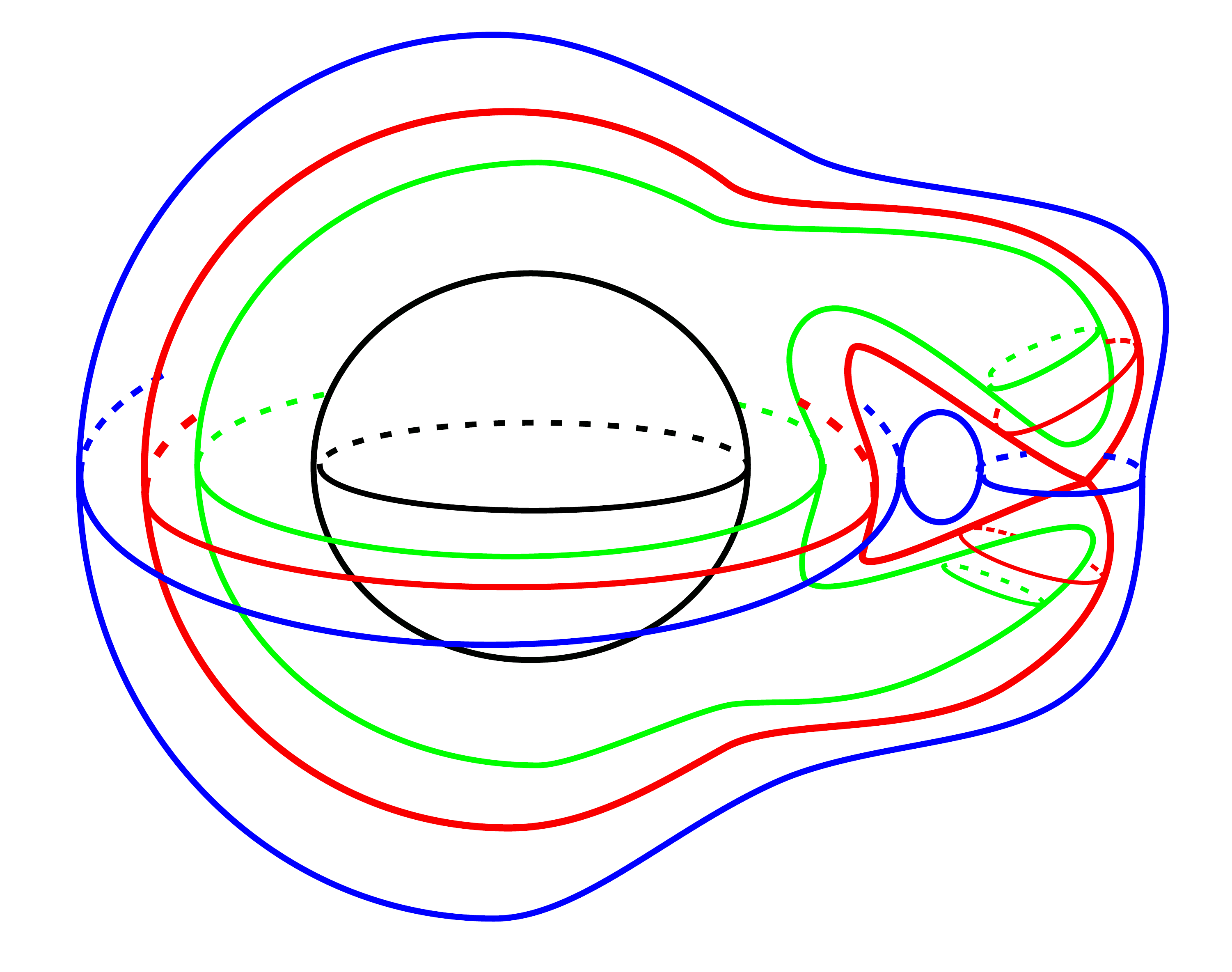}
        \caption{A sphere that grows into a torus by passing through one critical point. Together with a global minimum at the center of the concentric spheres this gives a singular foliation of the solid handlebody of genus 1.  Note the identical topology to that of the Dupin cyclides.}
        \label{fig:sing_fol}
    \end{figure}

    In order to finish this proof we need to turn this singular foliation on $S^3$ into a singular foliation on $X$ that satisfies smectic boundary conditions. There is a path $\gamma$ from the global maximum $p_1$ to the global minimum $p_2$ that intersects every shell exactly once with all of these intersections transverse. A tubular neighbourhood of $\gamma$ is diffeomorphic to $[0,1]\times D^2$ (where $D^2$ is a two-ball). Take such a tubular neighbourhood and shrink $\{0\}\times D^2$ and $\{1\}\times D^2$ to a point each. We call the result $N(\gamma)$. It is an embedded ``American Football", or prolate spheroid. Removing the interior of $N(\gamma)$ from $S^3$ results in a topological 3-ball, whose boundary contains $p_1$ and $p_2$ and is transverse to the leaves of the singular foliation. We can deform the 3-ball to a cube $X$ and grow the points $p_1$ and $p_2$ so that they become the lower and upper face of the cube, $[0,1]^2\times\{0\}$ and $[0,1]^2\times\{1\}$, respectively. Since the leaves are everywhere transverse to the non-horizontal boundary of the cube, we can deform them until they satisfy our boundary condition of being horizontal on $\partial X$.  In other words, we have made a cut in $S^3$ from the global maximum to the global minimum and `laid out' the surfaces punctured by this cut so that they satisfy smectic boundary conditions in $\mathbb{R}^3$.  This projection uses the global maximum and minimum on $S^3$, leaving behind only the $4g$ index 1 and 2 critical points.

    This results in the desired foliation. Note that, other than removing the global maximum and minimum, we have not changed the number of critical points. Since the cut $\gamma$ intersects $\Sigma^1\cup\Sigma^2$, this projection results in a surface of genus $2g$ in $X$ with one puncture.  We thus have a singular foliation of $X$ containing $4g$ critical points and a pair of Seifert surfaces -- one that reaches the boundary (punctured) and one compact -- that are glued along their knot $K$. Using the procedure discussed in Section \ref{sec:edge_dislocations}, we can thicken $\Sigma^2$ (the non-punctured Seifert surface) and modify the foliation locally so that it obtains an edge dislocation in the shape of $K$. This local procedure does not affect the number of critical points, so that there are still exactly $4g$ many.  Thus, for any fibred knot $K$, we have constructed a smectic texture with an edge dislocation in the shape of $K$ and exactly $4g$ critical points.
\end{proof}

Theorem~\ref{thm:fib_edge} only applies to fibred knots, as fibredness allowed the two components of $S^3\backslash\Sigma^1\cup\Sigma^2$ above and below $\Sigma^1\cup\Sigma^2$ to be solid handlebodies -- \emph{i.e.}, yielding a Heegaard splitting from the Seifert surfaces.  This allows the two handlebody components of $S^3$ to be foliated with Morse critical points in a straightforward manner.  Fibredness is required, since two Seifert surfaces of minimal genus glued along their boundary, $\Sigma^1\cup\Sigma^2$, provide a Heegaard splitting of $S^3$ if and only if $\Sigma^1$ and $\Sigma^2$ are equivalent fibre surfaces of a fibred knot -- see for example \cite{godapajitnov}.  Consider attempting a similar process as Theorem \ref{thm:fib_edge} for a non-fibred knot.  Taking again a pair of minimal-genus Seifert surfaces $\Sigma^1\cup\Sigma^2$ glued along their boundary, we have a splitting of $S^3$ into two manifolds $A$ and $B$ above and below $\Sigma^1\cup\Sigma^2$.  One of these manifolds, say $B$, can be made a solid handlebody: one can always take a minimal-genus Seifert surface in $S^3$ and thicken it to an interval of surfaces $\Sigma\times [0,1]$ such that the `inside' of $(\Sigma\times\{0\})\cup (\Sigma\times\{1\})$ is a solid handlebody of genus $2g$.  The `outside' component $A$, however, cannot be completed as an interval of Seifert surfaces $\Sigma\times[1,2\pi]$ since the knot is not fibred.  Instead, for a non-fibred knot with minimal-genus Seifert surfaces $\Sigma^1$ and $\Sigma^2$ of genus $g$, the closed surface $\Sigma^1\cup\Sigma^2$ bounds an embedded solid handlebody $H$ of genus $2g$ on one side, but the complement of $H$ in $S^3$ is not diffeomorphic to a solid handlebody.  In this case, Theorem \ref{thmedge} still applies -- the knotted dislocation must still be accompanied by at least $4g$ critical points, as the dislocation still forces a surface in the foliation of genus at least $2g$.  But, the absolute lower bound of $4g$ point defects cannot be saturated when the knot is not fibred.  Edge dislocations have a clear, topological dependence on fibredness: edge dislocations for knots without fibrations force \emph{extra} point defects beyond the $4g$ minimum achievable by fibred knots, heuristically required to turn the non-handlebody component $S^3\backslash H$ into a solid handlebody.

As was the case of the screw dislocation, we once again ask ourselves what the defect structure must be around an edge dislocation in the shape of a fibred knot.  For edge dislocations, the source of the difference between fibred and non-fibred knots is the failure of a pair of minimal genus Seifert surfaces to provide a Heegaard splitting of $S^3$.  Using this as the starting point, it is possible to characterize the number of additional point defects required for edge dislocations without fibrations.  One approach is to utilize the free genus of a knot.  The free genus $g_F$ of a knot is defined as the minimal genus of a Seifert surface $\Sigma_F$ such that $S^3\backslash \Sigma_F$ is a solid handlebody.  A free genus Seifert surface always provides the requisite Heegaard splitting for an edge dislocation, as thickening $\Sigma_F$ to an interval of surfaces bounded by $\Sigma_F^1\cup\Sigma_F^2$ provides a decomposition of $S^3$ into two handlebodies -- the `inside' and `outside' of $\Sigma_F^1\cup\Sigma_F^2$ -- glued along their boundary.  For fibred knots the free genus equals the minimal genus, $g_F=g$, as the knot fibration provides the handlebody structure.  For non-fibred knots the free genus cannot equal the minimal genus, but a free genus Seifert surface $\Sigma_F$ still exists for some $g_F$.  This does not mean that the knot has a fibration with Seifert surfaces of that genus; even though $S^3\backslash \Sigma_F$ is a handlebody, it is still the case that $S^3\backslash \Sigma_F$ cannot be decomposed into $\Sigma_F\times[0,2\pi]$, which is a stronger requirement on the topology of $S^3\backslash \Sigma_F$.  But, free genus Seifert surfaces allow us to build knotted edge dislocations for fibred and non-fibred knots alike.  By starting with Seifert surfaces of free genus, we may repeat a foliation akin to that of Theorem \ref{thm:fib_edge} with $4g_F$ point defects, as the `inside' and `outside' of $\Sigma_F^1\cup\Sigma_F^2$ are both solid handlebodies of genus $g_F$.  The number of extra point defects required when the edge dislocation is not fibred, then, is $4(g_F-g)$.

Alternatively, we could stick with a pair of minimal-genus Seifert surfaces $\Sigma^1\cup\Sigma^2$ and manually `fix' the topology of the non-handlebody component of $S^3\backslash(\Sigma^1\cup\Sigma^2)$.  Call the side of $S^3\backslash(\Sigma^1\cup\Sigma^2)$ that is a solid handlebody $H$.  For every embedded solid handlebody there is a finite collection of disjoint, simple paths $\gamma_j$, called ``tunnels". These start and end on distinct points on $\partial H$, but otherwise lie in $S^3\backslash H$, and have the property that attaching the one-handles whose cores are the tunnels $\gamma_j$ turns $H$ into an embedded solid handlebody (of greater genus) whose complement is now also a solid handlebody \cite{godapajitnov, murao}.  In our case, where $\partial H:=\Sigma^1\cup\Sigma^2$ for a non-fibred knot, we can deform the tunnels $\gamma_j$ so that they all start and end on the same part of $\partial H$, say, they all start and end on $\Sigma^2$. Attaching one-handles to $\Sigma^2$ results in a new Seifert surface $\tilde{\Sigma}^2$ of genus $g+T$, where $g$ is the initial genus of $\Sigma^1$ and $\Sigma^2$ and $T$ is the number of attached one-handles, \emph{i.e.}, the number of tunnels. Then $\Sigma^1\cup\tilde{\Sigma}^2$ bounds solid handlebodies of genus $2g+T$ on either side.  The total space, \emph{i.e.}, the union of the two sides of $\Sigma^1\cup\tilde{\Sigma}^2$, is still $S^3$.  This does not contradict the non-fibredness of the knot, as the Seifert surfaces $\Sigma^1$ and $\Tilde{\Sigma}^2$ are no longer equivalent -- they have different genus.  This differs from the free genus construction, where both Seifert surfaces `involved' in the dislocation had the same genus, $g_F$.  Now, starting from the surface $\Sigma^1\cup\tilde{\Sigma}^2$, we can foliate the solid handlebodies on either side using $4g+2T$ critical points of index 1 or 2 in the same manner as the proof of Theorem~\ref{thm:fib_edge}.  Since the procedure of turning such foliations into liquid crystalline configurations containing edge dislocations works for any pair of Seifert surfaces -- not just two equivalent surfaces -- we obtain an edge dislocation for any type of knot with exactly $4g+2T$ point defects.  

While these two methods allow one to construct smectic configurations with knotted edge dislocations for knots that are not fibred, it is, in practice, quite difficult to determine the absolute minimum number of extra critical points required.  While the free genus of a knot is a knot invariant, it is difficult to find the appropriate set of tunnels required for the latter construction.  For example, it could also be the case that a pair of Seifert surfaces of non-minimal genus requires fewer tunnels than a pair of minimal genus Seifert surfaces of a given knot.  In other words, the minimum number of tunnels, $T_{\text{min}}$, to turn the complement of the union of Seifert surfaces into a handlebody may require two Seifert surfaces of non-minimal genus $g_{T_{\text{min}}}>g$.  This causes the smectic configuration to contain $4g_{T_{\text{min}}}+2T_{\text{min}}$ point defects.  Since the addition of the tunnels creates a Seifert surface $\Sigma_{T_{\text{max}}}$ of non-minimal genus $g_{T_{\text{max}}}=g_{T_{\text{min}}}+T$ that does provide a Heegaard splitting (\emph{i.e.}, $S^3\backslash \Sigma_{T_{\text{max}}}$ is a solid handlebody), we also have that $g_{T_{\text{max}}}\geq g_F$, since the free genus is, by definition, the minimal genus such that $S^3\backslash \Sigma_F$ is a solid handlebody.  These considerations make the exact minimal number of additional point defects required for non-fibred edge dislocation knots difficult to compute in general.  The construction involving free genus does, however, provide a bridge to screw dislocations and the Morse-Novikov number.


\subsection{Screw Revisited}\label{sub:notfibredSCREW}

We now return, at long last, to screw dislocations.  The connection between screw dislocations and knot fibredness, outlined in Section \ref{sec:screw}, is quite clear: the purely-radial local defect structure of screw dislocations calls for knot fibrations in order to ensure a smooth global structure.  On the other hand, what goes wrong for a screw dislocation in the shape of a knot without fibrations is not obvious and requires highly technical machinery.  Results from Morse-Novikov theory\cite{weber} stipulate that a singular open book decomposition -- \emph{i.e.}, a screw dislocation with additional Morse-type point defects in the system -- can be made possible for any knot by allowing critical points, the minimum number of which is $MN(K)$.  While the topology behind these results is quite intricate, the study of edge dislocations helps us interpret these Morse-Novikov points as defects in smectic liquid crystals.  Knotted edge dislocations require critical points even for fibred knots, for the simple reason that the flat layers of the smectic ground state must change their genus to match the topology of Seifert surfaces.  Furthermore, additional point defects are required for non-fibred knots in order to ensure the edge dislocation provides a Heegaard splitting of $S^3$.  Importantly, these extra point defects are of precisely the same type as the original $4g$.  They also come in pairs of index 1 and index 2 critical points, having the same focal conic-like structure discussed for the original $4g$, and, from the perspective of the foliation of surfaces on $S^3$, are interchangeable with the original point defects.  We can use these basic principles to interpret the role of Morse-Novikov points in the topology of screw dislocations and their structure as smectic defects.  

Consider building a screw dislocation Seifert surface-by-Seifert surface.  Namely, begin with a single Seifert surface, $\Sigma$, of minimal genus and thicken it to an interval of Seifert surfaces $\Sigma\times[0,1]$ of the same knot.  The question of fibredness becomes whether it is possible to continue `growing' Seifert surfaces until the surface $\Sigma\times\{2\pi\}$ returns to $\Sigma\times\{0\}$ and $S^3=\Sigma\times[0,2\pi]$ with $\Sigma\times \{0\}$ and $\Sigma\times \{2\pi\}$ identified by a gluing map (the monodromy) that fixes the boundary.  Taking a step back, the intermediary stage -- having just an interval of Seifert surfaces -- is quite similar to the construction of edge dislocations.  For an edge dislocation, the bounding surface $(\Sigma\times\{0\})\cup(\Sigma\times\{1\})$ splits $S^3$ into two components in need of foliation by surfaces which do not have the knot as their boundary.  But, when the knot is not fibred, the outer component is not a solid handlebody.  At a bare minimum, if a screw dislocation is to be built by continuing to grow an interval of Seifert surface $\Sigma\times[0,1]$, the component $S^3\backslash(\Sigma\times[0,1])$ must itself be a solid handlebody of genus $g$ in order to `fit' the surfaces $\Sigma\times[1,2\pi]$, which necessarily form a solid handlebody.  As such, edge and screw dislocations face a similar topological `frustration' for knots that are not fibred: the failure of Seifert surfaces to split space into pieces with the requisite topology causes the expected foliation of surfaces to be impossible.

    \begin{figure}
        \centering
        \includegraphics[height=10.5cm]{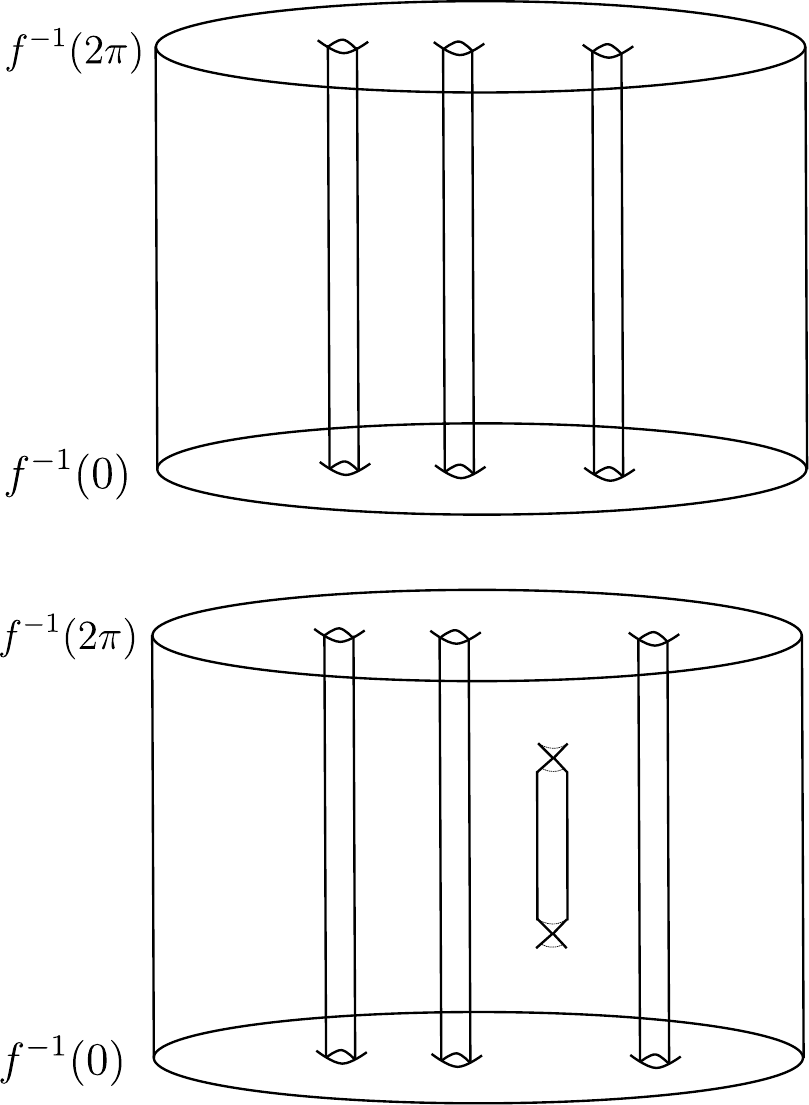}
        \caption{By laying out the smectic layers of a singular open book decomposition on top of each other, we get a representation of screw dislocations for a fibred and non-fibred knots.  The top shows, schematically, a trivial cobordism between Seifert surfaces $f^{-1}(0)$ and $f^{-1}(2\pi)$ representing the knot fibration.  The bottom shows a cobordism between two Seifert surfaces in which the level sets change genus at critical points, representing a singular open book decomposition with $MN(K)=2$ critical points.}
        \label{fig:cobordism}
    \end{figure}

Inspired by knotted edge dislocations, we gain insight into the topological origin of Morse-Novikov points by rephrasing knot fibrations in terms of $S^3$ split by a minimal-genus Seifert surface $\Sigma$ and the corresponding cobordism defined by singular fibrations.  Consider first the topology of the space $S^3\backslash\Sigma$ when the knot is fibred.  The fibration tells us that that $S^3\backslash\Sigma$ is simply $\Sigma\times [0,2\pi]$, where each $\Sigma\times\{t\}$ is an equivalent Seifert surface, \emph{i.e.}, they share the same genus.  We can thus visualize $S^3\backslash\Sigma$ schematically by laying out the $2\pi$ family of Seifert surfaces `on top of each other,' shown in Figure \ref{fig:cobordism}.  From this perspective, the space $S^3\backslash\Sigma$ can be viewed as a cobordism: a three-manifold whose boundaries are two equivalent Seifert surfaces $\Sigma\times\{0\}$ and $\Sigma\times\{2\pi\}$.  For a fibred knot, the cobordism is a trivial one; the `vertical' level sets are all equivalent Seifert surfaces.  Now consider a map $f:S^3\backslash K \rightarrow S^1$ for a non-fibred knot and its associated singular points.  Taking $f^{-1}(0)$ to be a level set containing no critical points and having minimal genus of all level sets of the singular fibration, we can similarly visualize the topology of $S^3\backslash f^{-1}(0)$ by laying out the level sets defined by $f^{-1}(\theta)$.  The space $S^3\backslash f^{-1}(0)$ is once again a cobordism with boundaries $f^{-1}(0)$ and $f^{-1}(2\pi)$ that are equivalent Seifert surfaces.  However, because the knot is not fibred, the cobordism cannot be a trivial one.  Instead, there must be places in the cobordism where the topology of the level sets change.  This change in topology happens precisely at index 1 and index 2 critical points, shown in Figure \ref{fig:cobordism}, where the genus of the Seifert surfaces changes.  These are the so-called Morse-Novikov points, and, from this picture, it is clear that they have exactly the same form as the point defects required for edge dislocations.

The language of cobordisms allows us to explicitly connect back to extra point defects required for non-fibred edge dislocations.  In a singular foliation with minimal (exactly $MN(K)$) singular points, there will be a surface of smallest genus $g_{\text{min}}^{MN(K)}$ and greatest genus $g_{\text{max}}^{MN(K)}$.  The difference in their genera is the number of handles attached to the surface by the critical points, $g_{\text{max}}^{MN(K)}-g_{\text{min}}^{MN(K)}=MN(K)/2$.  The surface of greatest genus $\Sigma_{\text{max}}$ is a free surface, since, by construction, the addition of Morse-Novikov points turns $S^3\backslash\Sigma_{\text{max}}$ into a handlebody.  The free genus of a knot is the minimum possible genus of a Seifert surface such that the complement of that Seifert surface is a solid handlebody, so we have the inequality $g_F\leq g_{\text{max}}^{MN(K)}$\cite{pajitnov}.  As such, we see that the extra point defects required for non-fibred edge dislocations, at least in the free-genus picture, can be related to Morse-Novikov points both in their function -- providing Heegaard splitting -- and their number, $4g_F\leq 4MN(K)$.  Note that the cobordism picture also tells us that the Morse-Novikov number must be even, any change in the topology of level sets must be undone so that the Seifert surfaces $f^{-1}(0)$ and $f^{-1}(2\pi)$ are equivalent.  And so, just like for edge dislocations, the point defect structure of a singular fibration is a focal conic-like structure, containing an equal number of index 1 and 2 critical points that cannot annihilate without changing the knot type.

With this interpretation of Morse-Novikov points, we can finally pin down the topological properties of knotted screw dislocations for all knots, fibred or not.  On $S^3$, any knotted screw dislocation, the topology of which is provided by a (potentially singular) knot fibration, can be achieved with $MN(K)$ point defects.   When the knot is fibred, $MN(K)=0$ and all of $S^3\backslash K$ is smooth and defect free, recovering the construction of Section \ref{sec:screw}.  For non-fibred knots, exactly $MN(K)/2$ of these point defects are index 1 and $MN(K)/2$ are index 2 Morse-critical points.  Recall that, for the edge dislocation, $4g$ point defects are required to change the genus of the Seifert surfaces to match boundary conditions, and $4(g_F-g)$ are required so that $S^3\backslash\Sigma^1\cup\Sigma^2$ provides a Heegaard splitting, where $\Sigma^1\cup\Sigma^2$ are a pair of Seifert surfaces glued along their knot.  The $MN(K)$ point defects required of screw dislocations perform the same function, changing the genus of Seifert surfaces.  These Morse-Novikov points not only provide a Heeggaard splitting of $S^3$ -- coming from the fact that $S^3\backslash\Sigma_{\text{max}}$ is a handlebody -- but they further fix the topology of the corresponding handlebody to allow for a fibration structure.  And so, just as all non-trivially knotted edge dislocation require focal conic-like singular points, screw dislocations for non-fibred knots require $MN(K)$ point defects.  Furthermore, the $MN(K)/2$ index 1 and $MN(K)/2$ index 2 critical points associated to a non-fibred screw dislocation knot cannot annihilate each other without changing the knot type of the defect.  We see here that notions of fibredness underpin the topological properties of smectic dislocations.  Edge and screw defects, understood as two different types of singular foliations on $S^3\backslash K$, can be viewed as two sides of the same coin.  While the different local defect structures and desired foliations lead to a different global structure -- \emph{i.e.}, number of point defects required -- for edge and screw, the two are both classified by knot and Morse theory and sensitive to knot fibration.  

Before moving on from dislocations and applying similar techniques to study the topological properties of knotted smectic disclinations, we briefly return to the method of projection utilized for screw dislocations.  To project a knot fibration on $S^3$ and achieve a smectic configuration in $\mathbb{R}^3$ with smectic boundary conditions, in Section \ref{sec:screw} we used a `poke and unwrap' projection method that resulted in a smectic texture with a screw dislocation inside of a three-ball and a hyperbolic disclination of charge $-1$ on the bounding two-sphere.  For knots without fibration, we can simply use the same projection on a singular fibration with $MN(K)$ singular points in $S^3$ so long as the Morse-Novikov points do not sit at the projection point in $S^3$.  This results in a smectic texture in $\mathbb{R}^3$ with a knotted screw dislocation (fibred or non-fibred), $MN(K)$ point defects, and a hyperbolic disclination loop.  However, note that, as discussed in Section \ref{sec:screw}, this extra disclination line can be shrunk to a point defect, shown in Figure \ref{Fig:Screw_Projection_LoopPoint}.  At the same time, we have now developed the machinery of Morse-critical points as they relate to smectic defects, and have introduced a seemingly different projection method for edge dislocations.  The projection for edge dislocations makes a cut between the global maxima and minima on $S^3$, making use of the index 0 and 3 critical points that were already required for the edge dislocation's foliation on $S^3$.  Although the foliation of surfaces for screw dislocations on $S^3$ require no such critical points (only index 1 and 2 points when the knot is not fibred) on $S^3$, we can always add an extraneous pair of neighboring Morse-index critical points to make this projection method work.

    \begin{figure}
        \centering
        \includegraphics[height=8.5cm]{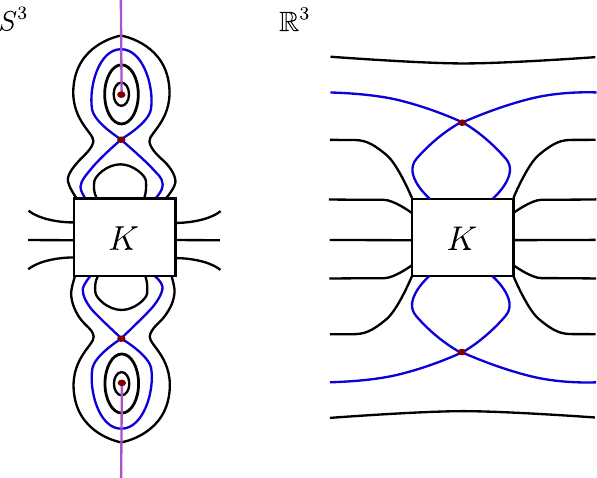}
        \caption{Projecting a screw dislocation on $S^3$ to a smectic configuration on $\mathbb{R}^3$ using point defects.  On $S^3$ (left), two pairs of critical points are added to a single Seifert surface (blue) of a knot fibration.  By making a cut between the resulting maxima and minima (purple) passing transversely through each surface once, the punctured surfaces can be `unwrapped' to match smectic boundary conditions in $\mathbb{R}^3$.  Note that the surfaces of the fibration are only shown in cross section far from the knot; locally the surfaces have the expected open book structure along the knotted defect.  In $\mathbb{R}^3$, this means that each surface comes in from the boundary and `enters' the knot once -- for example, the blue lines represent a single surface with the same knot as its boundary as the other (punctured) Seifert surfaces.}
        \label{fig:ScrewPtProjection}
    \end{figure}

While they are not forced by the knottedness of the defect, consider adding a pair of index 0 and 1 critical points and a pair of index 2 and 3 critical points to a knot fibration.  Acting as the maximum and minimum on $S^3$, we can use the index 0 and 3 critical points to project to the unit cube.  As in the projection of the edge dislocation in Theorem \ref{thm:fib_edge}, make a cut $\gamma$ between the maximum and minimum on $S^3$ that is transverse to all of the surfaces it intersects.  We can always deform the cut and surfaces of the knot fibration such that each level set is intersected by $\gamma$ once.  After removing the cut and deforming the surfaces to sit on the unit cube $X$, we are left with a configuration with smectic boundary conditions on $\partial X$, a knotted screw dislocation in the shape of $K$, and two point defects (one of Morse index 1 and one of index 2).  Under this projection, which is shown in Figure \ref{fig:ScrewPtProjection}, each Seifert surface of the knot fibration is punctured, meaning each smectic layer extends to the boundary of $X$.  The resulting configuration is more akin to the straight-line screw dislocation: a $2\pi$ family of surfaces that are flat at infinity share the knot as their common boundary.  Unlike the straight-line screw dislocation's family of helicoids, here we require two point defects to fill smectic layers `above' and `below' the $2\pi$ Seifert surfaces in $\mathbb{R}^3$.  Finally, note that, classified according to Morse theory, this projection and the `poke and unwrap' projection of Section \ref{sec:screw} are equivalent.  Shrinking the hyperbolic disclination to a point results in a point defect that is not Morse, since it is degenerate.  However, it can be split into a pair of Morse points by separating the hyperbolic regions along a layer.  These are the same index 1 and index 2 critical points present in this projection; by moving one of the critical points through the Seifert surfaces within the three-ball, one recovers this second projection wherein all smectic layers reach the boundary of the sample\footnote{The critical points from this projection change the number of connected components of the level sets rather than the genus, allowing the Seifert surfaces to `connect' with the flat layers outside the three-ball.}.  By using a unified projection method for both edge and screw, we have arrived at a final Morse-theoretic classification of knotted screw dislocations in $\mathbb{R}^3$: they require at least $2+MN(K)$ point defects, another focal conic-like structure of half index 1 and half index 2 critical points.  As mentioned for edge dislocations, the required Morse points may rearrange themselves in a manner compatible with the smectic liquid crystal, for example by forming a disclination loop as per the original projection method of Section \ref{sec:screw}.

\section{Knotted Disclinations}\label{sec:+12}

\subsection{$+1/2$ Disclinations}

The previous sections have outlined how tying screw dislocations into non-fibred knots and edge dislocations into any non-trivial knot leads to the presence of focal conic-like point defects in the system.  This analogy spawned from the point defect structure of the Dupin cyclides, which similarly contained an even number of index 1 and index 2 Morse critical points.  In experiments, the smectic layers of a focal conic domain tend to form not as the entirety of the Dupin cyclides but only as part of the cyclides, often the negative Gaussian curvature part of the surfaces.  Taking the negative Gaussian curvature portion of the symmetric cyclides, shown in Figure \ref{fig:cyclides}, results in a smectic configuration with a $+1/2$ disclination loop and a single point defect.  While the second point defect of the cyclides is not apparent in this picture, one can think of the disclination loop itself as having point charge and forcing the central Morse point into the system.  By similar means of the construction of edge dislocations, we can both formalize this statement using Morse theory and analyze the topological properties of knotted $+1/2$ disclinations.  This acts as a direct topological generalization of focal conic domains.  While knotted $+1/2$ disclinations cannot share the geometric properties of the Dupin cyclides\footnote{When the disclination is knotted, both lines of curvature cannot be circles.}, a characterization of their topological properties serves to illuminate what would happen if the characteristic focal conic domain ellipses observed in experiment were, instead, knots.  This furthers the analogy between focal conic domains and knotted dislocations: not only are additional point defects required due to the knottedness of $+1/2$ disclinations, but they are explicitly on par with the original focal conic domain point defects.

At first glance it seems as though the methods of Section \ref{sec:edge_dislocations} will not apply to disclinations.  The analysis of edge dislocations was aided by the necessity of exactly two Seifert surfaces: gluing the surfaces along their boundary resulted in a smooth, boundary-free surface.  On the other hand, $+1/2$ disclinations introduce both a surface with boundary and non-orientability in the director field (as in the straight-line disclination of Figure \ref{fig:Smectic_Examples}).  Fortuitously, we can still apply simple Morse theory arguments to knotted $+1/2$ disclinations with the following construction, shown in in Figure \ref{Ex_Disc_Knot}. Assume that the surface bounded by the disclination is orientable and hence a Seifert surface of the knot, say of genus $g$. Then the surrounding surfaces are closed surfaces of genus $2g$, as indicated in the bottom left of Figure~\ref{Ex_Disc_Knot}. Such a local picture can be obtained from any given Seifert surface $\Sigma_0$ for a knot $K$ by first thickening it to an interval $\Sigma_{t}$ with $t\in[-\pi,\pi]$.  Then take the surfaces $\Sigma_{\pi},\Sigma_{-\pi}$ and push them perpendicular to the knot opposite $\Sigma_0$.  Gluing the push-offs $\Tilde{\Sigma}_\pi$ and $\Tilde{\Sigma}_{-\pi}$ results in a closed surface $\Tilde{\Sigma}_{\pi}\cup\Tilde{\Sigma}_{-\pi}$ of genus $2g$, and the remaining surfaces $\Sigma_t$ sit inside.  Repeating this gluing process for all $\Sigma_t$ results in a foliation of the inside of $\Tilde{\Sigma}_{\pi}\cup\Tilde{\Sigma}_{-\pi}$ containing a knotted $+1/2$ disclination.  Importantly, this construction avoids both the non-co-orientability and surface boundaries of the $+1/2$ disclination: the disclination itself foliates the inside of a boundary-free surface of genus, and we need only worry about foliating the outside.

\begin{figure}
    \centering
    \includegraphics[scale=2.5]{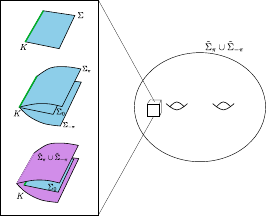}
    \caption{\label{Ex_Disc_Knot}{Construction of a knotted $+1/2$ disclination from the thickening of a Seifert surface $\Sigma$.  Note that, once the Seifert surfaces $\Sigma_\pi,\Sigma_{-\pi}$ are pushed off the knot and glued to form $\Tilde{\Sigma}_\pi\cup\Tilde{\Sigma}_{-\pi}$, the $+1/2$ disclination and its nearby surfaces fill the inside of $\Tilde{\Sigma}_\pi\cup\Tilde{\Sigma}_{-\pi}$, a surface of genus at least $2g$.  This differs from a knotted edge dislocation: here the smectic defect fills a surface of $2g$ whereas the edge dislocation acts only as the surface itself. 
 Thus we need only foliate the outside of $\Tilde{\Sigma}_\pi\cup\Tilde{\Sigma}_{-\pi}$, shown on the right, in $S^3$.}}
    \label{Ex_Disc_Knot}
\end{figure}

As in the Section \ref{sec:edge_dislocations} we only consider foliations that satisfy certain assumptions that allow us to apply Morse theory.  Again, we will only consider foliations with bounded domain, \emph{i.e.}, in the cube $X$ with the usual boundary condition. Furthermore, we want a configuration in which the $+1/2$ disclination is the only line defect and that the surface that is bounded by the knotted disclination is orientable. Recall that $\Tilde{\Sigma}_{\pi}\cup\Tilde{\Sigma}_{-\pi}$ is the closed surface of genus $2g$ that encloses the $+1/2$ disclination. We have described the foliation inside of the solid handlebody $H$ that is bounded by $\Tilde{\Sigma}_{\pi}\cup\Tilde{\Sigma}_{-\pi}$ and that contains the disclination.  Foliations containing a $+1/2$ disclination, such as that of $H$ itself, are not level sets of a circle-valued Morse function, due to the presence of a singular line $K$.  However, we assume that outside of $\Tilde{\Sigma}_{\pi}\cup\Tilde{\Sigma}_{-\pi}$ the leaves of the foliation are given by the level sets of a circle-valued Morse function. By the same arguments as in the previous section we can work with a real-valued Morse function on $S^3\backslash H$.

\

\begin{theorem}\label{thmdisc}
    Under the above assumptions, knotted smectic $+1/2$ disclinations require at least $2g+1$ point singularities, where $g$ is the minimum genus for a Seifert surface of the given knot.
\end{theorem}

\begin{proof} 
Once again we apply Morse theory to the surfaces around a knotted $+1/2$ disclination. Consider the embedding of the cube $X$ in $S^3$ of Figure \ref{Cube_S3}, filling the complement of $X$ with attached disks $D_t$, $t\in[0,1]$ and 3-balls $B_1$ and $B_2$ that are foliated by concentric spheres whose center points are $p_1\in B_1$ and $p_2\in B_2$. As in Section \ref{sub:MorseAndProjection}, we can describe the foliation of $S^3\backslash H$ via a real-valued Morse function $\tilde{F}$.

While the foliations on $S^3$ are identical, there are two cases to consider in $X$.  Suppose that the surface $\Sigma_0$ that is bounded by the $+1/2$ disclination intersects the boundary $\partial X$ in an even number of connected components. Then there is a path from $p_1$ to $p_2$ that is always transverse to the leaves of the foliation and intersects $\Sigma_0$ an even number of times. It follows that one of the points, $p_1$ or $p_2$, is a local minimum and the other, $p_2$ or $p_1$, is a local maximum. After a homotopy that does not change the level sets we can take the local maximum, say $p_1$, to be a global maximum, and $p_2$ to be a global minimum. However, $\tilde{F}$ also has a connected component of a level set that is $\Tilde{\Sigma}_{\pi}\cup\Tilde{\Sigma}_{-\pi}$. We can homotope $\tilde{F}$ without changing the foliation until $y:=\tilde{F}(p_2)=\tilde{F}(\Tilde{\Sigma}_{\pi}\cup\Tilde{\Sigma}_{-\pi})$, that is, both $p_2$ and $\Tilde{\Sigma}_{\pi}\cup\Tilde{\Sigma}_{-\pi}$ attain the same minimal value. The level set $\tilde{F}^{-1}(y+\varepsilon)$ with $\varepsilon>0$ small then consists of the disjoint union of a sphere and a closed surface of genus $2g$. The rest of the argument is again a simple consequence of the Morse lemma. In order to go from a level set of this topology to a sphere that lies around the global maximum $p_1$, we need at least $2g+1$ critical points of index 1 or 2. An index 1 critical point is necessary to connect the two connected components of $\tilde{F}^{-1}(y)$ and the remaining $2g$ critical points of index 2 reduce the genus from $2g$ to 0.

The second case is when $\Sigma_0$ intersects $\partial X$ in an odd number of connected components. Then there is a path from $p_1$ to $p_2$ that is transverse to the leaves of the foliation and that intersects $\Sigma_0$ an odd number of times. It follows that $p_1$ and $p_2$ are both local minima or both local maxima of $\tilde{F}$. Again we can homotope $\tilde{F}$ until $p_1$ and $p_2$ have the same image and are global minima (or maxima) and $\tilde{F}(\Tilde{\Sigma}_{\pi}\cup\Tilde{\Sigma}_{-\pi})$ is the global maximum value (or minimum) of $\tilde{F}$. Level sets close to the global minimum are the disjoint union of two two-spheres, so that we need again $2g+1$ critical points of index 1 or 2 in order to turn them into closed surface of genus $2g$. One of them is necessary to reduce the number of connected components and the remaining $2g$ critical points increase the genus from $0$ to $2g$.
\end{proof}

\begin{figure}
    \centering
    \includegraphics[scale=.6]{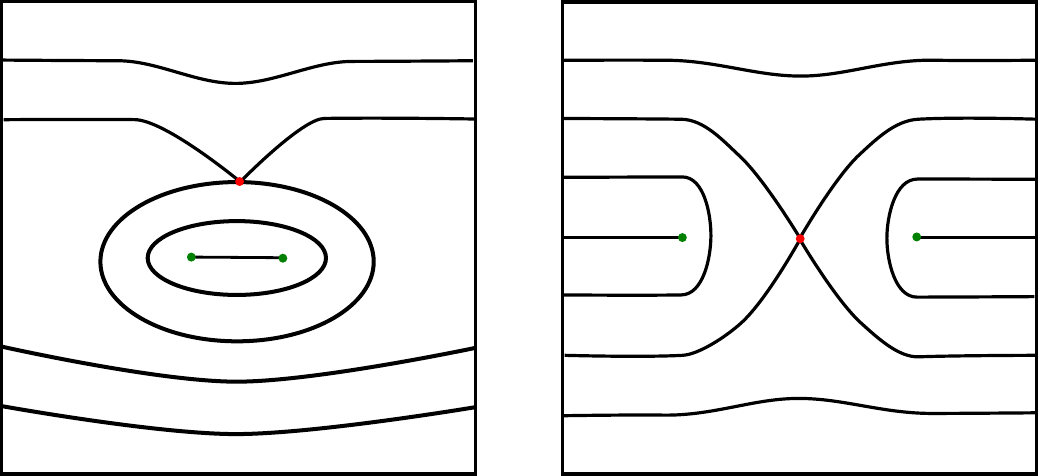}
    \caption{A cross-sectional view of the smectic foliation for unknotted $+1/2$ disclination loops.  The case in which the surface $\Sigma_0$ with the disclination (green) as its boundary does not intersect the boundary of $X$ is shown on the left, and the case in which $\Sigma_0$ reaches the boundary is shown on the right.  The first case is a point defect extended to a loop, and the latter case is the geometry of a focal conic domain.  In both cases, a critical point (red) is required to cancel the disclination's point charge in order to satisfy boundary conditions.  Additional point defects are required to change genus when the disclination is knotted.}
    \label{Fig_Disc_Proof}
\end{figure}

The two cases presented in the proof are distinguished by the parity of the number of connected components in the intersection of $\Sigma_0$ and $\partial X$. Since it was one of our assumptions that $K$ is the only line defect and that outside of $\Tilde{\Sigma}_{\pi}\cup\Tilde{\Sigma}_{-\pi}$ the foliation is given by a Morse function, it follows that this number of connected components is at most 1. Therefore, the two cases simply correspond to whether $\Sigma_0$ intersects $\partial X$ or not.  As an illustrative example, consider again the unknot.  In the case that $\Sigma_0$ does not intersect the boundary, the surfaces $\Tilde{\Sigma}_{\pi}\cup\Tilde{\Sigma}_{-\pi}$ and $\Tilde{\Sigma}_{\delta}\cup\Tilde{\Sigma}_{-\delta}$ are spheres inside of $X$ surrounding $\Sigma_0$ and the unknot.  Noting that defect loops in nematic and smectic liquid crystals can act as point charge\cite{Kamien_RevModPhys_Defects}, one recognizes that this $+1/2$ disclination is nothing but a point defect (at the center of concentric spheres) extended to a loop.  To satisfy boundary conditions, this point charge must be canceled by another point defect in the system.  This is achieved by the index 1 critical point, which, on $S^3$, originates from the need to reduce the number of connected components of the level set $\tilde{F}^{-1}(y+\varepsilon)$ where $y:=\tilde{F}(p_2)=\tilde{F}(\Tilde{\Sigma}_{\pi}\cup\Tilde{\Sigma}_{-\pi})$.  In $\mathbb{R}^3$ the critical point connects two disjoint flat layers so that they form a single sphere around the disclination, shown in the left of Figure \ref{Fig_Disc_Proof}.  The other case is where $\Sigma_0$, and thus $\Tilde{\Sigma}_{\pi}\cup\Tilde{\Sigma}_{-\pi}$, reach the boundary.  Since the boundary condition forces the surfaces above and below the unknot and $\Tilde{\Sigma}_{\pi}\cup\Tilde{\Sigma}_{-\pi}$ in $\mathbb{R}^3$ to be flat sheets, there is a `hole' left inside of the unknot that must be foliated with surfaces.  The only way to achieve this is to allow one sheet above to attach to one sheet below at a pinch, filling the center of the unknot.  This index 1 critical point connects two flat sheets into one connected component that wraps around the disclination line; on $S^3$ this corresponds to reducing the connected component of the level sets containing the spheres surrounding $p_1,p_2$.  Foliations of $X$ for these configurations are sketched on in Figure \ref{Fig_Disc_Proof}.  This latter case corresponds to the $+1/2$ disclination loop geometry of a focal conic domain.  Note that, on $S^3$, there is no difference between the focal conic domain disclination and the disclination loop acting as the core of a concentric-spheres point defect.  Both disclination loops carry point charge that must be canceled by a point defect, and the only difference between the two is the projection to $\mathbb{R}^3$.

Not only is this result interesting for knotted $+1/2$ disclinations themselves, but it also brings further context to the focal conic domain analogy used to described the point defects associated with knotted edge and screw dislocations.  In smectic texture containing a knotted $+1/2$ disclination, there must be at least $2g$ point defects associated to the knottedness and the original `focal conic domain' defect.  But, from the perspective of the Morse points, these point defects are not only identical but \emph{interchangeable}!  The requirement of extra point defects to knot edge dislocations, screw dislocations, and $+1/2$ disclinations is, truly, a topological extension of the well-studied focal conic domain.

Knotted $+1/2$ disclinations are also sensitive to knot fibrations, in an analogous manner to edge dislocations.  Once again, achieving the lower bound of $2g+1$ point defects requires a Heegaard splitting of $S^3$, and a minimal-genus Seifert surface only provides a Heegaard splitting if the knot is fibred.

\

\begin{corollary}\label{thm:fib_12}
    Let $K$ be a fibred knot. Then there is a configuration of a smectic liquid crystal that satisfies the assumptions above, that has a $+1/2$ disclination in the shape of $K$ and exactly $2g+1$ point defects, where $g$ is the minimal genus of the knot.
\end{corollary}

\begin{proof}
    We will present two constructions. One of them results in a $+1/2$ disclination that bounds a surface that extends to the boundary $\partial X$ and the other construction gives a surface that does not extend to $\partial X$.

    We start with the disjoint union of a sphere and a closed surface $\Tilde{\Sigma}_{\pi}\cup\Tilde{\Sigma}_{-\pi}$ of genus $2g$, the result of gluing two minimal Seifert surfaces of $K$ along their boundary, all of which lie in $S^3$. The inside of the sphere (the side which does not contain $\Tilde{\Sigma}_{\pi}\cup\Tilde{\Sigma}_{-\pi}$) is foliated by concentric spheres and a global minimum $p_2$. The inside of the surface $\Tilde{\Sigma}_{\pi}\cup\Tilde{\Sigma}_{-\pi}$ (the side which does not contain the sphere), is foliated by the $+1/2$ disclination and the surrounding surfaces of genus $2g$. We can join the sphere and the surface with an attaching one-handle, which can be foliated as in Theorem \ref{thm:fib_edge} using exactly one critical point. This gives a singular foliation of an embedded solid handlebody $H$ with global minimum, one critical point and one knotted $+1/2$ disclination. Note that, since the disclination loop and surrounding surfaces fill $\Tilde{\Sigma}_{\pi}\cup\Tilde{\Sigma}_{-\pi}$, it is not necessary to have more critical points on the inside. The rest of the proof is identical to the one of Theorem~\ref{thm:fib_edge}. Since $K$ is fibred, the complement of $H$ in $S^3$ is a solid handlebody, which can be foliated using the same foliation as in Theorem~\ref{thm:fib_edge}. It thus contains $2g$ critical points of index 2 and one global maximum $p_1$. There is again a path $\gamma$ from $p_1$ to $p_2$ that is transverse to the leaves of the foliation and using the same ``American Football" method as in Theorem~\ref{thm:fib_edge} to cut and project from $S^3$ to $X$ with smectic boundary conditions.  We thus obtain the desired foliation of $X$ with $2g+1$ critical points.  Note that $\gamma$ does not intersect $H$ and thus does not intersect $\Sigma_0$. Therefore, $\Sigma_0$ does not extend to $\partial X$.

    The second construction is similar and goes as follows. Take two disjoint two-spheres, whose interiors are foliated by concentric spheres and a local extremum each. We can join the two spheres via one attached one-handle, so that we have a foliation of a 3-ball $B$ that involves exactly one critical point of index 1 or 2. The complement of the solid handlebody $H$ bounded by $\Tilde{\Sigma}_{\pi}\cup\Tilde{\Sigma}_{-\pi}$ is again a solid handlbody, because $K$ is fibred. We can foliate this complementary solid handlebody minus $B$ as in the previous section. Since $H$ has genus $2g$ this requires $2g$ critical points of index 1 or 2. There is now a path $\gamma$ that goes from $p_1$ to $\Sigma_0$ and then from $\Sigma_0$ to $p_2$, always transverse to the leaves of the foliation. This is possible because $p_1$ and $p_2$ are global minima (or maxima) and the inside of $\Tilde{\Sigma}_{\pi}\cup\Tilde{\Sigma}_{-\pi}$ plays the role of the global maximum (or minimum) -- the cut passes through $\Tilde{\Sigma}_{\pi}\cup\Tilde{\Sigma}_{-\pi}$, then $\Sigma_0$, then back through $\Tilde{\Sigma}_{\pi}\cup\Tilde{\Sigma}_{-\pi}$ on its way from $p_1$ to $p_2$. By the same method as in the proof of Theorem~\ref{thm:fib_edge} we obtain the desired foliation on $X$ by removing an appropriate neighbourhood of $\gamma$ and applying a homotopy. Since $\gamma$ intersects $\Sigma_0$, the surface that is bounded by the disclination line, the surface $\Sigma_0$ becomes a punctured Seifert surface when $\gamma$ is removed. It thus extends to $\partial X$.

    In particular, for both cases -- corresponding to whether the surface $\Sigma_0$ that is bounded by the $+1/2$ disclination extends to $\partial X$ or not --  we find a foliation with exactly $2g+1$ point defects, the minimal possible number of such points by Theorem~\ref{thmdisc}.  As with edge dislocations, Theorem \ref{thmdisc} and Corollary \ref{thm:fib_12} only apply to knots.  Once again, the only modification for links is that the Euler character, not genus, it additive for $\Tilde{\Sigma}_{\pi}\cup\Tilde{\Sigma}_{-\pi}$, leading to $2g+(b-1)+1$ point defects instead.
\end{proof}

In a similar fashion as edge dislocations, we can characterize the extra point defects required for non-fibred knots in the case of $+1/2$ disclinations.  As in Section \ref{sub:notfibredEDGE}, we need both the `inside' and `outside' of $\Tilde{\Sigma}_{\pi}\cup\Tilde{\Sigma}_{-\pi}$ to be solid handlebodies.  One way to achieve this is to use, from the start, Seifert surfaces of free genus $g_F$ rather than minimal genus $g$.  Since the inside of $\Tilde{\Sigma}_{\pi}\cup\Tilde{\Sigma}_{-\pi}$ is already filled with the knotted $+1/2$ disclination -- whose Seifert surface is now simply of non-minimal genus -- this method only increases the number of point defects for non-fibred knots to $2g_F+1$.  On the other hand, we could start with a minimal-genus Seifert surface and transform $\Tilde{\Sigma}_{\pi}\cup\Tilde{\Sigma}_{-\pi}$ by attaching one-handles so that it bounds a solid handlebody on either side.  Attaching a one-handle adds an additional two Morse points: one critical point to increase the genus of surfaces inside of $\Tilde{\Sigma}_{\pi}\cup\Tilde{\Sigma}_{-\pi}$ and a second for the increased genus of the now-handlebodied complement $S^3\backslash H$.  This results in $2g+2T+1$ point defects.   As in the case of edge dislocations, it is not straightforward to determine which method will result in a smaller number of point defects.

We briefly note that one of our assumptions in the previous was that the surface that is bounded by the $+1/2$ disclination is orientable. While this is an essential assumption in the results of Theorems \ref{thm:fib_12} and \ref{thmdisc}, the assumption of orientability is not necessary to make the local picture of knotted $+1/2$ disclinations, i.e, that of Figure~\ref{Ex_Disc_Knot}, work. The following example shows that, by allowing non-orientable surfaces, we can beat the lower bound on the number of critical points established in Theorem~\ref{thmdisc}.

\

\begin{example}\label{example1}
    Consider the trefoil knot that bounds an annulus with 3 half-twists. This surface is not orientable and therefore Theorem~\ref{thmdisc} does not apply. Surfaces around this annulus are tori, surfaces of genus 1, shown in Figure~\ref{fig:unoriented}. We thus need only two point defects to realise this structure as a $+1/2$ disclination: one originating from the attached one-handle that joins the torus to a sphere surrounding the global minimum $p_1$ in $S^3$, and one to decrease the genus of the level sets, so that the boundary condition on $\partial X$ is satisfied. Therefore, two critical points suffice, while for a $+1/2$ disclination in the shape of a trefoil with an orientable surface we would have needed $2g+1=3$ critical points.
\end{example}

\begin{figure}
    \centering
    \includegraphics[height=5cm]{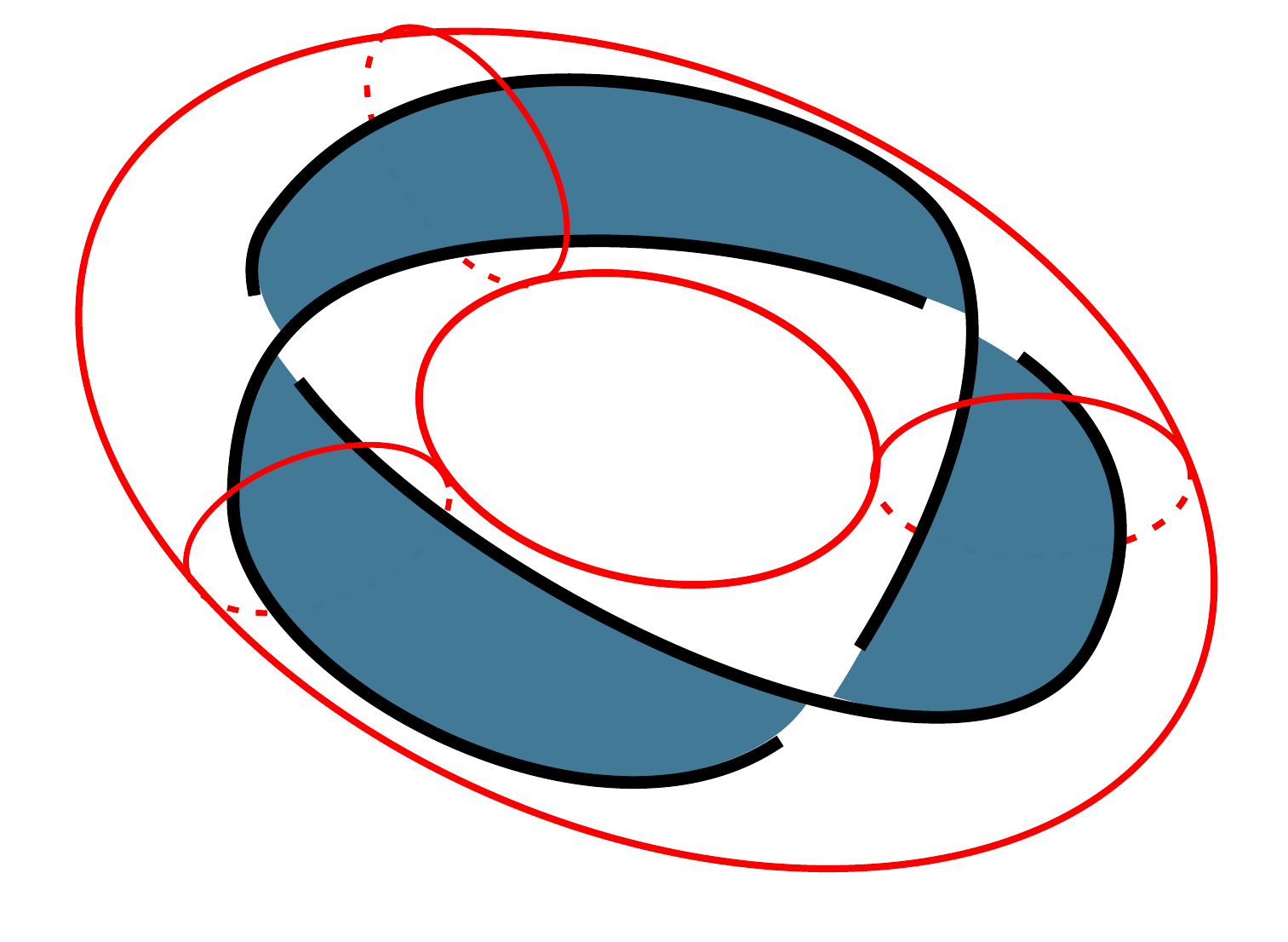}
    \caption{A non-orientable surface bounded by a trefoil knot and a torus that surrounds it, forming a $+1/2$ disclination in the shape of a trefoil knot inside of a handlebody of genus one.}
    \label{fig:unoriented}
\end{figure}

Via a slight modification, we can account for non-orientable surfaces bounding the $+1/2$ disclination (see \cite{kindred} for a discussion of non-orientable surface bounded by knots). Earlier, for oriented surfaces, we created a closed surface, which is the union of two Seifert surfaces $\Tilde{\Sigma}_\pi\cup\Tilde{\Sigma}_{-\pi}$ glued along the knot, by pushing the surface in the positive and negative normal direction, respectively. If the surface $\Sigma_0$ is non-orientable, there is by definition no well-defined global notion of positive and negative normal direction of the surface. However, locally a positive and negative normal direction can be defined. Non-orientability only means that these local choices of sign cannot be all consistent with each other.

It follows that we can push off the surface at every point in both normal directions, positive and negative. For the construction it does not matter which direction is considered positive and which negative. The results is again a closed, oriented surface $\Sigma$ that contains the knot. The only difference to the non-orientable case, where $\Sigma$ was $\Tilde{\Sigma}_\pi\cup\Tilde{\Sigma}_{-\pi}$, is that $\Sigma\backslash K$ is a connected surface, while $(\Tilde{\Sigma}_\pi\cup\Tilde{\Sigma}_{-\pi})\backslash K$ was not connected. It consisted of the disjoint interiors of two Seifert surfaces. We can then push $\Sigma$ away from the knot in the opposite direction of $\Sigma_0$ (as in the oriented case). We thus have again a closed, oriented surface, whose inside is foliated by a knotted $+1/2$ disclination bounding an unoriented surface $\Sigma_0$ and concentric surfaces around $\Sigma_0$. In order to fill the outside of the surface, we again require $g(\Sigma)+1$ critical points of index 1 or 2.

As in the orientable case, we can describe the topology of $\Sigma$, \emph{i.e.}, its genus $g(\Sigma)$, in terms of the original surface $\Sigma_0$. However, since $\Sigma_0$ is not orientable, we should use the Euler characteristic instead of the genus. Note that $\Sigma\backslash K$ is a double cover of $\Sigma_0\backslash K$. Therefore, its Euler characteristic $\chi(\Sigma\backslash K)$ is twice the Euler characteristic of $\Sigma_0\backslash K$. Since $K$ is topologically a circle, which has Euler characteristic 0, we have $\chi(\Sigma)=2\chi(\Sigma_0)$. Since $\Sigma$ is closed and oriented, its genus is then $g(\Sigma)=1-\chi(\Sigma)/2=1-\chi(\Sigma_0)$. It follows that any $+1/2$ disclination bounding an unoriented surface $\Sigma_0$ requires $g(\Sigma)+1=2-\chi(\Sigma_0)$ point defects. Note that this statement generalizes the result on the oriented case, where $g(\Sigma)=2g(\Sigma_0)$ and $\chi(\Sigma_0)=1-2g(\Sigma_0)$.

In particular, for the $(2,k)$-torus knots with $k$ odd, we can take $\Sigma_0$ to be an (non-orientable) annulus with $k$ half-twists. We have $\chi(\Sigma_0)=0$ and thus $g(\Sigma)=1$ as shown in Example \ref{example1} and Figure \ref{fig:unoriented}.  We note that this construction in terms of the double cover of the unoriented surface bears resemblance to the topological classification of nematic textures around knotted disclinations, which are classified according to the homology of the double cover of the knot complement\cite{Machon_Alexander_Knots_1}.

\subsection{Negative-Charge Disclinations}

\begin{figure}
    \centering
    \includegraphics[height=7cm]{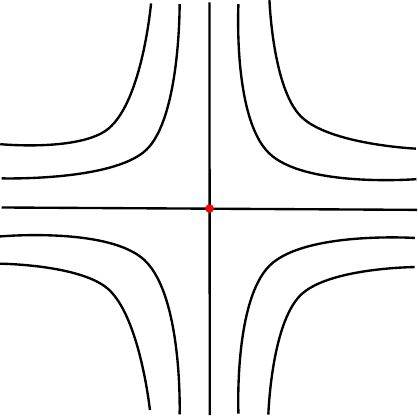}
    \caption{Example of a charge $-1$ disclination.  For such a configuration to occur along a knot, the four separatrix layers must be Seifert surfaces splitting $S^3$ into four components.}
    \label{fig:NEGATIVEdisclination}
\end{figure}

Characterizing the topology of knotted dislocations and $+1/2$ disclinations required studying how the smectic layers associated with the defects split space into handlebodies in need of foliation by smectic layers.  Using this approach, we can directly analyze all negative-charge disclinations.  Consider the local structure of a hyperbolic disclination, shown in Figure \ref{fig:NEGATIVEdisclination}.  The smectic texture around a disclination of charge $-|q|$ contains $2|q|+2$ regions of hyperbolic layers separated by $2|q|+2$ `separatrix' layers whose common boundary is the disclination.  Imagine that the surface structure of a charge $-|q|$ disclination occurs along a knot.  We once again assume that all the smectic layers are orientable, and we also stipulate that there is only integer twisting of the separatrix layers and hyperbolic regions.  In other words, each smectic layer returns to itself after traversing the entire knot rather than `mixing' with other, symmetric regions of the disclination.  This implies that the separatrices are $2|q|+2$ Seifert surfaces for the knot splitting $S^3$ into $2|q|+2$ components.  Filling in each component with smectic layers requires, once again, point defects.  

\ 

\begin{theorem}\label{thmnegdisc}
    Under these assumptions and the assumptions of Theorem \ref{thmdisc}, knotted smectic disclinations of charge $-|q|$ require at least $2g(2|q|+2)$ index 1 or 2 point singularities and $2|q|$ index 0 or 3 point singularities, where $g$ is the minimum genus of a Seifert surface for the given knot.  
\end{theorem}

\begin{proof}
    This follows immediately from the foliation of each of the $2|q|+2$ components in $S^3$.  The results from Theorem \ref{thmedge} on edge dislocations show we require at least $2g$ genus-changing index 1 or 2 critical points and one maximum or minimum to foliate a component of $S^3$ bounded by two Seifert surfaces.  Adding up the $2|q|+2$ components leads to a singular foliation of $S^3$ with a knotted charge $-|q|$ disclination alongside at least $2g(2|q|+2)$ saddles and $2|q|+2$ maxima or minima.  Recall that projection of such a singular foliation on $S^3$ to a configuration on $X$ (a cube in $\mathbb{R}^3$) with smectic boundary conditions requires isolating a maximum and minimum as the north and south poles on $S^3$ and cutting between them.  This reduces the number of index 0 or 3 critical points in $X$ by $2$, leaving behind $2g(2|q|+2)$ index 1 or 2 critical points and $2|q|$ index 0 or 3 critical points.  Thus, a smectic configuration in $\mathbb{R}^3$ with a knotted disclination of charge $-|q|$ requires at least $2g(2|q|+2)$ saddle point defects and $2|q|$ point defects of the form of concentric spheres. 
\end{proof}

Once again, we find that knotted disclinations of negative charge also have focal conic-like point defects, \emph{i.e.}, an even number of index 1 and index 2 Morse points.  Although, unlike $+1/2$ disclinations and dislocations, knotted negative-charge disclinatinos also require point defects in the form of concentric spheres (index 0 and 3 critical points).  These are required because each hyperbolic region requires one such point, and the projection from $S^3$ to $\mathbb{R}^3$ only rids the system of two.  The same considerations of fibredness for edge dislocations and $+1/2$ disclinations apply to negative-charge disclinations as well: each of the $2|q|+2$ components is a handlebody only when the knot is fibred, leading to saturation of the lower bound.  When the knot is not fibred, all but one component can be made a handlebody (again, by thickening a Seifert surface and using the resulting handlebody for $2|q|+1$ components).  This can be fixed in an analogous way to edge dislocations, either by adding tunnels or by starting with free genus Seifert surfaces.

For the special case that $|q|=0$, there is no disclination in the system.  Instead, two Seifert surfaces glued along their boundary act as separatrices for two `hyperbolic' regions, leading to $4g$ (index 1 or 2) point defects.  This is precisely the topology of an edge dislocation: thickening one of the Seifert surfaces into interval replicates the construction from Section \ref{sec:edge_dislocations}.  Note that one can take any knotted $-|q|$ disclination and thicken the separatrices into an interval of surfaces, creating regions of both hyperbolic and radial surface structure along the defect.  Since this does not change the number of point defects singularities for the foliation, the results of Theorem \ref{thmnegdisc} are unaffected by the thickened, radial sections.  A knotted singular line with both radial and hyperbolic surface structure is known as an broken book decomposition\cite{ColinDehornoyRechtman}, a generalization of open book decompositions where the surfaces are purely radial.  We note that, consequently, our results apply to singular broken book decompositions of $S^3$: given a knot type and fixed local (radial and hyperbolic) profile, Theorem \ref{thmnegdisc} and the corresponding considerations for fibredness classify the required number of singular points, with charge $-|q|$ disclinations and dislocations as a special cases.  To our knowledge, this is the first explicit result of this type.  

\section{Discussion}

In this work, we have analyzed the topological properties of knotted defects in the smectic phase.  We started by describing how the radial defect structure of a screw dislocation can be accommodated on $S^3$ by knot fibrations.  To understand what happens for knots without fibration, we turned to edge dislocations, where the non-zero genus of smectic layers causes at least $4g$ point defects to be present alongside any knot, where $g$ is the minimum genus of a Seifert surface for the given knot.  These defects, which increase or decrease the genus of smectic layers, are the point defects of focal conic domains.  By showing that edge dislocations require further defects beyond the initial $4g$ when the knot does not have a fibration, we were able to reinterpret the Morse-Novikov points required to form screw dislocations in the shape of non-fibred knots as the same focal conic-like point defects.  By analyzing knotted $+1/2$ and negative-charge disclinations, we further classified the point defect structure induced by knotted smectic defects.  In particular, the $2g+1$ point defects required to knot a $+1/2$ disclination acts as a topological generalization to focal conic domains, furthering the connection between knotted smectic defects and focal conics.

These results have implications for the theory of smectic defects writ large.  The intricate topology brought about by layered structure of the smectic marks a sharp contrast to knotted defects in other ordered media, namely the cousin nematic phase.  For example, the requirement of point defects for knotted disclinations in the smectic phase can be compared directly to knotted disclinations in nematic liquid crystals, where the global structure of the director field may be topologically interesting\cite{Machon_Knots_Review} but smooth.  Further, we note that, while nematic $+1/2$ and $-1/2$ disclination lines are topologically equivalent (same element of $\pi_1(\mathbb{R}P^2)\cong\mathbb{Z}_2$), they are distinct in the smectic phase.  This leads to a different global topology when knotting smectic disclinations of charge $+1/2$ compared to $-1/2$, \emph{i.e.}, a different number (and type) of point defects.  The contrast in topological structure between knotted screw and edge dislocations serves as an striking implication of the topological distinction between the two defects\cite{Severino_Kamien}, but the mathematics of fibred knots underpins the study of both.  Further, the requirement of at least $4g$ point defects for knotted smectic edge dislocations implies that point defects must be created in order to manually rewire and knot an unknotted edge dislocation.  At large, the behaviour of knotted smectic defects, as outlined by this paper, demonstrates the topological complexities of the smectic phase not captured by the traditional homotopy theory of defects. 

By connecting the topology of knotted smectic defects to knot and Morse-Novikov theory, this paper motivates questions for mathematicians.  The following questions are, to our knowledge, open questions that have implications for the topology of knotted smectic defects.

\begin{itemize}
    \item \underline{Properties of Morse-Novikov Points}: We have outlined how screw dislocations are sensitive to the Morse-Novikov number of a knot, as the minimum number of point defects required (in $\mathbb{R}^3$) to knot a screw defect is $MN(K)+2$.  While the Morse-Novikov number is difficult to compute in general, what can be said about the topology of singular open book decompositions of $S^3$ with exactly $MN(K)$ point singularities?  How many topologically distinct configurations with $MN(K)$ singularities exist for a given knot?  Does a singular open book decomposition for a non-fibred knot exist with $MN(K)$ singular points and level sets of arbitrary genus (at least the genus of $K$)? 
    \item \underline{Genus and (Non-)Fibredness}: While screw dislocations are sensitive to the Morse-Novikov number, edge dislocations and disclinations are sensitive to both fibredness and the genus of the Seifert surfaces.  In order to minimize the total number of point defects, knotted edge dislocations and disclinations thus require not just the requisite Heegaard splitting but one utilizing the smallest possible genus.  Since the maximal genus of a singular fibration provides a Heegaard splitting of $S^3$, $g_{\text{max}}^{MN}\geq g_F$.  Is it always the case that a singular fibration with minimal $MN(K)$ point singularities exists such that $g_{\text{max}}^{MN}=g_F$?  Is there always a singular fibration with $g_{\text{min}}^{MN}=g$?  This is asked, for instance, in \cite{BakerHandleNumber} but is an open question -- Goda showed that there are minimal genus Seifert surfaces that are not minimal genus level sets of any singular fibration \cite{GodaHandleNumber}. However, there could be other Seifert surfaces of the same genus that are level sets of a singular fibration.  The answers to these questions have direct implications for the minimal number of point defects required for non-fibred edge dislocation and disclination knots.  For example, a non-fibred $+1/2$ disclination knot may minimize its number of point defects by using a free genus Seifert surface ($2g_F+1$) rather than using the topology of a singular fibration ($2g_{\text{max}}^{MN}+1$).

    \item \underline{Broken Book Decompositions}: While there are many examples of knotted structures in physical media in which open book decompositions (knot fibrations) are relevant, our work on knotted edge dislocations and charge $-|q|$ disclinations provides the first example of broken book decompositions in a physical context.  The results of Theorem \ref{thmnegdisc} on the minimal number of point singularities required to knot these smectic defects constitute special cases of singular broken book decompositions of $S^3$.  How do our results fit into the broader classification of (potentially singular) broken book decompositions of three-manifolds?
    
    \item \underline{Non-Orientability}: In order to apply basic knot and Morse theory, one needs to assume that all smectic layers are orientable.  Example \ref{example1} demonstrates how non-orientable surfaces may be used to beat the lower bound of point defects required for knotted $+1/2$ disclinations.  In this instance, we were still able to place a lower bound on the number of point defects required, even when the surface bounded by the $+1/2$ disclination is non-orientable.  Can the inclusion of non-orientable surfaces with knotted boundaries be generalized for foliations with both orientable and non-orientable surfaces?  Can notions from singular open and broken book decompositions be generalized to include generically non-orientable surfaces?  We have showed that the minimal number of point defects required to knot smectic defects depends on whether the knot is fibred.  How might the inclusion of non-orientable surfaces change these considerations regarding non-fibred knots?

    \item \underline{Changing Defect Type}: We give a classification of the topology of smectic defects that retain the same defect type everywhere along the knot -- \emph{i.e.}, either edge or screw, but not a combination of the two.  However, it is possible, for example, for a straight-line dislocation to change from edge to screw at an additional singular point\cite{Disclination_Pairs}.  How our results generalize to include defect structures that change along the knot and how these additional singular points interact with those required for defect knottedness is a subject of future work.

\end{itemize}

We conclude by remarking about the connections between knotted smectic defects and knot theory.  Not only do we demonstrate a deep connection between ideas in modern knot theory and the topological properties of smectic defects, but our results show that the defect structure of knotted smectic defects explicitly depends on knot fibrations and, for example, the Morse-Novikov number.  Liquid crystals, due to their optical properties, are well known for displaying nuanced ideas in topology \emph{visually} -- \emph{i.e.}, through images taken under the microscope.  Since the focal conic-like point defect structure associated with knotted smectic defects would be, in principle, observable, knotting smectic defects provides a physical playground for knot and Morse-Novikov theory.  In particular, by counting the number of point defects associated with a knotted smectic defect and comparing to the prediction from the known minimal genus, smectic liquid crystals potentially provide an experimental test for knot fibredness.  Note that, while our results provide lower bounds on the number of point defects present in the system, liquid crystalline configurations often admit additional defects beyond those required from topological considerations in order to provide elastic relief.  However, it is not unreasonable to assume that a minimally distorted texture will energetically prefer to minimize the number of defects.  In this light, we aspire to experimentally realize the Morse-Novikov number of a knot by knotting a smectic defect and observing the number of point defects that form.

\subsection*{Acknowledgements and Funding}
P.G.S. was supported by a National Science Foundation Graduate Research Fellowship.  R.D.K and P.G.S. were supported by a Simons Investigator grant to R.D.K. from the Simons Foundation.  B.B. was supported by the European Union’s Horizon 2020 research and innovation programme through the Marie Sklodowska-Curie grant agreement 101023017.


\bibliography{KnotRef}


\begin{thebibliography}{65}
\ifx \bisbn   \undefined \def \bisbn  #1{ISBN #1}\fi
\ifx \binits  \undefined \def \binits#1{#1}\fi
\ifx \bauthor  \undefined \def \bauthor#1{#1}\fi
\ifx \batitle  \undefined \def \batitle#1{#1}\fi
\ifx \bjtitle  \undefined \def \bjtitle#1{#1}\fi
\ifx \bvolume  \undefined \def \bvolume#1{\textbf{#1}}\fi
\ifx \byear  \undefined \def \byear#1{#1}\fi
\ifx \bissue  \undefined \def \bissue#1{#1}\fi
\ifx \bfpage  \undefined \def \bfpage#1{#1}\fi
\ifx \blpage  \undefined \def \blpage #1{#1}\fi
\ifx \burl  \undefined \def \burl#1{\textsf{#1}}\fi
\ifx \doiurl  \undefined \def \doiurl#1{\url{https://doi.org/#1}}\fi
\ifx \betal  \undefined \def \betal{\textit{et al.}}\fi
\ifx \binstitute  \undefined \def \binstitute#1{#1}\fi
\ifx \binstitutionaled  \undefined \def \binstitutionaled#1{#1}\fi
\ifx \bctitle  \undefined \def \bctitle#1{#1}\fi
\ifx \beditor  \undefined \def \beditor#1{#1}\fi
\ifx \bpublisher  \undefined \def \bpublisher#1{#1}\fi
\ifx \bbtitle  \undefined \def \bbtitle#1{#1}\fi
\ifx \bedition  \undefined \def \bedition#1{#1}\fi
\ifx \bseriesno  \undefined \def \bseriesno#1{#1}\fi
\ifx \blocation  \undefined \def \blocation#1{#1}\fi
\ifx \bsertitle  \undefined \def \bsertitle#1{#1}\fi
\ifx \bsnm \undefined \def \bsnm#1{#1}\fi
\ifx \bsuffix \undefined \def \bsuffix#1{#1}\fi
\ifx \bparticle \undefined \def \bparticle#1{#1}\fi
\ifx \barticle \undefined \def \barticle#1{#1}\fi
\bibcommenthead
\ifx \bconfdate \undefined \def \bconfdate #1{#1}\fi
\ifx \botherref \undefined \def \botherref #1{#1}\fi
\ifx \url \undefined \def \url#1{\textsf{#1}}\fi
\ifx \bchapter \undefined \def \bchapter#1{#1}\fi
\ifx \bbook \undefined \def \bbook#1{#1}\fi
\ifx \bcomment \undefined \def \bcomment#1{#1}\fi
\ifx \oauthor \undefined \def \oauthor#1{#1}\fi
\ifx \citeauthoryear \undefined \def \citeauthoryear#1{#1}\fi
\ifx \endbibitem  \undefined \def \endbibitem {}\fi
\ifx \bconflocation  \undefined \def \bconflocation#1{#1}\fi
\ifx \arxivurl  \undefined \def \arxivurl#1{\textsf{#1}}\fi
\csname PreBibitemsHook\endcsname

\bibitem[\protect\citeauthoryear{Frank and Read}{1950}]{Frank_Read}
\begin{barticle}
\bauthor{\bsnm{Frank}, \binits{F.C.}},
\bauthor{\bsnm{Read}, \binits{W.T.}}:
\batitle{Multiplication processes for slow moving dislocations}.
\bjtitle{Phys. Rev.}
\bvolume{79},
\bfpage{722}--\blpage{723}
(\byear{1950})
\doiurl{10.1103/PhysRev.79.722}
\end{barticle}
\endbibitem

\bibitem[\protect\citeauthoryear{Helfrich}{1978}]{Helfrich_Dislocation_Loop_Melting}
\begin{barticle}
\bauthor{\bsnm{Helfrich}, \binits{W.}}:
\batitle{Defect model of the smectic {A}-nematic phase transition}.
\bjtitle{J. Phys. France}
\bvolume{39},
\bfpage{1199}--\blpage{1208}
(\byear{1978})
\doiurl{10.1051/jphys:0197800390110119900}
\end{barticle}
\endbibitem

\bibitem[\protect\citeauthoryear{Nelson and Toner}{1981}]{Nelson_Toner_Smectic_Melting}
\begin{barticle}
\bauthor{\bsnm{Nelson}, \binits{D.R.}},
\bauthor{\bsnm{Toner}, \binits{J.}}:
\batitle{Bond-orientational order, dislocation loops, and melting of solids and smectic-{A} liquid crystals}.
\bjtitle{Phys. Rev. B}
\bvolume{24},
\bfpage{363}--\blpage{387}
(\byear{1981})
\doiurl{10.1103/PhysRevB.24.363}
\end{barticle}
\endbibitem

\bibitem[\protect\citeauthoryear{Moreau et~al.}{2005}]{Smectic_Melting_Experiment_Moreau_2006}
\begin{barticle}
\bauthor{\bsnm{Moreau}, \binits{P.}},
\bauthor{\bsnm{Navailles}, \binits{L.}},
\bauthor{\bsnm{Giermanska-Kahn}, \binits{J.}},
\bauthor{\bsnm{Mondain-Monval}, \binits{O.}},
\bauthor{\bsnm{Nallet}, \binits{F.}},
\bauthor{\bsnm{Roux}, \binits{D.}}:
\batitle{Dislocation-loop-mediated smectic melting}.
\bjtitle{Europhysics Letters}
\bvolume{73}(\bissue{1}),
\bfpage{49}
(\byear{2005})
\doiurl{10.1209/epl/i2005-10348-y}
\end{barticle}
\endbibitem

\bibitem[\protect\citeauthoryear{Tkalec et~al.}{2011}]{Copar_Colloids_Knotted_Defects}
\begin{barticle}
\bauthor{\bsnm{Tkalec}, \binits{U.}},
\bauthor{\bsnm{Ravnik}, \binits{M.}},
\bauthor{\bsnm{Čopar}, \binits{S.}},
\bauthor{\bsnm{Žumer}, \binits{S.}},
\bauthor{\bsnm{Muševič}, \binits{I.}}:
\batitle{Reconfigurable knots and links in chiral nematic colloids.}
\bjtitle{Science (New York, N.Y.)}
\bvolume{333}(\bissue{6038}),
\bfpage{62}--\blpage{5}
(\byear{2011})
\doiurl{10.1126/science.1205705}
\end{barticle}
\endbibitem

\bibitem[\protect\citeauthoryear{Machon}{2019}]{Machon_Knots_Review}
\begin{barticle}
\bauthor{\bsnm{Machon}, \binits{T.}}:
\batitle{The topology of knots and links in nematics}.
\bjtitle{Liquid Crystals Today}
\bvolume{28}(\bissue{3}),
\bfpage{58}--\blpage{67}
(\byear{2019})
\doiurl{10.1080/1358314X.2019.1681113}
\end{barticle}
\endbibitem

\bibitem[\protect\citeauthoryear{Proment et~al.}{2012}]{BEC_Knots}
\begin{barticle}
\bauthor{\bsnm{Proment}, \binits{D.}},
\bauthor{\bsnm{Onorato}, \binits{M.}},
\bauthor{\bsnm{Barenghi}, \binits{C.F.}}:
\batitle{Vortex knots in a bose-einstein condensate}.
\bjtitle{Phys. Rev. E}
\bvolume{85},
\bfpage{036306}
(\byear{2012})
\doiurl{10.1103/PhysRevE.85.036306}
\end{barticle}
\endbibitem

\bibitem[\protect\citeauthoryear{Kleckner et~al.}{2016}]{Kauffman_Superfluids}
\begin{barticle}
\bauthor{\bsnm{Kleckner}, \binits{D.}},
\bauthor{\bsnm{Kauffman}, \binits{L.}},
\bauthor{\bsnm{Irvine}, \binits{W.}}:
\batitle{How superfluid vortex knots untie}.
\bjtitle{Nature Phys}
\bvolume{12},
\bfpage{650}--\blpage{655}
(\byear{2016})
\end{barticle}
\endbibitem

\bibitem[\protect\citeauthoryear{Sutcliffe}{2007}]{Skyrme_Knots}
\begin{barticle}
\bauthor{\bsnm{Sutcliffe}, \binits{P.}}:
\batitle{Knots in the {S}kyrme-{F}addeev model}.
\bjtitle{Proc. R. Soc.}
\bvolume{463},
\bfpage{3001}--\blpage{3020}
(\byear{2007})
\doiurl{10.1098/rspa.2007.0038}
\end{barticle}
\endbibitem

\bibitem[\protect\citeauthoryear{Dennis et~al.}{2010}]{Optical_Knots}
\begin{barticle}
\bauthor{\bsnm{Dennis}, \binits{M.}},
\bauthor{\bsnm{King}, \binits{R.}},
\bauthor{\bsnm{Jack}, \binits{B.}},
\bauthor{\bsnm{O'Holleran}, \binits{K.}},
\bauthor{\bsnm{Padgett}, \binits{M.J.}}:
\batitle{Isolated optical vortex knots}.
\bjtitle{Nature Phys}
\bvolume{6},
\bfpage{118}--\blpage{121}
(\byear{2010})
\doiurl{10.1038/nphys1504}
\end{barticle}
\endbibitem

\bibitem[\protect\citeauthoryear{Moffatt}{1969}]{Moffatt_Fluid_Knots}
\begin{barticle}
\bauthor{\bsnm{Moffatt}, \binits{H.}}:
\batitle{The degree of knottedness of tangled vortex lines}.
\bjtitle{Journal of Fluid Mechanics}
\bvolume{35},
\bfpage{117}--\blpage{127}
(\byear{1969})
\doiurl{10.1017/S0022112069000991}
\end{barticle}
\endbibitem

\bibitem[\protect\citeauthoryear{Kleckner and Irvine}{2013}]{Irvine_Fluid_Knots}
\begin{barticle}
\bauthor{\bsnm{Kleckner}, \binits{D.}},
\bauthor{\bsnm{Irvine}, \binits{W.}}:
\batitle{Creation and dynamics of knotted vortices}.
\bjtitle{Nature Phys}
\bvolume{9},
\bfpage{253}--\blpage{258}
(\byear{2013})
\doiurl{10.1038/nphys2560}
\end{barticle}
\endbibitem

\bibitem[\protect\citeauthoryear{Čopar et~al.}{2012}]{copar_stability_nematic_colloids}
\begin{barticle}
\bauthor{\bsnm{Čopar}, \binits{S.}},
\bauthor{\bsnm{Porenta}, \binits{T.}},
\bauthor{\bsnm{Jampani}, \binits{V.}},
\bauthor{\bsnm{Muševič}, \binits{I.}},
\bauthor{\bsnm{Žumer}, \binits{S.}}:
\batitle{Stability and rewiring of nematic braids in chiral nematic colloids}.
\bjtitle{Soft Matter}
\bvolume{8}(\bissue{33}),
\bfpage{8595}--\blpage{8600}
(\byear{2012})
\end{barticle}
\endbibitem

\bibitem[\protect\citeauthoryear{Machon and Alexander}{2013}]{Machin_Nonorientable_Boundary_Conditions}
\begin{barticle}
\bauthor{\bsnm{Machon}, \binits{T.}},
\bauthor{\bsnm{Alexander}, \binits{G.P.}}:
\batitle{Knots and nonorientable surfaces in chiral nematics}.
\bjtitle{Proceedings of the National Academy of Sciences}
\bvolume{110}(\bissue{35}),
\bfpage{14174}--\blpage{14179}
(\byear{2013})
\end{barticle}
\endbibitem

\bibitem[\protect\citeauthoryear{Rolfsen}{1976}]{Rolfsen_Knots_And_Links}
\begin{bbook}
\bauthor{\bsnm{Rolfsen}, \binits{D.}}:
\bbtitle{Knots and Links}.
\bpublisher{Publish or Perish},
\blocation{Berkeley, CA}
(\byear{1976})
\end{bbook}
\endbibitem

\bibitem[\protect\citeauthoryear{Machon and Alexander}{2014}]{Machon_Alexander_Knots_1}
\begin{barticle}
\bauthor{\bsnm{Machon}, \binits{T.}},
\bauthor{\bsnm{Alexander}, \binits{G.P.}}:
\batitle{Knotted defects in nematic liquid crystals}.
\bjtitle{Phys. Rev. Lett.}
\bvolume{113},
\bfpage{027801}
(\byear{2014})
\doiurl{10.1103/PhysRevLett.113.027801}
\end{barticle}
\endbibitem

\bibitem[\protect\citeauthoryear{Machon and Alexander}{2016}]{Machon_Alexander_Knots_2}
\begin{botherref}
\oauthor{\bsnm{Machon}, \binits{T.}},
\oauthor{\bsnm{Alexander}, \binits{G.P.}}:
Global defect topology in nematic liquid crystals.
Proc. R. Soc. A
\textbf{472}
(2016)
\doiurl{10.1098/rspa.2016.0265}
\end{botherref}
\endbibitem

\bibitem[\protect\citeauthoryear{Mermin}{1979}]{Mermin_RevMod_Defects}
\begin{barticle}
\bauthor{\bsnm{Mermin}, \binits{N.D.}}:
\batitle{The topological theory of defects in ordered media}.
\bjtitle{Rev. Mod. Phys.}
\bvolume{51},
\bfpage{591}--\blpage{648}
(\byear{1979})
\doiurl{10.1103/RevModPhys.51.591}
\end{barticle}
\endbibitem

\bibitem[\protect\citeauthoryear{Po\'{e}naru}{1981}]{poenaru}
\begin{barticle}
\bauthor{\bsnm{Po\'{e}naru}, \binits{V.}}:
\batitle{Some aspects of the theory of defects of ordered media and gauge fields related to foliations}.
\bjtitle{Commun.Math. Phys.}
\bvolume{80},
\bfpage{127}--\blpage{136}
(\byear{1981})
\doiurl{10.1007/BF01213598}
\end{barticle}
\endbibitem

\bibitem[\protect\citeauthoryear{Chen et~al.}{2009}]{chen_goldstone}
\begin{barticle}
\bauthor{\bsnm{Chen}, \binits{B.G.}},
\bauthor{\bsnm{Alexander}, \binits{G.P.}},
\bauthor{\bsnm{Kamien}, \binits{R.D.}}:
\batitle{Symmetry breaking in smectics and surface models of their singularities}.
\bjtitle{PNAS}
\bvolume{106},
\bfpage{15577}--\blpage{15582}
(\byear{2009})
\doiurl{10.1073/pnas.0905242106}
\end{barticle}
\endbibitem

\bibitem[\protect\citeauthoryear{Kamien and Mosna}{2016}]{Mosna_Kamien_Linked_Dislocations}
\begin{barticle}
\bauthor{\bsnm{Kamien}, \binits{R.D.}},
\bauthor{\bsnm{Mosna}, \binits{R.A.}}:
\batitle{The topology of dislocations in smectic liquid crystals}.
\bjtitle{New Journal of Physics}
\bvolume{18}(\bissue{5}),
\bfpage{053012}
(\byear{2016})
\doiurl{10.1088/1367-2630/18/5/053012}
\end{barticle}
\endbibitem

\bibitem[\protect\citeauthoryear{Hocking et~al.}{2022}]{Hocking_Kamien_Peierls_Nabarro}
\begin{botherref}
\oauthor{\bsnm{Hocking}, \binits{B.J.}},
\oauthor{\bsnm{Ansell}, \binits{H.S.}},
\oauthor{\bsnm{Kamien}, \binits{R.D.}},
\oauthor{\bsnm{Thomas}, \binits{M.}}:
The topological origin of the {P}eierls-{N}abarro barrier.
Proc. R. Soc. A.
\textbf{478}
(2022)
\doiurl{10.1098/rspa.2021.0725}
\end{botherref}
\endbibitem

\bibitem[\protect\citeauthoryear{Volterra}{1907}]{Volterra1907}
\begin{barticle}
\bauthor{\bsnm{Volterra}, \binits{V.}}:
\batitle{Sur l'équilibre des corps élastiques multiplement connexes}.
\bjtitle{Annales scientifiques de l'École Normale Supérieure}
\bvolume{24},
\bfpage{401}--\blpage{517}
(\byear{1907})
\end{barticle}
\endbibitem

\bibitem[\protect\citeauthoryear{Chaikin and Lubensky}{1995}]{chaikin_and_lubensky}
\begin{bbook}
\bauthor{\bsnm{Chaikin}, \binits{P.M.}},
\bauthor{\bsnm{Lubensky}, \binits{T.C.}}:
\bbtitle{Principles of Condensed Matter Physics}.
\bpublisher{Cambridge University Press},
\blocation{Cambridge}
(\byear{1995})
\end{bbook}
\endbibitem

\bibitem[\protect\citeauthoryear{Severino and Kamien}{2024}]{Severino_Kamien}
\begin{barticle}
\bauthor{\bsnm{Severino}, \binits{P.G.}},
\bauthor{\bsnm{Kamien}, \binits{R.D.}}:
\batitle{Escape from the second dimension: A topological distinction between edge and screw dislocations}.
\bjtitle{Phys. Rev. E}
\bvolume{109},
\bfpage{012701}
(\byear{2024})
\doiurl{10.1103/PhysRevE.109.L012701}
\end{barticle}
\endbibitem

\bibitem[\protect\citeauthoryear{Meyer}{1973}]{escapeinto}
\begin{barticle}
\bauthor{\bsnm{Meyer}, \binits{R.B.}}:
\batitle{On the existence of even indexed disclinations in nematic liquid crystals}.
\bjtitle{The Philosophical Magazine: A Journal of Theoretical Experimental and Applied Physics}
\bvolume{27}(\bissue{2}),
\bfpage{405}--\blpage{424}
(\byear{1973})
\doiurl{10.1080/14786437308227417}
\end{barticle}
\endbibitem

\bibitem[\protect\citeauthoryear{Camacho and Line~Neto}{1985}]{FoliationsCamacho}
\begin{bbook}
\bauthor{\bsnm{Camacho}, \binits{C.}},
\bauthor{\bsnm{Line~Neto}, \binits{A.}}:
\bbtitle{Geometric Theory of Foliations}.
\bpublisher{Birk\"{a}user},
\blocation{Bsoton, MA}
(\byear{1985})
\end{bbook}
\endbibitem

\bibitem[\protect\citeauthoryear{Hector and Hirsch}{1986}]{FoliationsGilbert}
\begin{bbook}
\bauthor{\bsnm{Hector}, \binits{G.}},
\bauthor{\bsnm{Hirsch}, \binits{U.}}:
\bbtitle{Introduction to the Geometry of Foliations, Part A}.
\bpublisher{Friedr. Vieweg \& Sohn},
\blocation{Braunschweig , Germany}
(\byear{1986})
\end{bbook}
\endbibitem

\bibitem[\protect\citeauthoryear{Meyer et~al.}{2010}]{screwone}
\begin{barticle}
\bauthor{\bsnm{Meyer}, \binits{C.}},
\bauthor{\bsnm{Nastishin}, \binits{Y.}},
\bauthor{\bsnm{Kleman}, \binits{M.}}:
\batitle{Helical defects in smectic-${A}$ and smectic-${A}^{\ensuremath{\ast}}$ phases}.
\bjtitle{Phys. Rev. E}
\bvolume{82},
\bfpage{031704}
(\byear{2010})
\doiurl{10.1103/PhysRevE.82.031704}
\end{barticle}
\endbibitem

\bibitem[\protect\citeauthoryear{Achard et~al.}{2005}]{screwtwo}
\begin{barticle}
\bauthor{\bsnm{Achard}, \binits{M.-F.}},
\bauthor{\bsnm{Kleman}, \binits{M.}},
\bauthor{\bsnm{Nastishin}, \binits{Y.A.}},
\bauthor{\bsnm{Nguyen}, \binits{H.-T.}}:
\batitle{Liquid crystal helical ribbons as isometric textures}.
\bjtitle{The European Physical Journal E}
\bvolume{16}(\bissue{1}),
\bfpage{37}--\blpage{47}
(\byear{2005})
\doiurl{10.1140/epje/e2005-00005-2}
\end{barticle}
\endbibitem

\bibitem[\protect\citeauthoryear{Aharoni et~al.}{2017}]{Disclination_Pairs}
\begin{barticle}
\bauthor{\bsnm{Aharoni}, \binits{H.}},
\bauthor{\bsnm{Machon}, \binits{T.}},
\bauthor{\bsnm{Kamien}, \binits{R.D.}}:
\batitle{Composite dislocations in smectic liquid crystals}.
\bjtitle{Phys. Rev. Lett.}
\bvolume{118},
\bfpage{257801}
(\byear{2017})
\doiurl{10.1103/PhysRevLett.118.257801}
\end{barticle}
\endbibitem

\bibitem[\protect\citeauthoryear{Kamien}{2002}]{RevModPhys_Geometry}
\begin{barticle}
\bauthor{\bsnm{Kamien}, \binits{R.D.}}:
\batitle{The geometry of soft materials: a primer}.
\bjtitle{Rev. Mod. Phys.}
\bvolume{74},
\bfpage{953}--\blpage{971}
(\byear{2002})
\doiurl{10.1103/RevModPhys.74.953}
\end{barticle}
\endbibitem

\bibitem[\protect\citeauthoryear{Kleman and Lavrentovich}{2009}]{FocalConicsKlemanLavrentovich}
\begin{barticle}
\bauthor{\bsnm{Kleman}, \binits{M.}},
\bauthor{\bsnm{Lavrentovich}, \binits{O.D.}}:
\batitle{Liquids with conics}.
\bjtitle{Liquid Crystals}
\bvolume{36},
\bfpage{1085}--\blpage{1099}
(\byear{2009})
\doiurl{10.1080/02678290902814718}
\end{barticle}
\endbibitem

\bibitem[\protect\citeauthoryear{Dupin}{1822}]{Dupin}
\begin{botherref}
\oauthor{\bsnm{Dupin}, \binits{C.}}:
Applications de g{\'e}om{\'e}trie et de m{\'e}chanique, a la marine, aux ponts et chauss{\'e}es, etc.
Bachelier
(1822)
\end{botherref}
\endbibitem

\bibitem[\protect\citeauthoryear{Grandjean and Friedel}{1910}]{GrandjeanFriedel}
\begin{barticle}
\bauthor{\bsnm{Grandjean}, \binits{F.}},
\bauthor{\bsnm{Friedel}, \binits{G.}}:
\batitle{Observations g\'eom\'etriques sur les liquides \`a coniques focales.}
\bjtitle{Bulletin de Min\'eralogie}
\bvolume{33},
\bfpage{409}--\blpage{465}
(\byear{1910})
\end{barticle}
\endbibitem

\bibitem[\protect\citeauthoryear{Sethna and Kl\'eman}{1982}]{SethnaKleman}
\begin{barticle}
\bauthor{\bsnm{Sethna}, \binits{J.P.}},
\bauthor{\bsnm{Kl\'eman}, \binits{M.}}:
\batitle{Spheric domains in smectic liquid crystals}.
\bjtitle{Phys. Rev. A}
\bvolume{26},
\bfpage{3037}--\blpage{3040}
(\byear{1982})
\doiurl{10.1103/PhysRevA.26.3037}
\end{barticle}
\endbibitem

\bibitem[\protect\citeauthoryear{de~Dieu~Niyonzima et~al.}{2024}]{Lacaze_Kamien_OilyStreaks}
\begin{botherref}
\oauthor{\bsnm{Dieu~Niyonzima}, \binits{J.}},
\oauthor{\bsnm{Jeridi}, \binits{H.}},
\oauthor{\bsnm{Essaoui}, \binits{L.}},
\oauthor{\bsnm{Tosarelli}, \binits{C.}},
\oauthor{\bsnm{Vlad}, \binits{A.}},
\oauthor{\bsnm{Coati}, \binits{A.}},
\oauthor{\bsnm{Royer}, \binits{S.}},
\oauthor{\bsnm{Trimaille}, \binits{I.}},
\oauthor{\bsnm{Goldmann}, \binits{M.}},
\oauthor{\bsnm{Gallas}, \binits{B.}},
\oauthor{\bsnm{Constantin}, \binits{D.}},
\oauthor{\bsnm{Babonneau}, \binits{D.}},
\oauthor{\bsnm{Garreau}, \binits{Y.}},
\oauthor{\bsnm{Croset}, \binits{B.}},
\oauthor{\bsnm{Kralj}, \binits{S.}},
\oauthor{\bsnm{Kamien}, \binits{R.D.}},
\oauthor{\bsnm{Lacaze}, \binits{E.}}:
X-ray diffraction reveals the consequences of strong deformation in thin smectic films: dilation and chevron formation
(2024).
\url{https://arxiv.org/abs/2407.10598}
\end{botherref}
\endbibitem

\bibitem[\protect\citeauthoryear{Alexander et~al.}{2012}]{Kamien_RevModPhys_Defects}
\begin{barticle}
\bauthor{\bsnm{Alexander}, \binits{G.P.}},
\bauthor{\bsnm{Chen}, \binits{B.G.-g.}},
\bauthor{\bsnm{Matsumoto}, \binits{E.A.}},
\bauthor{\bsnm{Kamien}, \binits{R.D.}}:
\batitle{Colloquium: Disclination loops, point defects, and all that in nematic liquid crystals}.
\bjtitle{Rev. Mod. Phys.}
\bvolume{84},
\bfpage{497}--\blpage{514}
(\byear{2012})
\doiurl{10.1103/RevModPhys.84.497}
\end{barticle}
\endbibitem

\bibitem[\protect\citeauthoryear{Machon et~al.}{2019}]{aspects_topology_smectics}
\begin{barticle}
\bauthor{\bsnm{Machon}, \binits{T.}},
\bauthor{\bsnm{Aharoni}, \binits{H.}},
\bauthor{\bsnm{Hu}, \binits{Y.}},
\bauthor{\bsnm{Kamien}, \binits{R.D.}}:
\batitle{Aspects of defect topology in smectic liquid}.
\bjtitle{Commun.Math. Phys.}
\bvolume{372},
\bfpage{525}--\blpage{542}
(\byear{2019})
\doiurl{10.1007/s00220-019-03366-y}
\end{barticle}
\endbibitem

\bibitem[\protect\citeauthoryear{Bode et~al.}{2017}]{BodeLemniscateFibrations}
\begin{barticle}
\bauthor{\bsnm{Bode}, \binits{B.}},
\bauthor{\bsnm{Dennis}, \binits{R.M.}},
\bauthor{\bsnm{Foster}, \binits{D.}},
\bauthor{\bsnm{King}, \binits{R.P.}}:
\batitle{Knotted fields and explicit fibrations for lemniscate knots}.
\bjtitle{Proceedings of the Royal Society A: Mathematical, Physical and Engineering Sciences}
\bvolume{473},
\bfpage{20160829}
(\byear{2017})
\doiurl{10.1098/rspa.2016.0829}
\end{barticle}
\endbibitem

\bibitem[\protect\citeauthoryear{Perron}{1982}]{Perron}
\begin{barticle}
\bauthor{\bsnm{Perron}, \binits{B.}}:
\batitle{Le n{\oe}ud ``huit'' est alg{\'e}brique r{\'e}el}.
\bjtitle{Inv. Math.}
\bvolume{65},
\bfpage{441}--\blpage{451}
(\byear{1982})
\end{barticle}
\endbibitem

\bibitem[\protect\citeauthoryear{Neuwirth}{1957}]{Neuwirth_Thesis}
\begin{botherref}
\oauthor{\bsnm{Neuwirth}, \binits{L.P.}}:
Knot groups.
Princeton University Doctoral Thesis
(1957)
\end{botherref}
\endbibitem

\bibitem[\protect\citeauthoryear{Stallings}{1962}]{Stallings_Comm_Subgroup}
\begin{botherref}
\oauthor{\bsnm{Stallings}, \binits{J.}}:
On fibering certain 3-manifolds.
Topology of 3-Manifolds and Related Topics. Prentice- Hall, New Jersey
(1962)
\end{botherref}
\endbibitem

\bibitem[\protect\citeauthoryear{Ghiggini}{2008}]{Ghiggini}
\begin{barticle}
\bauthor{\bsnm{Ghiggini}, \binits{P.}}:
\batitle{Knot {F}loer homology detects genus-one fibred knots}.
\bjtitle{American Journal of Mathematics}
\bvolume{130},
\bfpage{1151}--\blpage{1169}
(\byear{2008})
\end{barticle}
\endbibitem

\bibitem[\protect\citeauthoryear{Ni}{2007}]{Ni}
\begin{barticle}
\bauthor{\bsnm{Ni}, \binits{Y.}}:
\batitle{Knot {F}loer homology detects fibred knots}.
\bjtitle{Inv. Math.}
\bvolume{170},
\bfpage{577}--\blpage{608}
(\byear{2007})
\end{barticle}
\endbibitem

\bibitem[\protect\citeauthoryear{Meilhan}{2021}]{Meilhan}
\begin{bchapter}
\bauthor{\bsnm{Meilhan}, \binits{J.}}:
\bctitle{Linking number and {M}ilnor invariants}.
In: \beditor{\bsnm{Adams}, \binits{C.}},
\beditor{\bsnm{Flapan}, \binits{E.}},
\beditor{\bsnm{Henrich}, \binits{A.}},
\beditor{\bsnm{Kauffman}, \binits{L.H.}},
\beditor{\bsnm{Ludwig}, \binits{L.D.}},
\beditor{\bsnm{Nelson}, \binits{S.}} (eds.)
\bbtitle{Encyclopedia of Knot Theory},
pp. \bfpage{817}--\blpage{830}.
\bpublisher{Chapman and Hall/CRC},
\blocation{New York}
(\byear{2021})
\end{bchapter}
\endbibitem

\bibitem[\protect\citeauthoryear{Stallings}{1978}]{Stallings_Fibred_Knots_Links}
\begin{botherref}
\oauthor{\bsnm{Stallings}, \binits{J.}}:
Construction of fibred knots and links.
Proceedings of Symposia in Pure Mathematics
\textbf{32}
(1978)
\end{botherref}
\endbibitem

\bibitem[\protect\citeauthoryear{Hirasawa and Rudolph}{2003}]{hirasawa2003constructions}
\begin{botherref}
\oauthor{\bsnm{Hirasawa}, \binits{M.}},
\oauthor{\bsnm{Rudolph}, \binits{L.}}:
Constructions of {M}orse maps for knots and links, and upper bounds on the {M}orse-{N}ovikov number
(2003)
\end{botherref}
\endbibitem

\bibitem[\protect\citeauthoryear{Weber et~al.}{2002}]{weber}
\begin{barticle}
\bauthor{\bsnm{Weber}, \binits{C.}},
\bauthor{\bsnm{Pajitnov}, \binits{A.}},
\bauthor{\bsnm{Rudolph}, \binits{L.}}:
\batitle{Morse-{N}ovikov number for knots and links}.
\bjtitle{St. Petersburg Math. J.}
\bvolume{13},
\bfpage{417}--\blpage{426}
(\byear{2002})
\end{barticle}
\endbibitem

\bibitem[\protect\citeauthoryear{Pajitnov}{2010}]{pajitnov}
\begin{barticle}
\bauthor{\bsnm{Pajitnov}, \binits{A.}}:
\batitle{On the tunnel number and the {M}orse-{N}ovikov number of knots}.
\bjtitle{Algebraic and Geometric Topology}
\bvolume{10},
\bfpage{627}--\blpage{635}
(\byear{2010})
\end{barticle}
\endbibitem

\bibitem[\protect\citeauthoryear{Bode and Hirasawa}{2023}]{BodeMikami}
\begin{barticle}
\bauthor{\bsnm{Bode}, \binits{B.}},
\bauthor{\bsnm{Hirasawa}, \binits{M.}}:
\batitle{Saddle point braids of braided fibrations and pseudo-fibrations}.
\bjtitle{arXiv preprint}
(\byear{2023})
\doiurl{10.48550/arXiv.2309.04972}
\end{barticle}
\endbibitem

\bibitem[\protect\citeauthoryear{Milnor}{1963}]{milnor}
\begin{bbook}
\bauthor{\bsnm{Milnor}, \binits{J.}}:
\bbtitle{Morse Theory}.
\bpublisher{Princeton University Press},
\blocation{Princeton}
(\byear{1963})
\end{bbook}
\endbibitem

\bibitem[\protect\citeauthoryear{Pajitnov}{2006}]{pajitnovbook}
\begin{bbook}
\bauthor{\bsnm{Pajitnov}, \binits{A.}}:
\bbtitle{Circle-valued {M}orse Theory}.
\bpublisher{De Gruyter},
\blocation{Berlin, New York}
(\byear{2006})
\end{bbook}
\endbibitem

\bibitem[\protect\citeauthoryear{Borodzik et~al.}{2016}]{Morseboundary}
\begin{barticle}
\bauthor{\bsnm{Borodzik}, \binits{M.}},
\bauthor{\bsnm{Némethi}, \binits{A.}},
\bauthor{\bsnm{Ranicki}, \binits{A.}}:
\batitle{{M}orse theory for manifolds with boundary}.
\bjtitle{Algebraic Geom. Topol.}
\bvolume{16},
\bfpage{971}--\blpage{1023}
(\byear{2016})
\doiurl{10.2140/agt.2016.16.971}
\end{barticle}
\endbibitem

\bibitem[\protect\citeauthoryear{Licata and Vértesi}{2020}]{vera1}
\begin{barticle}
\bauthor{\bsnm{Licata}, \binits{J.E.}},
\bauthor{\bsnm{Vértesi}, \binits{V.}}:
\batitle{Foliated open books}.
\bjtitle{arXiv preprint}
(\byear{2020})
\doiurl{10.48550/arXiv.2002.01752}
\end{barticle}
\endbibitem

\bibitem[\protect\citeauthoryear{Vértesi and Licata}{2022}]{vera2}
\begin{bchapter}
\bauthor{\bsnm{Vértesi}, \binits{V.}},
\bauthor{\bsnm{Licata}, \binits{J.E.}}:
\bctitle{{M}orse foliated open books and right-veering monodromies}.
In: \beditor{\bsnm{Baldwin}, \binits{J.A.}},
\beditor{\bsnm{Boden}, \binits{H.U.}},
\beditor{\bsnm{Etnyre}, \binits{J.B.}},
\beditor{\bsnm{Watson}, \binits{L.}} (eds.)
\bbtitle{Gauge Theory and Low-Dimensional Topology: Progress and Interaction}.
\bsertitle{The Open Book Series},
vol. \bseriesno{5},
pp. \bfpage{309}--\blpage{324}.
\bpublisher{Mathematical Sciences Publishers},
\blocation{Berkeley}
(\byear{2022})
\end{bchapter}
\endbibitem

\bibitem[\protect\citeauthoryear{Rosati}{2012}]{Morsefoliation}
\begin{barticle}
\bauthor{\bsnm{Rosati}, \binits{L.}}:
\batitle{On smooth foliations with {M}orse singularities}.
\bjtitle{Topol. Appl.}
\bvolume{159},
\bfpage{1388}--\blpage{1403}
(\byear{2012})
\doiurl{10.1016/j.topol.2011.12.020}
\end{barticle}
\endbibitem

\bibitem[\protect\citeauthoryear{Budney}{2007}]{BudneyLongKnots}
\begin{barticle}
\bauthor{\bsnm{Budney}, \binits{R.}}:
\batitle{Little cubes and long knots}.
\bjtitle{Topology}
\bvolume{46},
\bfpage{1}--\blpage{27}
(\byear{2007})
\end{barticle}
\endbibitem

\bibitem[\protect\citeauthoryear{Vassiliev}{1990}]{VassilievCohomologyKnotSpaces}
\begin{bchapter}
\bauthor{\bsnm{Vassiliev}, \binits{V.A.}}:
\bctitle{Cohomology of knot spaces}.
In: \beditor{\bsnm{Arnold}, \binits{V.I.}} (ed.)
\bbtitle{Theory of Singularities and Its Applications},
pp. \bfpage{23}--\blpage{70}.
\bpublisher{American Mathematical Society},
\blocation{Providence}
(\byear{1990})
\end{bchapter}
\endbibitem

\bibitem[\protect\citeauthoryear{Goda and Pajitnov}{2005}]{godapajitnov}
\begin{barticle}
\bauthor{\bsnm{Goda}, \binits{H.}},
\bauthor{\bsnm{Pajitnov}, \binits{A.}}:
\batitle{Twisted {N}ovikov homology and circle-valued {M}orse theory for knots and links}.
\bjtitle{Osaka J. Math.}
\bvolume{42},
\bfpage{557}--\blpage{572}
(\byear{2005})
\doiurl{10.18910/9441}
\end{barticle}
\endbibitem

\bibitem[\protect\citeauthoryear{Murao}{2021}]{murao}
\begin{barticle}
\bauthor{\bsnm{Murao}, \binits{T.}}:
\batitle{The tunnel number and the cutting number with constituent handlebody-knots}.
\bjtitle{Topol. Appl.}
\bvolume{292},
\bfpage{107632}
(\byear{2021})
\doiurl{10.1016/j.topol.2021.107632}
\end{barticle}
\endbibitem

\bibitem[\protect\citeauthoryear{Kindred}{2021}]{kindred}
\begin{bchapter}
\bauthor{\bsnm{Kindred}, \binits{T.}}:
\bctitle{Non-orientable spanning surfaces for knots}.
In: \beditor{\bsnm{Adams}, \binits{C.}},
\beditor{\bsnm{Flapan}, \binits{E.}},
\beditor{\bsnm{Henrich}, \binits{A.}},
\beditor{\bsnm{Kauffman}, \binits{L.H.}},
\beditor{\bsnm{Ludwig}, \binits{L.D.}},
\beditor{\bsnm{Nelson}, \binits{S.}} (eds.)
\bbtitle{Encyclopedia of Knot Theory},
pp. \bfpage{197}--\blpage{203}.
\bpublisher{Chapman and Hall/CRC},
\blocation{New York}
(\byear{2021})
\end{bchapter}
\endbibitem

\bibitem[\protect\citeauthoryear{Colin et~al.}{2023}]{ColinDehornoyRechtman}
\begin{barticle}
\bauthor{\bsnm{Colin}, \binits{V.}},
\bauthor{\bsnm{Dehornoy}, \binits{P.}},
\bauthor{\bsnm{Rechtman}, \binits{A.}}:
\batitle{On the existence of supporting broken book decompositions for contact forms in dimension 3}.
\bjtitle{Inv. Math.}
\bvolume{231},
\bfpage{1489}--\blpage{1539}
(\byear{2023})
\end{barticle}
\endbibitem

\bibitem[\protect\citeauthoryear{Baker}{2021}]{BakerHandleNumber}
\begin{barticle}
\bauthor{\bsnm{Baker}, \binits{K.L.}}:
\batitle{The {M}orse–{N}ovikov number of knots under connected sum and cabling}.
\bjtitle{Journal of Topology}
\bvolume{14},
\bfpage{1351}--\blpage{1368}
(\byear{2021})
\doiurl{10.1112/topo.12210}
\end{barticle}
\endbibitem

\bibitem[\protect\citeauthoryear{Goda}{1993}]{GodaHandleNumber}
\begin{barticle}
\bauthor{\bsnm{Goda}, \binits{H.}}:
\batitle{On handle number of {S}eifert surfaces in $\text{S}^3$}.
\bjtitle{Osaka Journal of Mathematics}
\bvolume{30},
\bfpage{1351}--\blpage{1368}
(\byear{1993})
\end{barticle}
\endbibitem

\end{thebibliography}

\end{document}